\documentclass[final,a4paper,11pt]{book}

\pdfoutput=1

\usepackage[a4paper, left=3cm, right=3cm, top=3cm, bottom=3cm]{geometry}
\usepackage[T1]{fontenc}
\usepackage[utf8]{inputenc}
\usepackage{lmodern}
\usepackage{tocloft,sectsty}

\partfont{\fontfamily{lmss}\selectfont}
\chapterfont{\fontfamily{lmss}\selectfont}

\usepackage{xskak}
\usepackage{chessboard}

\usepackage{fancyhdr}
\usepackage{url}

\usepackage[%
    style=numeric-comp,sorting=nyt,
    sortcites=true,doi=false,url=false,isbn=false,
    giveninits=true,hyperref,backend=bibtex,bibencoding=ascii,maxbibnames=99]{biblatex}
\AtBeginBibliography{\small} 

\title{\bf Application of Comparator Networks in Encoding of Cardinality Constraints}
\author{Micha\l~Karpi\'nski}
\date{\today}

\usepackage{amsmath,amsfonts,amssymb}
\usepackage{amsthm}
\usepackage{booktabs}
\usepackage{centernot}
\usepackage{graphicx}
\usepackage{wrapfig}
\usepackage[caption=false]{subfig}
\usepackage{multirow}
\usepackage{tikz-cd}
\usepackage{tikzscale}
\usepackage{tikzpagenodes}
\usepackage{mathptmx}

\usepackage[chapter]{algorithm}
\usepackage[noend]{algpseudocode}
\floatname{algorithm}{Algorithm}

\algrenewcommand{\algorithmiccomment}[1]{\small \hfill\# #1}

\makeatletter
\newcounter{algorithmicH}
\let\oldalgorithmic\algorithmic
\renewcommand{\algorithmic}{%
  \stepcounter{algorithmicH}
  \oldalgorithmic}
\renewcommand{\theHALG@line}{ALG@line.\thealgorithmicH.\arabic{ALG@line}}
\makeatother

\usepackage{letltxmacro}
\LetLtxMacro{\originaleqref}{\eqref}
\renewcommand{\eqref}{Eq.~\originaleqref}

\newtheorem{theorem}{Theorem}
\newtheorem{lemma}{Lemma}
\newtheorem{observation}{Observation}
\newtheorem{corollary}{Corollary}
\theoremstyle{definition}
\newtheorem{definition}{Definition}
\newtheorem{example}{Example}

\numberwithin{definition}{chapter}
\numberwithin{theorem}{chapter}
\numberwithin{lemma}{chapter}
\numberwithin{example}{chapter}
\numberwithin{observation}{chapter}
\numberwithin{corollary}{chapter}

\newcommand{\tuple}[1]{\langle #1 \rangle}
\newcommand{\floor}[1]{\left \lfloor #1 \right \rfloor }
\newcommand{\zip}{\textrm{zip}}
\newcommand{\odd}{\textit{odd}}
\newcommand{\even}{\textit{even}}
\newcommand{\leftt}{\textit{left}}
\newcommand{\rightt}{\textit{right}}
\newcommand{\ceil}[1]{\left \lceil #1 \right \rceil }
\newcommand{\nat}{\mathbb{N}}
\newcommand{\bool}{\{0,1\}}

\newcommand{\pref}{\textrm{pref}}
\newcommand{\suff}{\textrm{suff}}
\newcommand{\remove}{\textrm{drop}}
\newcommand{\impl}{\Rightarrow}
\newcommand{\sort}{\textit{sort}}
\newcommand{\snd}{\textrm{2nd}}
\newcommand{\trd}{\textrm{3rd}}
\renewcommand{\iff}{\Leftrightarrow}
\newcommand{\Iff}{\Longleftrightarrow}

\usepackage[skins,breakable]{tcolorbox}
\newtcolorbox{examplebox}{colback=white,colframe=black,sharp corners}

\usepackage{makecell}

\begin{document}

\frontmatter

\begin{titlepage}
  \newcommand{\HRule}{\rule{\linewidth}{0.5mm}}

  \center

  \textsc{\large University of Wroc{\l}aw}\\[0.1cm]
  \textsc{\large Institute of Computer Science}\\[4cm]

  \HRule \\[0.4cm]
  { \Large \bfseries CNF Encodings of Cardinality Constraints Based on \\ Comparator Networks }\\[0.4cm]
  \HRule \\[1.5cm]

  \begin{minipage}{0.45\textwidth}
  \begin{flushleft} \normalsize
  \emph{Author:}\\
  Micha\l~\textsc{Karpi\'nski}
  \end{flushleft}
  \end{minipage}
  ~
  \begin{minipage}{0.45\textwidth}
  \begin{flushright} \normalsize
  \emph{Supervisor:} \\
  dr hab. Marek \textsc{Piotr\'ow}
  \end{flushright}
  \end{minipage}\\[6cm]

  {\large \today}\\[2cm]

  \vfill
\end{titlepage}

\chapter*{Acknowledgments}

This thesis -- the culmination of my career as a graduate student (or as an apprentice, in a broader sense) -- is a consequence
of many encounters I have had throughout my life, not only as an adult, but also as a youth. It is because I was lucky to meet
these particular people on my path, that I am now able to raise a toast to my success. A single page is insufficient to
mention all those important individuals (which I deeply regret), but still I would like to express my gratitude to a few selected ones.

First and foremost, I would like to thank my advisor -- Marek Piotr\'ow -- whose great effort,
dedication and professionalism guided me to the end of this journey.
I especially would like to thank him for all the patience he had while repeatedly correcting the same lazy mistakes I made
throughout my work as a PhD student. If I had at least 10\% of his determination, I would have finished this thesis sooner.
For what he has done, I offer my everlasting respect.
I would also like to thank Prof. Krzysztof Lory\'s for advising me for many years and not loosing faith in me.

I thank a few selected colleagues, classmates and old friends, those who have expanded my horizons,
and those who on many occasions offered me a helping hand (in many ways). They have been good friends
to me for many years, and although I have lost contact with most them I still feel obliged to mention them here:
Aleksander Balicki, Marcin Barczyszyn, Paweł Bogdan, Micha{\l}~Bruli\'nski, Piotr Gajowiak, Marcin Kolarczyk, Micha{\l}~Krasowski,
Kacper Kr\'ol, Micha\l~Lepka, Mateusz Lewandowski, Marcin Oczeretko, Jacek Olchawa, Krzysztof Piecuch, Damian and Dominik Rusak,
Grzegorz Sobiegraj, Kuba St\k{e}pniewicz, Damian Straszak, Micha{\l} Szczepanowski, Jakub Tarnawski, {\L}ukasz Zapart.

I would like to separately mention my colleagues from the Institute of Computer Science, for the great time I had during the years of my PhD program:
residents of the room 340 -- Pawe\l~Garncarek, Piotr Polesiuk;
residents of the room 202 -- Jan Chorowski, Adam Kaczmarek, Adrian {\L}a\'ncucki, Micha\l~R\'o\.za\'nski;
and all other wonderful people of the Institue of Computer Science.
Andr\'es Aristiz\'abal, even though your time in Poland was short, I think of you as my brother.

One particular person needs to be mentioned separately -- Maciej Pacut -- my friend, my roommate, and the fellow PhD student,
who gave me the final push to continue my graduate studies after obtaining master's degree.
Whether his suggestion was an epiphany, or a curse, I will ponder for years to come.

Finally, I would like to thank my family, especially my mother, for providing for me and for allowing me to peacefully pursue my dreams.

\chapter*{Abstract}

Boolean Satisfiability Problem (SAT) is one of the core problems in computer science. As one of the fundamental NP-complete problems, it can be used -- by known reductions -- to represent instances of variety of hard decision problems. Additionally, those representations can be passed to a program for finding satisfying assignments to Boolean formulas, for example, to a program called \textsc{MiniSat}. Those programs (called SAT-solvers) have been intensively developed for many years and -- despite their worst-case exponential time complexity -- are able to solve a multitude of hard practical instances. A drawback of this approach is that clauses are neither expressive, nor compact, and using them to describe decision problems can pose a big challenge on its own.

We can improve this by using high-level constraints as a bridge between a problem at hand and SAT. Such constraints are then automatically translated to eqisatisfiable Boolean formulas. The main theme of this thesis revolves around one type of such constraints, namely Boolean Cardinality Constraints (or simply {\em cardinality constraints}). Cardinality constraints state that at most (at least, or exactly) $k$ out of $n$ propositional literals can be true. Such cardinality constraints appear naturally in formulations of different real-world problems including cumulative scheduling, timetabling or formal hardware verification.

The goal of this thesis is to propose and analyze new and efficient methods to {\em encode} (translate) cardinality constraints into equisatisfiable proposition formulas in CNF, such that the resulting SAT instances are small and that the SAT-solver runtime is as short as possible. The ultimate contribution of this thesis is the presentation and analysis of several new translation algorithms, that improve the state-of-the-art in the field of encoding cardinality constraints. Algorithms presented here are based on {\em comparator networks}, several of which have been recently proposed for encoding cardinality constraints and experiments have proved their efficiency. With our constructions we achieve better encodings than the current state-of-the-art, in both theoretical and experimental senses. In particular, they make use of so called {\em generalized comparators}, that can be efficiently translated to CNFs. We also prove that any encoding based on generalized comparator networks preserves {\em generalized arc-consistency} (GAC) -- a theoretical property that guarantees better propagation of values in the SAT-solver computation.

Finally, we explore the possibility of using our algorithms to encode a more general type of constraints - {\em the Pseudo-Boolean Constraints}. To this end we implemented a PB-solver based on the well-known \textsc{MiniSat+} and the experimental evaluation shows that on many instances of popular benchmarks our technique outperforms other state-of-the-art PB-solvers.

\newpage

~

\newpage

\tableofcontents

\mainmatter

\part{SAT-solving and Constraint Programming}

\chapter{Introduction}\label{ch:intro}

  \def\nqueenssolution{Qd4}
  \setchessboard{smallboard,labelleft=false,labelbottom=false,showmover=false,setpieces=\nqueenssolution}

  \begin{tikzpicture}[remember picture,overlay]
    \node[anchor=east,inner sep=0pt] at (current page text area.east|-0,3cm) {\chessboard};
  \end{tikzpicture}


Several hard decision problems can be efficiently reduced to the Boolean satisfiability (SAT)
problem and tried to be solved by recently-developed SAT-solvers. Some of them are
formulated with the help of different high-level constraints, which should be either
encoded into CNF formulas \cite{minisatp,philipp2015pblib,manthey2012npsolver}
or solved inside a SAT-solver by a specialized extension \cite{een2003extensible}.
In this thesis we study how a SAT-solver can be used 
to solve Boolean Constraint Problems by translation to clauses.

The major part of this thesis is dedicated to encoding Boolean Cardinality Constraints that take the form
$x_1 + x_2 + \dots + x_n \; \# \; k$, where  $x_1, x_2, \dots, x_n$ are Boolean literals (that
is, variables or their negations), $\#$ is a relation from the set $\{<, \le, = , \ge
, >\}$ and $k,n \in \nat$. Such cardinality constraints appear naturally in formulations of various
real-world problems including cumulative scheduling \cite{schutt2009cumulative}, timetabling
\cite{acha2014curriculum} or formal hardware verification \cite{biere1999symbolic}.

The goal of this thesis is to study the technique of translating Boolean Cardinality Constraints
into SAT, based on comparator networks approach. We propose several new classes of networks and we prove
their utility both theoretically and experimentally. We show that on many instances of popular benchmarks our
algorithms outperform other state-of-the-art solvers. The detailed description of the results is given in Section \ref{sec:org}.

This introductory chapter familiarizes the reader with the concepts of SAT-solving
and Constraint Programming (CP) -- the central topics providing motivation for this thesis.
First, we take a look at the SAT problem and its continuous interest in computer science. We show
some applications of SAT and give a short summary of the history of SAT-solving. Then, we turn
to the notion of {\em Constraint Satisfaction Problem} (CSP) and define the main object that
is studied in this thesis, namely, {\em the clausal encoding of cardinality constraints}. We end this chapter
with a section explaining how the rest of this thesis is organized and how it contributes
to the field of constraint programming.

\section{Brief History of SAT-solving}

  SAT, or in other words, {\em Boolean Satisfiability Problem} or {\em satisfiability problem of propositional logic},
  is a decision problem in which we determine whether a given Boolean formula (often called {\em propositional formula})
  has a satisfying assignment or not. We use propositional formulas in conjunctive normal form (CNF). A CNF formula
  is a conjunction (binary operator $\wedge$) of clauses. Each clause is a disjunction (binary operator $\vee$) of literals,
  where literal is an atomic proposition $x_i$ or its negation $\neg x_i$. In general, CNF formula on $n$ variables and $m$ clauses
  can be expressed as:

  \[
    \psi = \bigwedge_{i=1}^{m}\left( \bigvee_{j \in P_i} x_j \vee \bigvee_{j \in N_i} \neg x_j \right),
  \]

  \noindent where $P_i, N_i \subseteq \{1,\dots,n\}$, $P_i \cap N_i = \emptyset$, $n,m \in \nat$. To make certain ideas more clear,
  when translating something to CNF, we also call implication (binary operator $\impl$) a clause,
  keeping in mind the following equivalence:

  \[
    x_1 \wedge x_2 \wedge \dots \wedge x_n \impl y_1 \vee y_2 \vee \dots \vee y_m \;\; \Iff \;\; \neg x_1 \vee \neg x_2 \vee \dots \vee \neg x_n \vee y_1 \vee y_2 \vee \dots \vee y_m.
  \]

  \noindent We also present formulas simply as sets of clauses, and single clauses as a sets of literals,
  for succinctness. Note that any propositional formula can be transformed into an
  equivalent formula in CNF, in linear time \cite{blair1986some}.

  A {\em truth assignment} (also called 
  {\em instantiation} or {\em assignment}) is a partial function $I$ that maps variables $x \in V$ to the elements of set $\{true,false\}$.
  Therefore, a single variable can be either true, false or free. The truth values of propositional logic
  {\em true} and {\em false} will be represented by 1 and 0, respectively.
  A variable $x$ is said to be assigned to (or fixed to) 0 by instantiation $I$ if $I(x)=0$, assigned to 1 if $I(x)=1$,
  and free if $I(x)$ is undefined. In non-ambiguous context, $x=1$ denotes $I(x)=1$ (similarly for $x=0$). We also generalize
  this concept for sets of variables, so that if we write $V=1$, we mean that for each $x\in V$, $I(x)=1$ (same for $V=0$).
  An instantiation $I$ of $V$ is said to be {\em complete} if it fixes all the variables in $V$. The instantiations that are not
  complete are said to be {\em partial}. We further extend our notation, such that if $\phi$ is a Boolean formula,
  then a value of $\phi$ under assignment $I$ is denoted by $I(\phi)$, which can be either $0$, $1$ or undefined.
  Furthermore, if $I$ is a complete instantiation and $\bar{x}=\tuple{x_1,\dots,x_n}$ is a sequence of Boolean literals, then
  $I(\bar{x})=\tuple{I(x_1),\dots,I(x_n)}$.

  Although the language of SAT is very limited, it is very powerful, allowing us to model many mathematical and
  real-world problems. Unfortunately we do not expect to see a fast algorithm that could solve all SAT instances, as the 
  problem is NP-complete by the famous theorem of Cook \cite{cook1971complexity}. The situation did not improve much after
  almost 50 years. The best known deterministic algorithm solving SAT runs in worst-case time $O(1.439^n)$ \cite{kutzkov2010using},
  where $n$ is the number of variables. 

  On the positive side, specialized programs called {\em SAT-solvers} have emerged,
  and even though they process CNFs -- in worst case -- in exponential time with respect to the size of the 
  formula, for many practical instances they can quickly determine the satisfiability of the formula.
  This section presents the most important milestones in the field of SAT-solving. Before we look at
  historical results, let us see why solving SAT instances is a task of great significance.

  \subsection{Applications}\label{sec:sat:app}

  \paragraph{Mathematical puzzles.} We begin our journey with some academic examples.
  First, let us take a look at a 170-year old puzzle called {\em 8-Queens Puzzle}
  (a recurring theme of this thesis). The 8-Queens Puzzle is the problem of placing eight chess queens on an $8 \times 8$
  chessboard so that no two queens threaten (attack) each other. Chessboard consists of 64 squares with eight rows called 
  {\em ranks} labeled 1--8 and eight columns called {\em files} labeled a--h. The {\em queen} chess piece from its position observes
  all other squares on the rank, file and both diagonals that she currently occupies (see marked squares in Figure \ref{fig:8queens}a).
  Thus, a solution requires that no two queens share the same rank, file, or diagonal. In Figure \ref{fig:8queens}b we see a chessboard with
  four queens on it: queen on a1 attacks queen on f6, queen on f2 also attacks queen on f6, and since the relation of "attacks" is symmetrical,
  queen on f6 attacks both queens on a1 and f2. On the other hand, queen on c4 does not threaten any other queen on the board.
 
 \begin{figure}[ht!]
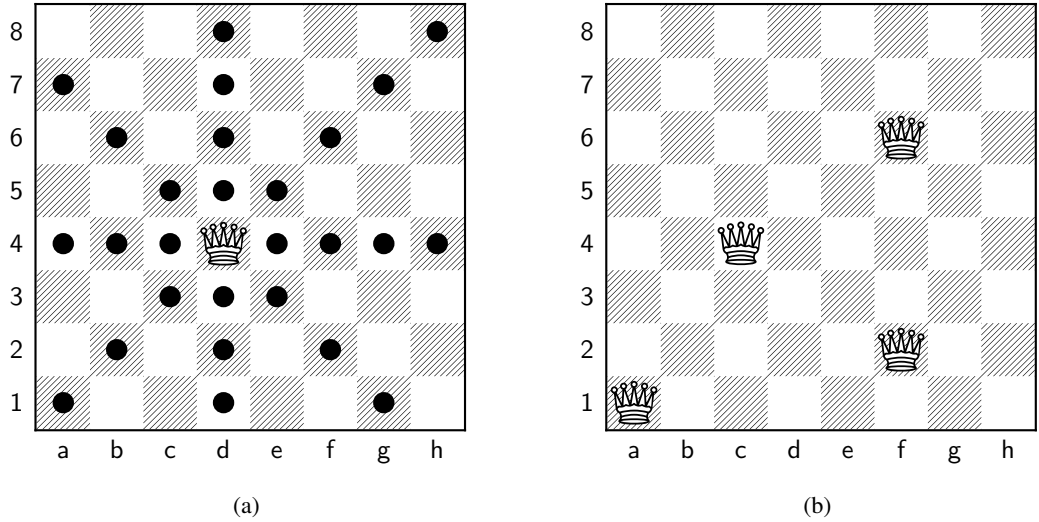

    \centering
    \subfloat[]{
      \def\nqueenssolution{Qd4}
      \setchessboard{showmover=false,labelleft=true,labelbottom=true,setpieces=\nqueenssolution}
      \chessboard[normalboard,pgfstyle={[fill]circle},padding=-0.8ex,color=black,
        markfields={d5,d6,d7,d8,d3,d2,d1,a4,b4,c4,e4,f4,g4,h4,a1,b2,c3,e5,f6,g7,h8,a7,b6,c5,e3,f2,g1}]
    }
    ~
    \subfloat[]{
      \def\nqueenssolution{Qa1, Qc4, Qf6, Qf2}
      \setchessboard{showmover=false,labelleft=true,labelbottom=true,setpieces=\nqueenssolution}
      \chessboard[normalboard]
    }
    \caption{Rules of the n-Queens puzzle: (a) all squares that a single queen can attack; (b) sample placement of queens on board}
    \label{fig:8queens}
  \end{figure}

  The 8-Queens Puzzle can be generalized to $n$-Queens Puzzle -- a problem of placing
  $n$ non-attacking queens on an $n \times n$ chessboard -- for which solutions exist for
  all natural numbers $n>3$. For simplicity, let us focus on the case $n=4$ and try to construct a clause
  set that is satisfiable, if and only if 4-Queens Puzzle has a solution.

  We would like to model our SAT instance in a way that allows easy extraction of a solution from each satisfying assignment.
  A natural way to achieve it is to introduce a variable $x_{i,j}$ for each square on the board,
  where $i \in \{a,b,c,d\}$ and $j \in \{1,2,3,4\}$. This way, we can relate a placement of a queen
  on a square $\{i,j\}$ with the assignment of the variable $x_{i,j}$. Now, we need to create a set of clauses that restricts
  the placement of queens according to the rules. We do this by adding a set of clauses $\psi$ that are satisfiable,
  if and only if exactly one variable is set to true for every set of variables
  that represents some rank and file,
  and at most one variable is true for each diagonal. For example, for the 1st rank we add:

  \small
  \begin{align*}
    & (\neg x_{a,1} \vee \neg x_{b,1}), (\neg x_{a,1} \vee \neg x_{c,1}), (\neg x_{a,1} \vee \neg x_{d,1}),
      (\neg x_{b,1} \vee \neg x_{c,1}), (\neg x_{b,1} \vee \neg x_{d,1}), (\neg x_{c,1} \vee \neg x_{d,1}), \\
    & (x_{a,1} \vee x_{b,1} \vee x_{c,1} \vee x_{d,1}),
  \end{align*}
  \normalsize

  \noindent and for the a2-c4 diagonal, we add:

  \small
  \begin{align*}
    & (\neg x_{a,2} \vee \neg x_{b,3}), (\neg x_{a,2} \vee \neg x_{c,4}), (\neg x_{b,3} \vee \neg x_{c,4}).
  \end{align*}
  \normalsize

  \noindent We proceed similarly for all other ranks, files and diagonals on the board. We do not show the full encoding for readability reasons.
  This would require printing exactly 84 clauses in total for $\psi$.
  This example illustrates that even for small instances of the considered problem, we can get large sets of clauses.

  For $n=4$, the n-Queens Puzzle has exactly two solutions, as presented in Figure \ref{fig:4queens}. As we can see,
  for both $S_1=\{x_{a,3}, x_{b,1}, x_{c,4}, x_{d,2}\}$ and $S_2=\{x_{a,2}, x_{b,4}, x_{c,1}, x_{d,3}\}$, setting
  either $S_1=1$ or $S_2=1$ (and the rest of the variables to 0) gives a satisfying assignment of $\psi$. Figuring
  out whether there is no other solution simply by examining $\psi$ would be cumbersome, but even the simplest SAT-solver
  would find the answer instantly.

  \begin{figure}[t!]
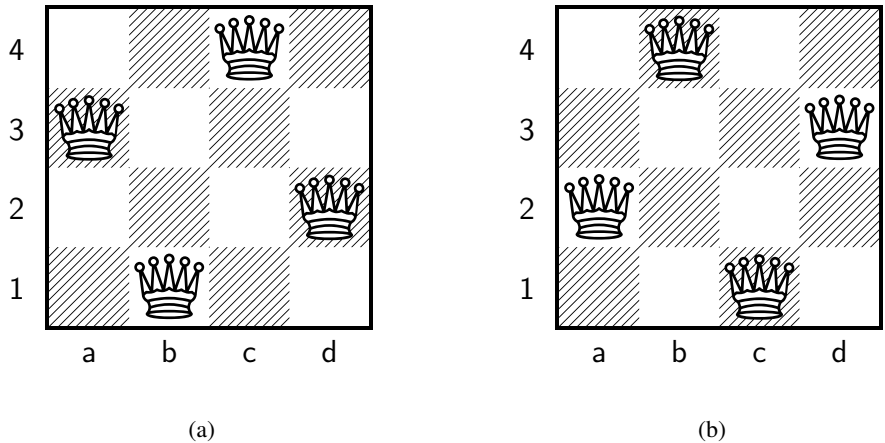

    \centering
    \subfloat[]{
      \def\nqueenssolution{Qa3, Qb1, Qc4, Qd2}
      \setchessboard{largeboard,showmover=false,labelleft=true,labelbottom=true,setpieces=\nqueenssolution,maxfield=d4}
      \chessboard
    }
    ~
    \subfloat[]{
      \def\nqueenssolution{Qa2, Qc1, Qb4, Qd3}
      \setchessboard{largeboard,showmover=false,labelleft=true,labelbottom=true,setpieces=\nqueenssolution,maxfield=d4}
      \chessboard
    }
    \caption{All solutions to 4-Queens puzzle}
    \label{fig:4queens}
  \end{figure}

  There are many more logic puzzles that could be modeled as SAT instances, for example: Sudoku, Magic squares, Nonograms etc.
  However, SAT-solving would not become so popular if it was used only for recreational purposes. Thus, we now turn to more
  practical applications.

  \paragraph{Timetabling.} In many real-world situations it is useful to have a chart showing
  some events scheduled to take place at particular times. Examples of that would be: a chart showing the departure
  and arrival times of trains, buses, or planes; a class schedule for students and teachers. Another example can be
  a shift schedule at a workplace. Let us consider the following hypothetical situation.

  Assume we own a restaurant and among all our employees we hire 5 waiters: Adam, Brian, Carl, Dan and Eddy. We want to create a shift
  schedule for them for the next week (starts on Monday, ends on Sunday) given the following constraints:

  \begin{enumerate}
    \item Carl cannot work with Dan, because they do not like each other.
    \item Eddy is still learning, therefore he has to work together with either Adam or Brian, who are more experienced.
    \item No waiter can work for three consecutive days.
    \item Brian cannot work on weekends, because of his studies.
    \item Each day, at least two waiters need to work, except Friday, when we expect a lot of customers.
          Then at least three waiters have to be present.
  \end{enumerate}

  We can easily model this problem as a SAT instance. First, we create variables indicating that a waiter is assigned to a day of the week.
  To this end we introduce $x_{v,i}$ for each waiter $v \in W = \{A,B,C,D,E\}$ and for each day $i \in \{1,2,3,4,5,6,7\}$ (1 - Monday, 2 - Tuesday, etc.).
  Now, if $x_{v,i}$ is true, then a waiter $v$ has to come to work on day $i$. We add the following clauses to encode the given constraints:

  \begin{enumerate}
    \item We simply add $(\neg x_{C,i} \vee \neg x_{D,i})$ for each day $i$.
    \item If Eddy works in a given day $i$, then either Adam or Brian has to work: $(x_{E,i} \impl x_{A,i} \vee x_{B,i})$.
    \item For each waiter $v \in W$ and day $i<6$ we add: $(\neg x_{v,i} \vee \neg x_{v,i+1} \vee \neg x_{v,i+2})$.
    \item Brian's absence on Saturday and Sunday can be handled by two singleton clauses: $(\neg x_{B,6}) \wedge (\neg x_{B,7})$,
    virtually setting $x_{B,6}$ and $x_{B,7}$ to 0.
    \item We add a clause $(x_{A,i} \vee x_{B,i} \vee x_{C,i} \vee x_{D,i} \vee x_{E,i})$ for each day $i$ to force at least
    one waiter to appear for work each day. To make at least two waiters go to work each day $i$, we add a clause for each waiter $v \in W$:
    $(x_{v,i} \impl \bigvee_{v' \in W-\{v\}} x_{v',i})$. To make three waiters come on Friday, we add the set of clauses:
    $((x_{v,5} \wedge x_{v',5}) \impl \bigvee_{v'' \in W-\{v,v'\}} x_{v'',5})$, for each pair of $v,v' \in W$, where $v\neq v'$.
  \end{enumerate}

  The union of all above clauses forms a CNF that is satisfiable, if and only if it is possible to create a shift schedule
  that satisfied all the constraints. The solution can be extracted by simply looking at the truth assignment of the variables. If no solution
  exists, then as managers we need to either relax some of the constraints or hire more staff.

  The problem of timetabling goes beyond our hypothetical considerations. For example, As{\'\i}n and Nieuwenhuis \cite{acha2014curriculum}
  present SAT encodings for Curriculum-based Course Timetabling and shows that experiments performed on real-world
  instances improves on what was then considered the state-of-the-art.

  \paragraph{Stable Marriages.} In examples above encodings were fairly straightforward,
  with a 1-to-1 correspondence between variable definitions and values' truth assignments.
  A slightly more complex example is shown here. It is based on a study by Gent et al. \cite{gent2001constraint}.

  In the {\em stable marriage problem with preference lists} we have $n$ men and $n$ women. Each man $i$ ranks the women in order
  of preference, and women similarly rank men, creating preference lists $L^m_{i}$ for man and $L^w_{i}$ for women. The closer a person is on the list
  the more desirable that person is as a mate. The problem is to marry men and women in a {\em stable} way, meaning that
  there is no incentive for any two individuals to elope and marry each other. A person is willing to leave his/her current partner
  for a new partner only if he/she is either unmatched or considers the new partner better than the current one. A pair who
  mutually prefer each other than their partners is a blocking pair, and a matching without blocking pairs is stable.

  We define a variable $x_{i,p}$ to be true, if and only if the man $i$ is either unmatched or matched to the women in position $p$ or later in
  his preference list, where $1 \leq p \leq |L^m_{i}|$. We also define $x_{i,|L^m_{i}|+1}$ which is true, if and only if man $i$ is unmatched.
  Likewise, we define variables $y_{j,q}$ for each women $j$. Note that each variable valuation corresponds to a set of possibilities and that marriages
  are modeled indirectly via the preference lists. We define the clause set as follows:

  \begin{itemize}
    \item Let $1 \leq i \leq n$. Each man or woman is either matched with someone in their preference list or is unmatched:
    $x_{i,1} \wedge y_{i,1}$.
    \item Let $2 \leq p \leq |L^m_{i}|$ and $2 \leq q \leq |L^w_{j}|$. If a man $i$ (women $j$) gets his (hers) $(p-1)$-th or better choice,
      then he (she) certainly gets his (hers) $(p)$-th or better choice:
      $(\neg x_{i,p} \impl \neg x_{i,p+1}) \wedge (\neg y_{j,p} \impl \neg y_{j,p+1})$.
    \item  Now we express the monogamy constraints. Let $p$ be the rank of women $j$ in
      the preference list of man $i$, and $q$ be the rank of man $i$ in the preference list of woman $j$. 
      If man $i$ has partner no better than women $j$ or is unmatched, and women $j$ has a partner she prefers to man $i$, then man $i$
      cannot be matched to women $j$: $(x_{i,p} \wedge \neg y_{j,q} \impl x_{i,p+1})$.
      Similarly: $(y_{j,q} \wedge \neg x_{i,p} \impl y_{j,q+1})$.
    \item Finally, we enforce stability by stating that if man $i$ is matched to a woman he ranks no better than woman $j$,
    then woman $j$ must be matched
    to a man she ranks no lower than man $i$, and vice-versa. Again, let $p$ be the rank of women $j$ in
      the preference list of man $i$, and $q$ is the rank of man $i$ in the preference list of woman $j$. Then, we add:
      $(x_{i,p} \impl \neg y_{j,q+1}) \wedge (y_{j,q} \impl \neg x_{i,p+1})$.
  \end{itemize}

  In the book by Biere et al. \cite{biere2009handbook} authors report on multitude
  of different problems originating from computer science,
  that can be handled by translation to SAT. The list of applications consists of, among others: software verification, bounded model checking,
  combinatorial design, and even statistical physics. Finally, before reviewing the history of SAT-solving, we take a look at
  the optimization version of SAT, called MaxSAT, by which not only decision problems but also optimization ones can be modeled with clauses.

  \paragraph{MaxSAT.} SAT-solvers usually either report on the found solution or inform the user, that the given instance is unsatisfiable.
  In practice, we would like to gain more insight about the unsatisfiable instance, for example, which set of clauses causes the
  unsatisfiability or what is the maximum number of clauses that can be satisfied by some truth assignment. In the example with creating
  a shift schedule for the restaurant, if we knew that the instance has no solutions, and clauses representing the first constraint makes
  the CNF unsatisfiable, then we could just remove them, tell Carl and Dan to act professionally regardless of their personal animosity, and
  move on with the schedule. We could do that if we modeled our problem in MaxSAT.

  {\em Maximum Satisfiability Problem}, or MaxSAT in short, consists of finding an assignment
  that maximizes the number of satisfied clauses in the given CNF. Even though MaxSAT is NP-hard,
  some success was made in finding approximate solutions. The first MaxSAT polynomial-time approximation
  algorithm, created in 1974, with a performance guarantee of $1/2$, is a greedy algorithm by Johnson \cite{johnson1974approximation},
  and the best theoretical result is of Karloff and Zwick \cite{karloff19977}, who gave a $7/8$ algorithm for $Max3SAT$
  (variant, where at most three literals are allowed per clause). We do not expect to get any better than that,
  unless $P=NP$ \cite{haastad2001some}. 

  Nevertheless, in many applications an exact solution is required. Similarly to SAT, 
  we would also like to solve MaxSAT instances quickly in practical applications.
  In recent years, there has been considerable interest in developing efficient algorithms and several
  families of algorithms (MaxSAT-solvers) have been proposed. For recent survey, see \cite{morgado2013iterative}.

  In the context of this section, one particular method is of interest to us: the SAT-based approach.
  It was first developed by Le Berre and Parrain \cite{le2010sat4j}. Given a MaxSAT instance $\psi=\{C_1,\dots,C_m\}$,
  a new blocking variable $v_i$, $1 \leq i \leq m$, is added to each clause $C_i$, and solving the MaxSAT problem for
  $\psi$ is reduced to minimize the number of satisfied blocking variables in $\psi'=\{C_1 \vee v_1, \dots, C_m \vee v_m\}$.
  Then, a SAT-solver that supports cardinality constraints (possibly by algorithms given in this thesis in Chapters 4--6)
  solves $\psi'$, and each time a satisfying assignment $A$ is found, a better satisfying assignment is searched by adding
  the cardinality constraint $v_1 + \dots + v_m < B(A)$, where $B(A)$ is the number of blocking variables satisfied by $A$.
  Once $\psi'$ with a newly created cardinality constraint is unsatisfiable, the latest satisfying assignment is declared to be an optimal
  solution.

\subsection{Progress in SAT-solving}

The success of SAT-solving and its expanding interest comes from the fact that SAT stands at the crossroads of many fields,
such as logic, graph theory, computer engineering and operations research. Thus, many problems with origin in one of these
fields usually have multiple translations to SAT. Additionally, with plethora of ever-improving SAT-solvers available for both individual, academic
and commercial use, one has almost limitless possibilities and modeling tools for solving a variety of scientific and practical problems.  

The first program to be considered a SAT-solver was devised in 1960 by Davis and Putnam \cite{davis1960computing}.
Their algorithm based on resolution is now called simply DP (Davis-Putnam). Soon after, an improvement to DP was proposed by
Davis, Logemann and Loveland \cite{davis1962machine}. The new algorithm called DPLL (Davis-Putnam-Logemann-Loveland) guarantees
linear worst-case space complexity.

The basic implementation of a SAT-solver consists of a backtracking algorithm that
for a given formula $\psi$ chooses a literal, assigns a value to it (let us say true), and simplifies all the clauses containing
that literal resulting in a new formula $\psi'$. Then, a recursive procedure checks the satisfiability of $\psi'$.
If $\psi'$ is satisfied, then $\psi$ is also satisfied. Otherwise, the same recursive check is done assuming
the opposite truth value (false) of the chosen literal. DPLL enhances this simple procedure by eagerly using {\em unit propagation}
and {\em pure literal elimination} at every step of the algorithm. Unit propagation eliminates {\em unit clauses}, i.e., clauses
that contain only a single unassigned literal. Such clause can be trivially satisfied by assigning the necessary value to make the literal true.
Pure literals are variables that occur in the formula in only one polarity. Pure literals can always be assigned in a way that makes all clauses
containing them true. Such clauses can be deleted from the formula without changing its satisfiability.

Most modern methods for solving SAT are refinements of the basic DPLL framework, and include
improvements to variable choice heuristics or early pruning of the search space. In the DPLL
algorithm it is unspecified how one should choose the next variable to process. Thus,
many heuristics have emerged over the years. For example, the unit propagation rule was
generalized to the {\em shortest clause rule}: choose a variable from a clause containing
the fewest free literals \cite{chao1990probabilistic}. Another example is the 
{\em majority rule} \cite{chao1986probabilistic}: choose a variable with
the maximum difference between the number of its positive and negative literals.
Later, it was observed that the activity of variable assignment was an important factor
in search space size. This led to the VSIDS (Variable State Independent Decaying Sum)
heuristic of Chaff \cite{moskewicz2001chaff} that assigns to each variable a score proportional
to the number of clauses that variable is in. As the SAT algorithm progresses, periodically,
all scores are divided by a constant. VSIDS selects the next variable with the highest score
to determine where the algorithm should branch.

Another refinement is a method of early detection of a sequence of variable decisions which results
in a lot of propagation, and therefore reduction of the formula. The idea is to run the backtracking procedure
for some constant number of steps and then rollback the calculations
remembering the sequence of variable decisions that reduced the input formula the most.
This way, the most promising parts of the search tree can be explored first, so the chance
of finding a solution early increases.
This framework has been established as {\em look-ahead} algorithms and the main representatives
of this trend are the breath-first technique by St{\aa}lmarck \cite{sheeran2000tutorial}
and depth-first technique (now called {\em restarts}) implemented in Chaff \cite{moskewicz2001chaff}.

Arguably, the biggest step forward in the field of SAT-solving was
construction of the CDCL algorithm ({\em conflict-driven clause-learning}).
This turned many problems that were considered
intractable into trivially solvable problems. The CDCL algorithm can
greatly reduce the search space by discovering early that certain branches
of the search tree do not need to be explored. In short, the goal of CDCL
is to deduce new clauses when discovering a conflict (unsatisfiable clause) during
the exploration of the search tree. Those clauses are constructed, so that the same conflict cannot be repeated.
First solver to successfully apply this technique was Chaff \cite{moskewicz2001chaff}.
Then in MiniSat \cite{een2003extensible} several improvements were implemented.
Now, many more top solvers are built upon CDCL as the main
heuristic, for example, Glucose \cite{audemard2009predicting}.

After CDCL, no major improvement has been made.
The current trend shifts toward parallelization of SAT-solvers.
The related work on this topic can be found, for example,
in the PhD thesis of Norbert Manthey \cite{manthey2016towards}.
See \cite{biere2009handbook} for further reading on the subject of satisfiability.

\section{Introduction to Constraint Programming}

Problem formulations in Section \ref{sec:sat:app} use expressions like: "{\bf exactly one} queen" or
"{\bf at least two} waiters". These are examples of {\em constraints}, and more specifically -- cardinality constraints. In the example
with MaxSAT such constraints were defined explicitly. In this section we introduce the field of {\em Constraint Programming} and
provide motivation behind its success. We base this section on parts of the book by Apt \cite{apt2003principles}.

Informally, a {\em constraint} on a sequence of variables is a relation on their domains. It can be viewed as a requirement
that states which combinations of values from the variable domains are allowed. To solve a given problem by means of
constraint programming we express it as a {\em constraints satisfaction problem} (CSP), which consists of a finite set of constraints.
To achieve it, we introduce a set of variables ranging over specific domains and constraints over these variables. Constraints are expressed in a specific language, for instance,
in SAT variables range over the Boolean domain and the only constraints
that the language allows are clauses.

\begin{examplebox}
\begin{example}
  In the n-Queens puzzle, exactly one queen must be placed on each row and column, and at most one on each diagonal of the chessboard.
  This can be modeled as a CSP by introducing the following constraints:

  \begin{itemize}
    \item $Eq_1(x_1,\dots,x_n)$ with semantics that exactly one out of $n$ propositional literals $\{x_1,\dots,x_n\}$ can be true,
    \item $Lt_1(x_1,\dots,x_n)$ stating that at most one out of $n$ propositional literals $\{x_1,\dots,x_n\}$ can be true. 
  \end{itemize}

  Returning to the example from Section \ref{sec:sat:app}, if $n=4$ then we can express the problem
  using the following set of constraints:

  \small
  \begin{align*}
    & Eq_1(x_{a,1},x_{b,1},x_{c,1},x_{d,1}), Eq_1(x_{a,2},x_{b,2},x_{c,2},x_{d,2}), Eq_1(x_{a,3},x_{b,3},x_{c,3},x_{d,3}), Eq_1(x_{a,4},x_{b,4},x_{c,4},x_{d,4}), \\
    & Eq_1(x_{a,1},x_{a,2},x_{a,3},x_{a,4}), Eq_1(x_{b,1},x_{b,2},x_{b,3},x_{b,4}), Eq_1(x_{c,1},x_{c,2},x_{c,3},x_{c,4}), Eq_1(x_{d,1},x_{d,2},x_{d,3},x_{d,4}), \\
    & Lt_1(x_{a,2}, x_{b,1}), Lt_1(x_{a,3}, x_{b,2}, x_{c,1}), Lt_1(x_{a,4}, x_{b,3}, x_{c,2}, x_{d,1}), Lt_1(x_{b,4}, x_{c,3}, x_{d,2}), Lt_1(x_{c,4}, x_{d,3}), \\
    & Lt_1(x_{a,3}, x_{b,4}), Lt_1(x_{a,2}, x_{b,3}, x_{c,4}), Lt_1(x_{a,1}, x_{b,2}, x_{c,3}, x_{d,4}), Lt_1(x_{b,1}, x_{c,2}, x_{d,3}), Lt_1(x_{c,1}, x_{d,2}). \\
  \end{align*}
  \normalsize

  The CNF from Section \ref{sec:sat:app} required 84 clauses. The above CSP is definitely more succinct.
\end{example}
\end{examplebox}

The next step is to solve the CSP by a dedicated {\em constraint solver} which can incorporate various domain specific methods
and general methods, depending on the type of constraints we are dealing with. The domain specific methods are usually
provided in the form of implementations of special purpose algorithms, for example, a program that solves systems of
linear equations, or a package for linear programming. On the other hand, the general methods are concerned with the ways of
reducing the search space. These algorithms maintain equivalence while simplifying the considered problem.
One of the aims of constraint programming is to search for efficient domain specific methods
that can be used instead of the general methods and to apply them into a general framework.
While solving CSPs we are usually interested in: determining whether the instance has a solution, finding all solutions, and
finding all (or some) optimal solutions w.r.t. some goal function.

One of the advantages of modeling problems with constraints over the development of classic algorithms is the following.
The classic computational problem is usually formulated in a very basic form and rarely can be applied to real-world situation as it is.
In fact, the shorter the description, the more "fundamental" the problem feels, and the bigger interest in the community.
For example: finding a matching in graph with some properties,
the shortest path between nodes in graph, the minimum spanning tree, the minimum cut, the number of sub-words in compressed text, etc.
In the real-world industrial applications, there are plenty of constraints to handle at once. The hope that someone would construct a classic algorithm
to solve some complicated optimization problem we throw at them is bleak. We will now take a peek at couple of such problems.

\subsection{Applications}

Apt \cite{apt2003principles} provides a long list of problems where CSPs were successfully applied. The list consists of:
interactive graphic systems, scheduling problems, molecular biology, business applications, electrical engineering, numerical computation,
natural language processing, computer algebra. It has been 15 years since the book was 
published and more applications for constraint programming have emerged
since then. One can find more recent (and more exotic) examples by studying the {\em Application Track}
of the {\em Principles and Practice of Constraint Programming} conference. We reference some of them from the recent years here.

\paragraph{Facade-Layout Synthesis.} This problem occurs when buildings are renovated to improve their thermal insulation and reduce
the impact of heating on the environment. It involves covering a facade with a set of disjoint and configurable insulating panels.
This can be viewed as a constrained rectangle packing problem for which the number of rectangles to be used and their size are not known
{\em a priori}. Barco et al. \cite{barco2015open} devise a CSP for the facade-layout problem. They point out that 
buildings energetic consumption represents more than one third of total energy consumption in developed countries,
which provides great motivation for studying this problem.

\paragraph{Differential Harvesting.} In grape harvesting, the machines (harvesters) are supplied with two hoppers --
tanks that are able to differentiate between two types of grape quality. Optimizing harvest consists on minimizing the working
time of a grape harvester. Given estimated qualities and quantities on the different areas of the vineyard, the problem is to optimize
the routing of the harvester under several constraints. Briot et al. \cite{briot2015constraint} model the differential
harvest problem as a constraint optimization problem and present experimental results on real data.

\paragraph{Transit Crew Rescheduling.} Scheduling urban and trans-urban transportation is an important issue for industrial societies.
The urban transit crew scheduling problem is one of the most important optimization problems related to this issue.
Lorca et al. \cite{lorca2016using} point out that this problem has been intensively studied from a tactical point of view,
but the operational aspect has been neglected while the problem becomes more and more complex and prone to disruptions.
In their paper, they present how the constraint programming technologies are able to recover the tactical plans at the operational level
in order to efficiently help in answering regulation needs after disruptions.

\paragraph{Reserve Design.} An interesting problem originates from the field of ecology -- the delineation of areas of high ecological
or biodiversity value. The selection of optimal areas to be preserved necessarily results from a compromise between the complexity of
ecological processes and managers' constraints. A paper by Justeau-Allaire et al. \cite{justeau2018unifying} shows that
constraint programming can be the basis of a unified, flexible and extensible framework for planning the reserve. They use their 
model on a real use-case addressing the problem of rainforest fragmentation in New Caledonia.

\subsection{Types of Constraints}

Over the years there have been many constraint types used in CSPs. This thesis studies {\em cardinality constraints} over
Boolean domain. For $k \in \nat$, $n$ propositional literals $\{x_1,\dots,x_n\}$ and a relation $\# \in \{<,\leq,=,\geq,>\}$,
a cardinality constraints takes the form:

  \[
    x_1 \;\; + \;\; x_2 \;\; + \;\; \cdots \;\; + \;\; x_n  \;\; \# \;\; k.
  \]

\noindent Informally, it means that at least (at most, or exactly) $k$ out of $n$ propositional literals can be true.
The main contribution of this thesis is the presentation of new methods to translate such constraints into CNF formulas.
For the historical review on this subject we dedicate entire Chapter \ref{ch:history}.

We are also interested in a closely related generalization of cardinality constraints called {\em Pseudo-Boolean constraints},
or PB-constraints, in short. We define them in a similar way, but with additional integer coefficients $\{a_1,\dots,a_n\}$:

  \[
    a_1x_1 \;\; + \;\; a_2x_2 \;\; + \;\; \cdots \;\; + \;\; a_nx_n \;\; \# \;\; k.
  \]

\noindent From other types of constraints heavily studied by the community, we choose to mention the following:

\begin{itemize}
  \item $all\_different(x_1,\dots,x_n)$ -- a constraint stating that each variable --
    defined with its own domain -- need to be assigned a unique value.
  \item Linear inequalities over reals, i.e., the language of Linear Programming \cite{schrijver1998theory}.
\end{itemize}

\subsection{Clausal Encoding}

There are many native CSP solvers available for different types of constraints. This thesis focuses on another
approach in which we {\em encode} or {\em translate} a CSP (in our case a set of cardinality constraints) into
a CNF formula. The generality and success of SAT-solvers in recent years has led to many CSPs being encoded
and solved via SAT \cite{ansotegui2004mapping,barahona2014efficient,barahona2014representative,nguyen2013application,prestwich2009cnf,tamura2009compiling}.
The idea of encoding cardinality constraints into SAT is captured by the following definition.

\begin{definition}[clausal encoding]\label{def:cl_encoding}
  Let $k,n \in \nat$. A clause set $E$ over variables $V=\{x_1,\dots,x_n,$ $s_1,\dots,s_m\}$
  is a {\em clausal encoding} of $x_1 + x_2 + \cdots + x_n \leq k$
  if for all assignments $\alpha : \{x_1,\dots,x_n\} \rightarrow \{0,1\}$ there is an extension of
  $\alpha$ to $\alpha^* : V \rightarrow \{0,1\}$ that satisfies $E$ if and only if $\alpha$ satisfies the original constraint
  $x_1 + x_2 + \cdots + x_n \leq k$, i.e., if and only if at most $k$ out of the variables $x_i$ are set to $1$ by $\alpha$.
\end{definition}

\noindent The similar notions can be defined for relations other than $\leq$ and other types of constraints,
for example, Pseudo-Boolean constraints, but we omit that to avoid repetition.

The main idea in developing algorithms for solving constraint satisfaction problems is to reduce a given CSP
to another one that is equivalent but easier to solve. The process is called {\em constraint propagation} and
the algorithms that achieve this reduction are called {\em constraint propagation algorithms}. In case of SAT-solving
we use {\em unit propagation}.

Informally, Unit Propagation (UP) is a process, where for a given CNF formula and a partial assignment (initially -- empty),
clauses are sought in which all literals but one are false (say $l$) and $l$ is undefined (initially
only clauses of size one satisfy this condition). This literal $l$ is set to true and
the process is iterated until a fix point is reached. A formal definition is presented
in Chapter~\ref{ch:pre}.

We try to construct encodings that guarantee better propagation of values using unit propagation.
We use is the notion of {\em general arc-consistency} (GAC), 
often shortened to just {\em arc-consistency} (which usually has different meaning in CP theory and deals with binary constraints).
We use this notion in the context of cardinality constraints and unit propagation in SAT. Informally,
an encoding of $x_1 + x_2 + \cdots + x_n \leq k$ is arc-consistent if as soon as $k$ input variables are fixed to $1$,
unit propagation will fix all other input variables to $0$.
A formal definition is presented in Chapter~\ref{ch:pre}. If the encoding is arc-consistent,
then this has a positive impact on the practical efficiency of SAT-solvers.

\section{Thesis Contribution and Organization}\label{sec:org}

At the highest level, the thesis improves the state-of-the-art of encoding cardinality constraints
by introducing several algorithms based on selection networks.
The structure of the thesis and the contribution is the following. The first part consists of three chapters, one of them is this introduction, and
the other two are:

\begin{itemize}
  \item In Chapter \ref{ch:pre} we introduce the necessary definitions, notation and conventions used throughout the thesis.
  We define comparator networks and introduce basic constructions used in encoding of cardinality constraints.
  We give the definition of a {\em standard encoding} of cardinality constraints.
  We show how a single comparator can be encoded using a set of clauses and how to generalize
  a comparator to directly select $m$ elements from $n$
  inputs. This is the main building block of our fastest networks presented here -- the generalized selection networks.
  We also present a proof of arc-consistency for all encodings presented in the later parts.
  In fact, we give the first rigorous proof of a more stronger statement,
  that any standard encoding based on generalized selection networks preserves arc-consistency.
  There are several results where researchers use properties of
  their constructions to prove arc-consistency, which is usually long and technical
  (see, for example, \cite{abio2013parametric,asin2011cardinality,codish2010pairwise}).
  In \cite{karpinski2015smaller} we relieve some of this burden
  by proving that the standard encoding of any selection network preserves
  arc-consistency. Here we generalize our previous proof
  to the extended model of selection networks.

  \item In Chapter \ref{ch:history} we present a historic review of methods
  for translating cardinality constraints and Pseudo-Boolean constraints into SAT.
\end{itemize}

The second part consists of two chapters dedicated to the pairwise selection networks:

\begin{itemize}
  \item We begin Chapter \ref{ch:cp15} with the presentation of the {\em Pairwise Selection Network} (PSN)
  by Codish and Zazon-Ivry \cite{codish2010pairwise}, which is based on the {\em Pairwise Sorting Network}
  by Parberry \cite{parberry1992pairwise}. The goal of this chapter is to improve the PSN
  by introducing two new classes of selection networks called {\em Bitonic Selection Networks}
  and {\em Pairwise (Half-)Bitonic Selection Networks}. We prove that we have produced a smaller selection
  network in terms of the number of comparators. We estimate also the size of our networks and compute the difference in sizes
  between our selection networks and the PSN. The difference
  can be as big as $n\log n / 2$ for $k = n/2$.

  \item In Chapter \ref{ch:mw} we show construction of the {\em $m$-Wise Selection Network}.
  This is the first attempt at creating an encoding that is based on the generalized selection networks introduced in Chapter \ref{ch:pre}.
  The algorithm uses the same {\em pairwise} idea as in Chapter \ref{ch:cp15}, but the main difference is that
  the basic component of the new network is an $m$-selector (for $m \geq 2$).
  The new algorithm works as follows. The inputs are organized into $m$ columns, in which elements are recursively selected
  and, after that, columns are merged using a dedicated merging network. 
  The construction can be directly applied to any values of $k$ and $n$
  (in contrast to the previous algorithms from Chapter \ref{ch:cp15}).
  We show the high-level algorithm for any $m$, but we only present a complete construction (which includes a merging network) for $m=4$ 
  for theoretical evaluation, where we prove that using 4-column merging networks produces smaller encodings than their 2-column counterpart.
\end{itemize}

The third part focuses on the generalized version of the odd-even selection networks. Here we show our best construction
for encoding cardinality constraints, as well as the description of a PB-solver based on the same algorithm:

\begin{itemize}
  \item In Chapter \ref{ch:cp18} we show more efficient construction using generalized comparator networks model.
  We call our new network the {\em $m$-Odd-Even Selection Network}. It generalizes the standard odd-even algorithm
  similarly to how {\em $m$-Wise Selection Network} generalizes PSN in Chapter \ref{ch:mw}. The inputs are organized into $m$ columns,
  and after recursive calls, resulting elements are merged using a dedicated merging network.
  The calculations show that encodings based on our merging networks use less number of
  variables and clauses, when compared to the classic 2-column construction.
  In addition, we investigate the influence of our encodings on the execution times of SAT-solvers
  to be sure that the new algorithms can be used in practice.
  We show that generalized comparator networks are superior to standard selection
  networks previously proposed in the literature, in the context
  of translating cardinality constraints into propositional formulas. We also conclude that
  although encodings based on pairwise approach use less number of variables than odd-even encodings,
  in practice, it is the latter that achieve better reduction in SAT-solving runtime.
  It is a helpful observation, because from the practical point of view, implementing odd-even
  networks is a little bit easier.

  \item In Chapter \ref{ch:pos18} we explore the possibility of using our algorithms in encoding Pseudo-Boolean constraints, which
  are more expressive than simple cardinality constraints. We describe the system for solving PB-problems
  based on the popular \textsc{MiniSat+} solver by E{\'e}n and S\"orensson \cite{minisatp}. 
  Recent research have favored the approach that uses Binary Decision Diagrams (BDDs),
  which is evidenced by several new constructions and optimizations \cite{abio2012,sakai2015}.
  We show that encodings based on comparator networks can still be very competitive.
  We have extended \textsc{MiniSat+} by adding a construction of selection network called 4-Way Merge Selection Network
  (an improved version of $4$-Odd-Even Selection Network from Chapter \ref{ch:cp18}),
  with a few optimizations based on other solvers. In Chapter \ref{ch:cp18} we show a top-down, divide-and-conquer
  algorithm for constructing 4-Odd-Even Selection Network. The difference in our new implementation is that
  we build our network in a bottom-up manner, which results in the easier and cleaner implementation.
  Experiments show that on many instances of popular benchmarks our technique outperforms other state-of-the-art PB-solvers.

  \item Finally, in the last chapter we summarize the results presented in this thesis and show the possibilities for future work.
\end{itemize}

The contributions mentioned above are based on several scientific papers which are mostly the joint work of the author and the supervisor
(\cite{karpinski2018encoding,karpinski2015smaller,karpinski2017encoding,karpinski2018multiway,pos18}).
Note that we do not explore the topic of SAT computation itself. 
Although the use of constantly improving SAT-solvers is an inseparable part
of our work, we focus solely on translating constraints into CNFs. The advantage of such approach is that our techniques
are not bound by the workings of a specific solver. As new, faster SAT-solvers are being produced, we can just swap
the one we use in order to get better results.

We would like to point out that our constructions are not optimal in the sense of computational complexity.
Our algorithms are based on classical sorting networks which use $O(n\log^2 n)$ comparators. From the point of view
of the $O$ notation, we do not expect to breach this barrier. We are aware of an optimal sorting network
which uses only $O(n\log n)$ comparators by the celebrated result of Ajtai, Koml{\'o}s and Szemer{\'e}di \cite{ajtai19830},
but the constants hidden in $O(n\log n)$ are so large, that CNF encodings based on such networks would never
be used in practice -- and this is a very important point for the Constraint Programming community. There has been many improvements
to the original $O(n\log n)$ sorting network \cite{paterson1990improved,pippenger1991selection,goodrich2014zig},
but at the time of writing of this thesis none is yet applicable to encoding of constraints.

\chapter{Preliminaries}\label{ch:pre}

    \def\nqueenssolution{Qd4, Qe2}
    \setchessboard{smallboard,labelleft=false,labelbottom=false,showmover=false,setpieces=\nqueenssolution}

    \begin{tikzpicture}[remember picture,overlay]
      \node[anchor=east,inner sep=0pt] at (current page text area.east|-0,3cm) {\chessboard};
    \end{tikzpicture}

  In this chapter we introduce definitions and notations used in the rest of the thesis.
  We take a special care in introducing comparator networks, as they are the central mechanism
  in all our encodings. To this end we present several different approaches to define comparator networks and
  explain the choices of comparator models we make for presenting our main results. As an example, we show
  different ways to construct a classic odd-even sorting network by Batcher \cite{batcher1968sorting}. Next, we present how
  to encode a comparator network into a set of clauses and how to make encodings of sorting (selection)
  networks enforce cardinality constraints. Then, we extend the model of comparator network so that
  the atomic operation does not handle only two inputs, but any fixed number of inputs. This way we
  construct networks which we call {\em Generalized Selection Networks} (GSNs). Finally, we show
  that any encoding of cardinality constraints based on GSNs preserves arc-consistency.

  \section{The Basics}

  Let $X$ denote a totally ordered set, for example the set of natural numbers $\nat$ or the set
  of binary values $\{0, 1\}$. We introduce the auxiliary "smallest" element
  $\bot \notin X$, such that for all $x \in X$ we have $\bot < x$.
  Thus, $X \cup \{\bot\}$ is totally ordered. The element $\bot$ is used in the later chapters
  to simplify presentation of algorithms.

  \begin{definition}[sequences]
  A sequence of length $n$, say $\bar{x} = \tuple{x_1, \dots, x_n}$, is an element of
  $X^n$. In particular, an element of $\{0, 1\}^n$ is called a {\em binary} sequence.
  Length of a sequence is denoted by $|\bar{x}|$.
  We say that a sequence $\bar{x} \in X^n$ is {\em sorted} if $x_i \geq x_{i+1}$, $1 \le i <
  n$. Given two sequences $\bar{x} = \tuple{x_1, \dots, x_n}$ and $\bar{y} = \tuple{y_1,
  \dots, y_m}$ we define several operations and notations:

  \begin{itemize}
    \item {\em concatenation} as $\bar{x} :: \bar{y} = \tuple{x_1, \dots, x_n, y_1, \dots, y_m}$,
    \item {\em domination} relation: $\bar{x} \succeq \bar{y} \Iff
          \forall_{i \in \{1,\dots,n\}} \forall_{j \in \{1,\dots,m\}} \; x_i \geq y_j$,
    \item {\em weak domination} relation (if $n=m$): $\bar{x} \succeq_{w} \bar{y} \Iff \forall_{i \in \{1,\dots,n\}} x_i \geq y_i$,
    \item $\bar{x}_{\odd} = \tuple{x_1, x_3, \dots}$, $\bar{x}_{\even} = \tuple{x_2, x_4, \dots}$,
    \item $\bar{x}_{a,\dots,b} = \tuple{x_a, \dots, x_b}$, $1 \le a \leq b \le n$,
    \item $\bar{x}_{\leftt} = \bar{x}_{1,\dots,\floor{n/2}}$, $\bar{x}_{\rightt} = \bar{x}_{\floor{n/2} + 1,\dots,n}$,
    \item {\em prefix/suffix} operators: $\pref(i,\bar{x}) ~=~ \bar{x}_{1,\dots,i}$ and
          $\suff(i,\bar{x}) ~=~ \bar{x}_{i,\dots,n}$, $1 \leq i \leq n$,
    \item the number of occurrences of a given value $b$ in $\bar{x}$ is denoted by $|\bar{x}|_b$,
    \item and the result of removing all occurrences of $b$ in $\bar{x}$ is written as $\remove(b,\bar{x})$.
  \end{itemize}
  \end{definition}

  \begin{definition}[top k sorted sequence]
    A sequence $\bar{x} \in X^n$ is top $k$ sorted, with $k \leq n$, if $\tuple{x_1,\dots,x_k}$ is sorted and
    $\tuple{x_1,\dots,x_k} \succeq \tuple{x_{k+1},\dots,x_n}$.
  \end{definition}

  \begin{definition}[bitonic sequence]
    A sequence $\bar{x} \in X^n$ is a bitonic sequence if $x_1 \leq \ldots \leq x_i \geq \ldots \geq x_n$ 
    for some $i$, where $1 \leq i \leq n$, or a circular shift of such sequence. We distinguish a special case of a bitonic sequence:

    \begin{itemize}
      \item {\em v-shaped}, if $x_1 \geq \ldots \geq x_i \leq \ldots \leq x_n$
    \end{itemize}
    
    \noindent and among v-shaped sequences there are two special cases:
    
    \begin{itemize}
      \item {\em non-decreasing}, if $x_1 \leq \ldots \leq x_n$,
      \item {\em non-increasing}, if $x_1 \geq \ldots \geq x_n$.
    \end{itemize}
    \label{def:bit}
  \end{definition}

  \begin{definition}[zip operator]
  For $m \ge 1$ given sequences (column vectors) $\bar{x}^i = \tuple{x^i_1,\dots,
  x^i_{n_i}}$, $1 \le i \le m$ and $n_1 \ge n_2 \ge \dots \ge n_m$, let us define the
  $\zip$ operation that outputs the elements of the vectors in row-major order:
    \[
    \zip(\bar{x}^1, \dots, \bar{x}^m) = \begin{cases}
    \bar{x}^1 & \mbox{if~} m = 1 \\ 
    \zip(\bar{x}^1, \dots, \bar{x}^{m-1}) & \mbox{if~} |\bar{x}^m| = 0 \\ 
    \tuple{x^1_1, x^2_1, \dots, x^m_1} ::
    \zip(\bar{x}^1_{2,\dots, n_1}, \dots, \bar{x}^m_{2, \dots, n_m}) & \mbox{otherwise}
    \end{cases}
  \]
  \end{definition}

  \section{Comparator Networks}
  
  \begin{figure}[t!]
  \centering
  \subfloat[]{
    \includegraphics[width=0.25\textwidth]{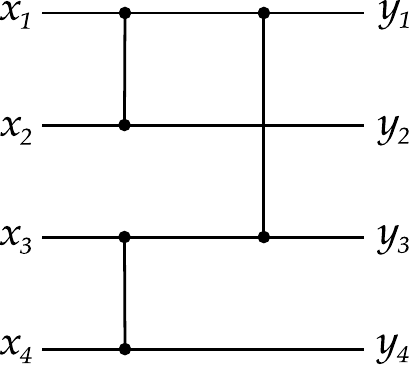}
  }
  ~~~~~~~~
  \subfloat[]{
    \centering
    \raisebox{20pt}{\includegraphics[width=0.3\textwidth]{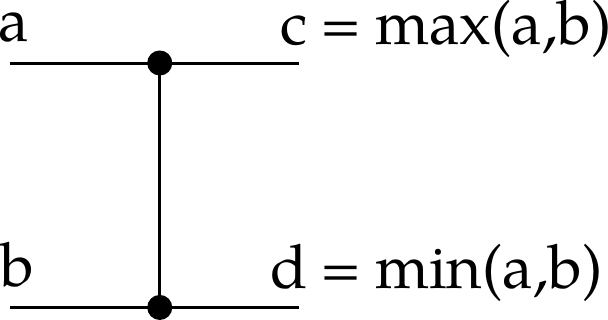}}
  }
  \caption{a) An example of comparator network; b) a single comparator}
  \label{fig:max}
  \end{figure}


  We construct and use comparator networks in this thesis. Traditionally comparator
  networks are presented as circuits that receive $n$ inputs and permute them using
  comparators (2-sorters) connected by "wires". Each comparator has two inputs and two
  outputs. The "upper" output is the maximum of inputs, and the "lower" one is the minimum.
  The standard definitions and properties of them can be found, for example, in
  \cite{knuth}. The only difference is that we assume that the output of any sorting
  operation or comparator is in a non-increasing order. We begin with presenting different ways
  to model comparator networks and explain strengths and weaknesses of such models.
  We use the network from Figure \ref{fig:max}a as a running example.

  \subsection{Functional Representation}

  In our first representation we model comparators as functions and comparator networks as a composition of
  comparators. This way networks can be represented in a clean, strict way. 

  \begin{definition}[comparator as a function]\label{def:comparatorf}
  Let $\bar{x} \in X^n$ and let $i,j \in \nat$, where $1\leq i < j \leq n$. A comparator is a function $c^n_{i,j}$ defined as:
  
  \[
    c^n_{i,j}(\bar{x})=\bar{y} \Iff y_i= \max\{x_i,x_j\} \wedge y_j= \min\{x_i,x_j\} \wedge \forall_{k \neq i,j} \; x_k=y_k
  \]
  \end{definition}

  \begin{examplebox}
  \begin{example}
    Notice that a comparator is defined as a function of the type $X^n \rightarrow X^n$, that is, it takes a sequence of length
    $n$ as an input and outputs the same sequence with at most one pair of elements swapped. For example, let $X=\{0,1\}$
    and $\bar{x}=\tuple{1,1,0,0,1}$. Then $c^5_{3,5}(\bar{x}) = \tuple{1,1,1,0,0}$.
  \end{example}
  \end{examplebox}

  \begin{definition}[comparator network as a composition of functions]
  We say that $f^n: X^n \rightarrow X^n$ is a comparator network of order $n$, if it can be represented as the composition of
  finite number of comparators, namely, $f^n=c^n_{i_1,j_1} \circ \cdots \circ c^n_{i_k,j_k}$.
  The number of comparators in a network is denoted by $|f^n|$.
  Comparator network of size $0$ is denoted by $id^n$.
  \end{definition}

  \begin{examplebox}
  \begin{example}
    Figure \ref{fig:max}a is an example of a simple comparator network consisting of 3 comparators.
    It outputs the maximum from 4 inputs on the top horizontal line,
    namely, $y_1=\max\{x_1,x_2,x_3,x_4\}$. Using our notation, we can
    say that this comparator network is constructed by composing three comparators:
    $\max^4 = c_{1,3}^4 \circ c_{3,4}^4 \circ c_{1,2}^4$.
  \end{example}
  \end{examplebox}

  Most comparator networks which are build for sorting (selection) purposes are presented using divide-and-conquer
  paradigm. In particular, the classic sorting networks are variations of the merge-sort algorithm.
  In such setup, to present comparator networks as composition of functions,
  we require that all smaller networks created using recursive calls be of the same order as the resulting one.
  In other words, each of them has to have same number of inputs. In 
  order to achieve that goal, we define a {\em rewire} operation. Let $[n]=\{1,\dots,n\}$, for any $n \in \nat$.
  Rewire is a map $\rho: [m] \rightarrow [n]$ ($m<n$) that is a monotone injection. 
  We use $\rho^*$ as a map from order $m$ comparator network to order $n$ comparator network,
  applying $\rho$ to each comparator in the network, that is,
  $\rho^*(c^m_{i_1,j_1} \circ \cdots \circ c^m_{i_k,j_k})=c^n_{\rho(i_1),\rho(j_1)} \circ \cdots \circ c^n_{\rho(i_k),\rho(j_k)}$.
  Some useful rewirings are presented below (rewirings of type $[n] \rightarrow [2n]$).
  Examples of how they work are shown in Figure \ref{fig:rewire}.

  \[
    left(f^n) = \rho^*_1(f^n) \quad \text{(where $\rho_1(i)=i$)}, \quad\quad right(f^n)= \rho^*_2(f^n) \quad \text{(where $\rho_2(i)=n+i$)}
  \]
  \[
    odd(f^n) = \rho^*_3(f^n) \quad \text{(where $\rho_3(i)=2i-1$)}, \quad\quad even(f^n) = \rho^*_4(f^n) \quad \text{(where $\rho_4(i)=2i$)}
  \]

  \begin{figure}[t!]
    \centering
    \subfloat[$left(max^4)$\label{fig:rewire_left}]{
      \includegraphics[width=0.3\textwidth]{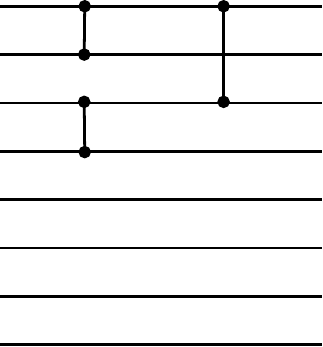}
    }
    ~
    \subfloat[$odd(max^4)$\label{fig:rewire_odd}]{
      \includegraphics[width=0.3\textwidth]{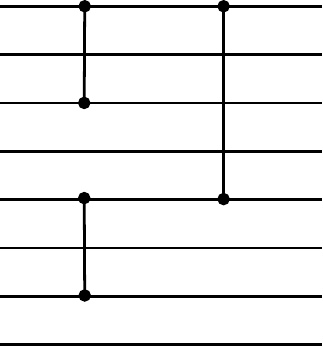}
    }
    ~
    \subfloat[$even(max^4)$\label{fig:rewire_even}]{
      \includegraphics[width=0.3\textwidth]{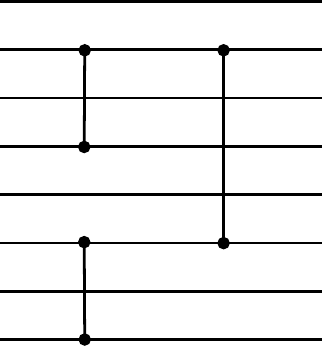}
    }
    \caption{Rewirings of comparator network $max^4$ to order $8$ comparator networks}
    \label{fig:rewire}
  \end{figure}

  \begin{definition}[k-selection network]
    A comparator network $sel^n_k$ is a {\em k-selection network} (of order $n$), if for each $\bar{x} \in X^n$, $sel^n_k(\bar{x})$ is top $k$ sorted.
  \end{definition}

  \noindent Notice that by the definition of a selection network, $sort^n = sel^n_n$ is a {\em sorting network} (of order $n$), that is,
  for each $\bar{x} \in X^n$, $sort^n(\bar{x})$ is sorted.

  \begin{examplebox}
  \begin{example}\label{ex:oe_func}
    Let us try to recreate the odd-even sorting network by Batcher \cite{batcher1968sorting} using the functional representation.
    To simplify the presentation we assume that $n$ is a power of $2$ and $X=\{0,1\}$.

    Construction of the {\em odd-even} sorting network uses the idea of merge-sort algorithm: in order to sort a list of $n$ elements, first
    partition the list into two lists (each of size $n/2$), recursively sort those lists, then merge the two sorted lists.
    Using functional representation it looks like this:

    \begin{align*}
      oe\_sort^1 &= id^1 \\
      oe\_sort^{n} &= merge^{n} \circ left(oe\_sort^{n/2}) \circ right(oe\_sort^{n/2})
    \end{align*}

    Network $merge^{n}$ merges two sorted sequences into one sorted sequence, that is: if $\bar{x}\in \bool^{n/2}$ and $\bar{y} \in \bool^{n/2}$
    are both sorted, then $merge^{n}(\bar{x}::\bar{y})$ is also sorted. Notice that when unfolding the recursion,
    the sorting network can be viewed as a composition of mergers, so we only need to specify how a single merger is constructed.
    In the {\em odd-even} approach, merger uses the idea of {\em balanced sequences}. Sequence $\bar{x} \in \bool^n$ is called balanced, if
    $\bar{x}_{\odd}$ and $\bar{x}_{\even}$ are sorted and $0 \leq |\bar{x}_{\odd}|_1 - |\bar{x}_{\even}|_1 \leq 2$. We define a
    balanced merger $bal\_merge^n$ that sorts a given balanced sequence. The balanced merger can be constructed in a straightforward way:

    \begin{align*}
      bal\_merge^2&=id^2 \\
      bal\_merge^n&=c^n_{2,3} \circ c^n_{4,5} \circ \cdots \circ c^n_{n-2,n-1}
    \end{align*}
    
    \noindent Finally, we get the merger for the odd-even sorting network:

    \begin{align*}
      oe\_merge^2 &= c^2_{1,2} \\
      oe\_merge^{2n} &= bal\_merge^{2n} \circ odd(oe\_merge^n) \circ even(oe\_merge^n)
    \end{align*}



    The number of comparators used in the odd-even sorting network is:

    \[
      |oe\_sort^n| = \frac{1}{4}n(\log n)(\log n - 1) + n - 1.
    \]

  \end{example}
  \end{examplebox}

  Functional representation allows for comparator networks to be presented in a formal, rigorous way.
  As seen in the example above, algorithms written in such form are very clean.
  One can also see the work of Codish and Zazon-Ivry \cite{zazonpairwise}
  to confirm that proofs of properties of such networks are -- in certain sense -- elegant. Unfortunately
  this is only true if the structure of the algorithm is simple, like in the odd-even sorting network.
  Networks presented in the later chapters are more complex and therefore need a more practical representation.

  \subsection{Declarative Representation}

  Another approach is to view comparators as relations on their inputs and outputs.
  Figure \ref{fig:max}b depicts a single comparator with inputs $a$ and $b$, and outputs $c$ and $d$.

  

  \begin{definition}[comparator as a relation]
  Let $a,b,c,d \in X$. A comparator is a relation defined as:
  
  \[
    \mathit{comp}(\tuple{a,b},\tuple{c,d}) \Iff c = \max(a, b) \wedge d = \min(a,b)
  \]
  \end{definition}

  \begin{definition}[comparator network as a relation]\label{def:net:decl}
  Let $\bar{x} \in X^n$, $\bar{y} \in X^m$ and $\bar{z} \in X^p$. We say that a relation $\mathit{net}(\bar{x},\bar{y},\bar{z})$ is a comparator network,
  if it is the conjunction of a finite number of comparators, namely,
  $\mathit{net}(\bar{x},\bar{y},\bar{z}) = \mathit{comp}(\bar{a}_1,\bar{b}_1) \wedge \dots \wedge \mathit{comp}(\bar{a}_k,\bar{b}_k)$,
  where elements from $\bar{x}$ appear exactly once in $\bar{a}_1 :: \dots :: \bar{a}_k :: \bar{y}$, elements from $\bar{y}$ appear exactly once in
  $\bar{b}_1 :: \dots :: \bar{b}_k :: \bar{x}$, and elements from $\bar{z}$ appear exactly once in $\bar{b}_1 :: \dots :: \bar{b}_k$ and at most once
  in $\bar{a}_1 :: \dots :: \bar{a}_k$.
  \end{definition}

  In the definition above, for readability we split the sequence of parameters into three groups: $\bar{x}$ are inputs of the network,
  $\bar{y}$ are outputs, and $\bar{z}$ are auxiliary elements, which can appear once as an input in some comparator and once as an output in some other comparator. 

  \begin{examplebox}
  \begin{example}
    In the declarative representation, a comparator network from Figure \ref{fig:max}a can be written as:

    \begin{align*}
      max&(\tuple{x_1,x_2,x_3,x_4}, \tuple{y_1,y_2,y_3,y_4}, \tuple{z_1,z_2}) = \\
      &= \mathit{comp}(\tuple{x_1,x_2},\tuple{z_1,y_2}) \wedge \mathit{comp}(\tuple{x_3,x_4},\tuple{z_2,y_4})
      \wedge \mathit{comp}(\tuple{z_1,z_2},\tuple{y_1,y_3})
    \end{align*} 
  \end{example}
  \end{examplebox}

  Definition \ref{def:net:decl} implies that comparator networks can be written in terms of propositional logic, where the only non-trivial
  component is a comparator. We can see the advantage over functional approach -- we do not need rewiring operations to specify networks.

  \begin{tcolorbox}[colback=white,colframe=black,sharp corners,enhanced jigsaw,breakable,break at=7cm/10cm]
  \begin{example}
    Let us again see the construction of the odd-even sorting network, this time using declarative representation. The entire
    algorithm can be written as a relation $oe\_sort$ described below. To simplify the presentation we assume that $n$
    is a power of $2$.

    \begin{align*}
      oe\_sort&(\tuple{x},\tuple{x}, \tuple{}), \\
      oe\_sort&(\tuple{x_1,\dots,x_{2n}},\tuple{y_1,\dots,y_{2n}},\tuple{z_1,\dots,z_{2n}} :: S_1 :: S_2 :: M) \; = \\
        &oe\_sort(\tuple{x_1,\dots,x_{n}},\tuple{z_1,\dots,z_{n}}, S_1) \; \wedge \\
        &oe\_sort(\tuple{x_{n+1},\dots,x_{2n}},\tuple{z_{n+1},\dots,z_{2n}}, S_2) \; \wedge \\
        &oe\_merge(\tuple{z_1,\dots,z_{n},z_{n+1},\dots,z_{2n}},\tuple{y_1,\dots,y_{2n}}, M).
    \end{align*}

    \begin{align*}
      oe\_merge&(\tuple{x_1,x_2},\tuple{y_1,y_2}, \tuple{}) = \mathit{comp}(\tuple{x_1,x_2},\tuple{y_1,y_2}), \\
      oe\_merge&(\tuple{x_1,\dots,x_{2n}},\tuple{y_1,\dots,y_{2n}},\tuple{z_1,z'_1,\dots,z_{n},z'_n} :: M_1 :: M_2) \; = \\
        &oe\_merge(\tuple{x_1,x_3,\dots,x_{2n-1}},\tuple{z_1,\dots,z_{n}}, M_1) \; \wedge \\
        &oe\_merge(\tuple{x_{2},x_4,\dots,x_{2n}},\tuple{z'_{1},\dots,z'_{n}}, M_2) \; \wedge \\
        &bal\_merge(\tuple{z_1,z'_1,\dots,z_{n},z'_{n}}, \tuple{y_1,\dots,y_{2n}}).
    \end{align*}

    \begin{align*}
      bal\_merge&(\tuple{x_1,x_{2}}, \tuple{x_1,x_{2}}), \\
      bal\_merge&(\tuple{x_1,\dots,x_{n}}, \tuple{x_1} :: \tuple{y_2,\dots,y_{n-1}} :: \tuple{x_{n}}) \; = \\
        &\mathit{comp}(\tuple{x_2,x_3},\tuple{y_2,y_3}) \wedge \dots \wedge \mathit{comp}(\tuple{x_{n-2},x_{n-1}},\tuple{y_{n-2},y_{n-1}}).
    \end{align*}

    The semantics of $sort(\bar{x},\bar{y},\bar{z})$ is that $\bar{y}$ is sorted and is a permutation of $\bar{x}$,
    while $\bar{z}$ is a sequence of auxiliary elements.
  \end{example}
  \end{tcolorbox}

  Networks given so far have been presented in a formal, rigorous way, but as we will see in further chapters,
  it is easier to reason about and prove properties of comparator networks when presented as pseudo-code or
  a list of procedures. We remark that it is possible to restate all the algorithms presented in this thesis
  in the form of composition of functions or as relations, but this would vastly complicate most of the proofs, which are
  already sufficiently formal.

  \subsection{Procedural Representation}

  Although functional and declarative approaches has been used in the context of encoding cardinality constraints
  (functional in \cite{zazonpairwise}, declarative in \cite{asin2009cardinality}), they are not very practical,
  as modern solvers are written in general-purpose, imperative, object-oriented programming languages like Java or C++.
  In order to make the implementation of the algorithms presented in this thesis easier (and to simplify the proofs),
  we propose the procedural representation of comparator networks \cite{karpinski2017encoding,karpinski2018multiway,pos18},
  in which a network is an algorithm written in pseudo-code with a restricted set of operations.
  What we want, for a given network, is a procedure that generates its declarative representation, i.e.,
  the procedure should define a set of comparators. At the same time we want to treat such procedures as oblivious sorting
  algorithms, so that we can easily prove their correctness. We explain our approach with the running example.

  \begin{algorithm}[ht!]
    \caption{$max^4$}\label{ex:net:max}
    \begin{algorithmic}[1]
      \Require {$\tuple{x_1,x_2,x_3,x_4} \in X^4$}
      \Ensure{The output is top 1 sorted and is a permutation of the inputs}
      \State $\bar{z} \gets sort^2(x_1,x_2)$
      \State $\bar{z'} \gets sort^2(x_3,x_4)$
      \State \Return $sort^2(z_1,z'_1) :: \tuple{z_2,z'_2}$
    \end{algorithmic}
  \end{algorithm}

  \begin{examplebox}
  \begin{example}
    A comparator network from Figure \ref{fig:max}a can be expressed as pseudo-code like in Algorithm \ref{ex:net:max}.
    The algorithm showcases several key properties of our representation:

    \begin{itemize}
      \item The only allowed operation that compares elements is $sort^2$, which puts two given elements in non-increasing order.
      \item New sequences can be defined and assigned values, for example, as a result of $sort^2$ or a sub-procedure.
      \item The algorithm returns a sequence that is a permutation of the input sequence in order to highlight the fact
        that it represents a comparator network. For example, in Algorithm \ref{ex:net:max} we concatenate $\tuple{z_2,z'_2}$
        to the returned sequence.
    \end{itemize}
  \end{example}
  \end{examplebox}

  \begin{algorithm}[ht!]
    \caption{$oe\_sort^{n}$}\label{ex:net:oddeven}
    \begin{algorithmic}[1]
      \Require {$\bar{x} \in \bool^{n}$; $n$ is a power of $2$}
      \Ensure{The output is sorted and is a permutation of the inputs}
      \If {$n=1$}
        \Return $\bar{x}$
      \EndIf
      \State $\bar{y} \gets oe\_sort^{n/2}(\bar{x}_{\leftt})$
      \State $\bar{y'} \gets oe\_sort^{n/2}(\bar{x}_{\rightt})$
      \State \Return $oe\_merge^n(\bar{y},\bar{y}')$
    \end{algorithmic}
  \end{algorithm}

  \begin{algorithm}[ht!]
    \caption{$oe\_merge^{n}$}\label{ex:net:oemerge}
    \begin{algorithmic}[1]
      \Require {$\bar{x}^1,\bar{x}^2 \in \bool^{n/2}$; $\bar{x}^1$ and $\bar{x}^2$ are sorted;
        $n$ is a power of $2$}
      \Ensure{The output is sorted and is a permutation of the inputs}
      \If {$n=2$}
        \Return $sort^2(x^1_1,x^2_1)$
      \EndIf
      \State $\bar{y}  \gets oe\_merge^{n/2}(\bar{x}^1_{\odd},\bar{x}^2_{\odd})$
      \State $\bar{y}' \gets oe\_merge^{n/2}(\bar{x}^1_{\even},\bar{x}^2_{\even})$
      \State $z_1 \gets y_1$; $z_n \gets y'_n$
      \ForAll {$i \in \{1, \dots, n/2-1\}$}
        $\tuple{z_{2i},z_{2i+1}} \gets sort^2(y'_{i},y_{i+1})$
      \EndFor
      \State \Return $\bar{z}$
    \end{algorithmic}
  \end{algorithm}

  \begin{examplebox}
  \begin{example}
    The odd-even sorting network is presented in Algorithm \ref{ex:net:oddeven}.
    It uses Algorithm \ref{ex:net:oemerge} -- the merger -- as a sub-procedure.
    From this example we can also see that the code convention allows for loops, conditional statements and recursive calls,
    but they cannot depend on the results of comparisons between elements in sequences.
  \end{example}
  \end{examplebox}

  \section{Encoding Cardinality Constraints}

  What we want to achieve using comparator networks, is to produce a clausal encoding for a given
  cardinality constraint. Notice that a cardinality constraint in Definition \ref{def:cl_encoding}
  is defined in terms of "$\leq$" relation. We now show that it is in fact the only type of relation we need to be concerned about.

  \begin{observation}\label{obs:card}
  Let $k,n \in \nat$ where $k \leq n$, and let $\tuple{x_1,\dots,x_n}$ be a sequence of Boolean literals. Then:

  \begin{enumerate}
    \item $x_1 + x_2 + \dots + x_n \geq k$ is eqivalent to $\neg x_1 + \neg x_2 + \dots + \neg x_n \leq n-k$,
    \item $x_1 + x_2 + \dots + x_n > k$ is eqivalent to $\neg x_1 + \neg x_2 + \dots + \neg x_n \leq n-k-1$,
    \item $x_1 + x_2 + \dots + x_n < k$ is equvalent to $x_1 + x_2 + \dots + x_n \leq k-1$,
    \item $x_1 + x_2 + \dots + x_n = k$ is equvalent to $x_1 + x_2 + \dots + x_n \leq k$
      and $\neg x_1 + \neg x_2 + \dots + \neg x_n \leq n-k$.
  \end{enumerate}
  \end{observation}

  \noindent As we can see, in every case we can reduce the cardinality constraint to the equivalent
  one which uses the "$\leq$" relation. In case of equality relation this produces two cardinality constraints,
  but then we handle them (encode them) separately. Therefore from now on, when we mention a cardinality constraint,
  we mean the one in the form:

  \[
    x_1 \;\; + \;\; x_2 \;\; + \;\; \dots \;\; +  \;\; x_n  \;\; \leq \;\; k.
  \]

  We now describe how to translate a cardinality constraint into equisatisfiable set of clauses.
  We begin with an encoding of a single comparator. If we use the notation as in Figure \ref{fig:max}b,
  then the maximum on the upper output is translated into a disjunction of the inputs ($a \vee b$) and the
  minimum on the lower output is translated into a conjunction of the inputs ($a \wedge b$). Therefore a single
  comparator with inputs $\tuple{a,b}$ and outputs $\tuple{c,d}$ can be encoded using the formula:

  \[
    (c \iff a \vee b ) \wedge (d \iff a \wedge b),
  \]

  \noindent which is equivalent to the following six clauses:

  \begin{center}
  \begin{tabular}{c@{\hskip 1in}c@{\hskip 1in}c}
  {$\!\begin{aligned}
    a &\Rightarrow c \\
    b &\Rightarrow c
  \end{aligned}$}
  &
  {$\!\begin{aligned}
    a &\wedge b \Rightarrow d \\
    c &\Rightarrow a \vee b
  \end{aligned}$}
  &
  {$\!\begin{aligned}
    d &\Rightarrow a \\
    d &\Rightarrow b
  \end{aligned}$}
  \end{tabular}
  \end{center}

  Encoding consists of translating every comparator into a set of clauses.
  Thus different representations of a comparator network in the previous section can be viewed as
  different procedures that outputs the same set of comparators. Since the cardinality constraints are over Boolean domain,
  we are using comparator networks in the context of Boolean formulas, therefore
  we limit the domain of the inputs to 0-1 values.

  The clausal encoding of cardinality constraints is defined as follows. First, we build
  a sorting network with inputs $\tuple{x_1,\ldots,x_n}$ and outputs $\tuple{y_1,\ldots,y_n}$.
  Since it is a sorting network, if exactly $k$ inputs are set to $1$,
  we will have $\tuple{y_1,\ldots,y_k}$ set to $1$ and $\tuple{y_{k+1},\ldots,y_n}$ set to $0$. Therefore, in order to enforce
  the constraint $x_1 + x_2 + \dots + x_n \leq k$, we need to do two things: translate each comparator into a set of clauses, and 
  add a unit clause $\neg y_{k+1}$, which virtually sets $y_{k+1}$ to $0$. This way the resulting CNF is satisfiable if and only if
  at most $k$ input variables are set to $1$.

  Two improvements to this basic encoding can be made. Notice that we do not need to
  sort the entire input sequence, but only the first $k+1$ largest elements. Therefore rather than
  using sorting networks, we can use {\em selection networks}.

  \begin{figure}[ht!]
    \centering
    \subfloat[]{
      \includegraphics[width=0.45\textwidth]{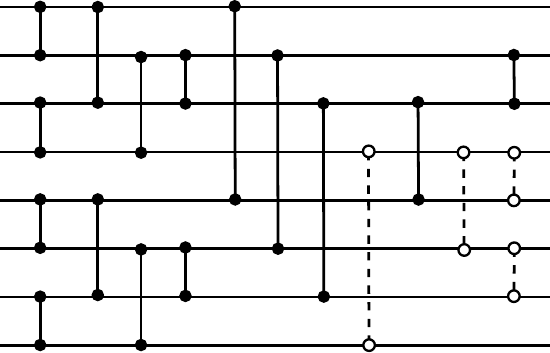}
    }
    ~
    \subfloat[]{
      \includegraphics[width=0.45\textwidth]{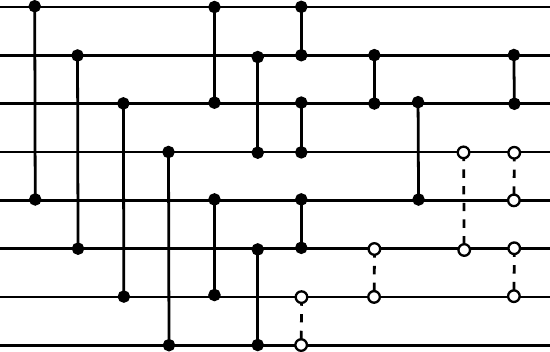}
    }
    \caption{a) Odd-Even Selection Network; b) Pairwise Selection Network; $n=8$, $k=3$}
    \label{fig:oe-vs-pw}
  \end{figure}

  \begin{examplebox}
  \begin{example}
    The first selection network used in the context of encoding cardinality constraints is called {\em Pairwise Selection Network} by
    Codish and Zazon-Ivry \cite{zazonpairwise}. They construct selection networks recursively, just like sorting networks: in order
    to get top $k$ sorted sequence, first we split input in half and get top $k$ sorted sequences for both parts,
    then we merge the results. In reality, the only thing we are actually doing is cutting out "unnecessary"
    comparators from sorting networks. In Figure \ref{fig:oe-vs-pw} we present odd-even and pairwise selection networks
    for $n=8$ and $k=3$. Removed comparators from sorting networks are marked with dashed lines. Notice that
    with pairwise approach we save 5 comparators, where with odd-even it is only 4.
  \end{example}
  \end{examplebox}
  
  Another improvement lies in the encoding of a single comparator. In the encoding of $x_1 + x_2 + \dots + x_n \leq k$
  one can use 3 out of 6 clauses to encode a single comparator:

  \begin{equation}\label{eq:3c}
    a \Rightarrow c, \quad b \Rightarrow c, \quad a \wedge b \Rightarrow d.
  \end{equation}
 
  \noindent Notice that this encoding only guarantees that if at least $i \in \{0,1,2\}$ inputs in a single comparator are set to 1,
  then at least $i$ top outputs must be set to 1. This generalizes to the entire network, that is, if at least $k$ inputs of a selection network are set to 1, then
  at least $k$ top outputs must be set to $1$. Although this set of clauses are not equivalent to the comparator
  of Definition \ref{def:comparatorf}, it is enough to be used in the encoding of $x_1 + x_2 + \dots + x_n \leq k$
  while still maintaining the arc-consistency (Definition \ref{def:gac}). Thanks to that, we can use half as many clauses to
  encode a network, which significantly reduces the size of an instance passed to a SAT-solver.

  \subsection{Generalized Selection Networks}\label{sec:standard}

  The encoding of cardinality constraints using comparator networks has been known for some time now.
  In \cite{abio2013parametric,asin2011cardinality,codish2010pairwise,minisatp} authors are using sorting (selection) networks to encode
  cardinality constraints in the same way as we: inputs and outputs of a comparator are Boolean variables
  and comparators are encoded as a CNF formula. In addition, the $(k+1)$-th greatest output
  variable $y_{k+1}$ of the network is forced to be $0$ by adding $\neg y_{k+1}$ as a clause to the
  formula that encodes $x_1 + \dots + x_n \leq k$. The novelty in most of our constructions is that
  rather than using simple comparators (2-sorters), we also use comparators of higher order
  as building blocks ($m$-sorters, for $m\geq 2$). Since a sorter is just a special case of a selector,
  so we only need one definition. For a selector we want to output top $m$ sorted elements from the inputs.
  The following definition captures this notion.

  \begin{definition}[m-selector of order n]\label{def:msel}
    Let $n,m \in \nat$, where $m \leq n$. Let $\bar{x}=\tuple{x_1,\dots,x_n}$ and $\bar{y}=\tuple{y_1,\dots,y_m}$
    be sequences of Boolean variables. The Boolean formula $s^n_m(\bar{x},\bar{y})$ which consists of the set of clauses
    $\{x_{i_1} \wedge \dots \wedge x_{i_p} \Rightarrow y_p \, : \, 1 \leq p \leq m, 1 \leq i_1 < \dots < i_p \leq n\}$
    is an {\em $m$-selector} of order $n$.
  \end{definition}

  Notice that we are identifying a selector with its clausal encoding. We do this often in this thesis, in non-ambiguous
  context. Observe that a $2$-selector of order $2$ gives the same set of clauses as in Eq. \ref{eq:3c}. In fact, Definition \ref{def:msel}
  is a natural generalization of the encoding of a single comparator. 

  \begin{examplebox}
  \begin{example}\label{ex:a}
    We would like to encode the network that selects maximum out of four inputs into a set of clauses.
    We can use the network from Figure \ref{fig:max}a to do it. If we name the input variables of
    the longer comparator as $\{z_1,z_2\}$, then the entire network can be encoded by encoding each 2-sorter separately.
    This produces the clause set $\{x_1 \impl z_1, x_2 \impl z_1, x_1 \wedge x_2 \impl y_2\} \cup
    \{x_3 \impl z_2, x_4 \impl z_2, x_3 \wedge x_4 \impl y_4\} \cup
    \{z_1 \impl y_1, z_2 \impl y_1, z_1 \wedge z_2 \impl y_3\}$. This approach uses 6 auxiliary
    variables (not counting $x_i$'s) and 9 clauses. Another way to encode the same network is to
    simply use a single $1$-selector of order $4$. This gives the clause set
    $\{x_1 \impl y_1, x_2 \impl y_1, x_3 \impl y_1, x_4 \impl y_1\}$, where we only need 1 additional
    variable and 4 clauses. Notice that to achieve $y_1=\max\{x_1,x_2,x_3,x_4\}$
    we are only interested in the value of the top output variable,
    therefore we do not need to assert other output variables.
  \end{example}
  \end{examplebox}

  Informally, {\em Generalized Selection Networks} (GSNs) are selection networks
  that use selectors as building blocks. For example, in declarative representation (Definition \ref{def:net:decl})
  one would substitute the relation $\mathit{comp}(\tuple{x_1,x_2},\tuple{z_1,y_2})$
  for a more general $\mathit{sel}(\tuple{x_1,\dots,x_n},\tuple{y_1,\dots,y_m})$ which
  is true if and only if $\tuple{y_1,\dots,y_m}$ is sorted and contains $m$ largest elements from $\tuple{x_1,\dots,x_n}$.
  In procedural representation we generalize $sort^2$ operation to $sel^n_m$ which outputs $m$ sorted, largest elements from $n$ inputs.
  In this context, $sort^n$ is the same operation as $sel^n_n$. In the end, those modifications are the means to obtain a set of selectors (for a given network),
  which then we encode as in Definition \ref{def:msel}.

  Encodings of cardinality constraints using (generalized) selection networks where each selector is encoded
  as described in Definition \ref{def:msel} and additional clause $\neg y_{k+1}$ is added are said to be
  encoded in a {\em standard} way.


  \subsection{Arc-Consistency of The Standard Encoding} 

  We formally define the notion of arc-consistency. A partial assignment $\sigma$ is consistent with a CNF formula
  $\phi = \{C_1,\dots,C_m\}$, if for each $1 \leq i \leq m$, $\sigma(C_i)$ is either true or undefined.

  \begin{definition}[unit propagation]\label{def:up}
    Unit propagation (UP) is a process that extends a partial, consistent assignment $\sigma$ of some CNF formula $\phi$ into
    a partial assignment $\sigma'$ using the following rule repeatedly until reaching a fix point: if there
    exists a clause $(l \vee l_1 \vee \dots \vee l_k)$ in $\phi$ where $l$ is undefined and either $k=0$
    or $l_1,\dots,l_k$ are fixed to $0$, then extend $\sigma$ by fixing $\sigma(l)=1$.
  \end{definition}

  \begin{definition}[arc-consistency]\label{def:gac}
    Let $n,k \in \nat$ and let $\bar{x}=\{x_1,\dots,x_n\}$ be a set of propositional literals. Encoding $\phi$ of 
    the cardinality constraint $x_1 + x_2 + \dots + x_n \leq k$ is {\em arc-consistent} if the two following conditions hold:

    \begin{itemize}
      \item in any partial assignment consistent with $\phi$, at most $k$ propositional variables from $\bar{x}$ are assigned to $1$, and
      \item in any partial assignment consistent with $\phi$, if exactly $k$ variables from $\bar{x}$ are assigned to $1$, then all the other
      variables occurring in $\bar{x}$ must be assigned to $0$ by unit propagation.
    \end{itemize}
  \end{definition}
  
  We already mentioned in the previous chapter that arc-consistency property guarantees better propagation
  of values using unit propagation and that this has a positive impact on the practical efficiency of SAT-solvers.
  Here we show our first result of this thesis, that is, we prove that any encoding based on the standard
  encoding of GSNs preserves arc-consistency. Our proof is the generalization the proof of arc-consistency for selection networks
  \cite{karpinski2015smaller}. For the sake of the proof we give a precise, clausal definition
  of a GSN, which will help us formally prove the main theorem. We define networks as a sequence of {\em layers} and
  a layer as a sequence of selectors. But first, we prove two technical lemmas regarding selectors.
  We introduce the convention, that $\tuple{x_1,\ldots,x_n}$ will denote the input and $\tuple{y_1,\ldots,y_m}$ will denote
  the output of some order $n$ comparator network (or GSN). We would also like to view them as
  sequences of Boolean variables, that can be set to either true ($1$), false ($0$) or undefined.

  The following lemma shows how 0-1 values propagates through a single selector.

  \begin{lemma}\label{lma:pp}
    Let $m \leq n$. A single $m$-selector of order $n$, say $s^n_m$, with inputs $\tuple{x_1,\ldots,x_n}$
    and outputs $\tuple{y_1,\ldots,y_m}$ has the following {\em propagation properties}, for any partial assignment $\sigma$
    consistent with $s^n_m$:

    \begin{enumerate}
      \item If $k$ input variables ($1 \leq k \leq n$) are set to $1$ in $\sigma$, then UP sets all $y_1,\dots,y_{min(k,m)}$ to $1$.
      \item If $y_k=0$ (for $1 \leq k \leq m$) and exactly $k-1$ input variables are set to $1$ in $\sigma$,
        then UP sets all the rest input variables to $0$.
    \end{enumerate}
  \end{lemma}

  \begin{proof}
    Let $\pi \! : \! \{1,\dots,k\} \rightarrow \{1,\dots,n\}$ be
    a 1-1 function such that input variables $x_{\pi(1)},\dots,x_{\pi(k)}$ are set to $1$.
    The following clauses exist: $(x_{\pi(1)} \Rightarrow y_1)$, $(x_{\pi(1)} \wedge x_{\pi(2)} \Rightarrow y_2)$,
    $\dots$, $(x_{\pi(1)} \wedge \dots \wedge x_{\pi(\min(k,m))} \Rightarrow y_{\min(k,m)})$.
    Therefore UP will set all $y_1,\dots,y_{min(k,m)}$ to $1$.

    Now assume that $y_k=0$ and let $\pi \, : \, \{1,\dots,k-1\} \rightarrow \{1,\dots,n\}$ be a 1-1 function
    such that input variables $I = \{x_{\pi(1)},\dots,x_{\pi(k-1)}\}$ are set to $1$. For each input variable
    $u \not\in I$ there exist a clause $u \wedge x_{\pi(1)} \wedge \dots \wedge x_{\pi(k-1)} \Rightarrow y_k$.
    Therefore $u$ will be set to $0$ by UP.
  \end{proof}

  We say that a formula $\phi$ can be reduced to $\phi'$ by a partial assignment $\sigma$,
  if we do the following procedure repeatedly, until reaching the fix point:
  take clause $C$ (after applying assignment $\sigma$) of $\phi$. If it is satisfied, then remove it. If it is of the form
  $C = (0 \vee \psi)$, then exchange it to $C' = \psi$. Otherwise, do nothing. Notice that $\phi$ is satisfied iff $\phi'$ is satisfied.

  \begin{lemma}\label{lma:red_s}
    Let $n,m \in \nat$ and let $s^{n}_{m}$ be an $m$-selector of order $n$ with inputs $\tuple{x_1,\dots,x_{n}}$ and outputs $\tuple{y_1,\dots,y_{m}}$.
    Let $1 \leq i \leq n$. If we set $x_i=1$, then $s^{n}_{m}$ can be reduced to $s^{n-1}_{m-1}$ by the partial assignment used in unit propagation, where
    $s^{n-1}_{m-1}$ is an $(m-1)$-selector of order $n-1$.
  \end{lemma}

  \begin{proof}
    By Lemma \ref{lma:pp}, UP sets $y_1=1$. Then each clause which contains $x_i$
    can be reduced: $(x_i \Rightarrow y_1)$ is satisfied and can be removed; $(x_i \wedge x_q \Rightarrow y_2)$ can be reduced
    to $(x_q \Rightarrow y_2)$, for each $q \neq i$; and in general, let $Q$ be any subset of $\{1,\dots,n\} \setminus \{i\}$,
    then clause $(x_i \wedge (\bigwedge_{q \in Q} x_q) \Rightarrow y_{|Q|+1})$ can be reduced to $(\bigwedge_{q \in Q} x_q \Rightarrow y_{|Q|+1})$.
    The remaining clauses form an $(m-1)$-selector of order $n-1$.
  \end{proof}
  
  \begin{definition}[layer]\label{def:layer}
    Let $r \in \nat$, $n_1,m_1,\dots,n_r,m_r \in \nat$, $S = \tuple{s^{n_1}_{m_1}(\bar{x}^1,\bar{y}^1),\dots,s^{n_r}_{m_r}(\bar{x}^r,\bar{y}^r)}$
    and let $\bar{x}=\tuple{x_1,\dots,x_n}$, $\bar{y}=\tuple{y_1,\dots,y_m}$, $\bar{x}^i=\tuple{x^i_1,\dots,x^i_{n_i}}$, $\bar{y}^i=\tuple{y^i_1,\dots,y^i_{m_i}}$
    be sequences of Boolean variables, for $1 \leq i \leq r$. A Boolean formula $L^{(n,m)}(\bar{x}, \bar{y}, S) = \bigwedge s^{n_i}_{m_i}(\bar{x}^i,\bar{y}^i)$
    is a {\em layer} of order $(n,m)$, for some $n,m \in \nat$, if:

    \begin{enumerate}
      \item $n = \sum n_i$, $m = \sum m_i$,
      \item for each $1 \leq i \leq r$, $s^{n_i}_{m_i}(\bar{x}^i,\bar{y}^i)$ is an $m_i$-selector of order $n_i$,
      \item for each $1 \leq i < j \leq r$, $\bar{x}^i$ and $\bar{x}^j$ are disjoint; $\bar{y}^i$ and $\bar{y}^j$ are disjoint,
      \item for each $1 \leq i \leq r$, $\bar{x}^i$ is a subsequence of $\bar{x}$ and $\bar{y}^i$ is a subsequence of $\bar{y}$,
    \end{enumerate}
  \end{definition}
  
  \begin{definition}[generalized network]\label{def:gn}
    Let $n,m \in \nat$, where $m \leq n$, and let $n_1,m_1,\dots,n_r,m_r \in \nat$. Let $\bar{x}=\tuple{x_1,\dots,x_n}$,
    $\bar{y}=\tuple{y_1,\dots,y_m}$, $\bar{x}^i=\tuple{x^i_1,\dots,x^i_{n_i}}$, $\bar{y}^i=\tuple{y^i_1,\dots,y^i_{m_i}}$
    be sequences of Boolean variables, and let $S_i$ be a sequence of selectors, for $1 \leq i \leq r$.
    Let $L = \tuple{L^{(n_1,m_1)}_1(\bar{x}^1, \bar{y}^1, S_1), \dots,$ $L^{(n_r,m_r)}_r(\bar{x}^r, \bar{y}^r, S_r)}$.
    A Boolean formula $f^n_m(\bar{x},\bar{y},L) = \bigwedge L^{(n_i,m_i)}_i$ is a {\em generalized network} of order $(n,m)$, if:

    \begin{enumerate}
      \item $\bar{x}=\bar{x}^1$, $\bar{y}=\bar{y}^r$, and $\bar{y}^i=\bar{x}^{i+1}$, for $1 \leq i < r$,
      \item for each $1 \leq i \leq r$, $L^{(n_i,m_i)}(\bar{x}^i, \bar{y}^i, S_i)$ is a layer of order $(n_i,m_i)$,
    \end{enumerate}
  \end{definition}
  
  \begin{definition}[generalized selection network]\label{def:gsn}
    Let $n,m \in \nat$, where $m \leq n$. Let $\bar{x}=\tuple{x_1,\dots,x_n}$ and $\bar{y}=\tuple{y_1,\dots,y_m}$
    be sequences of Boolean variables and let $L$ be a sequence of layers.
    A generalized network $f^n_m(\bar{x},\bar{y},L)$ is a {\em generalized selection network} of order $(n,m)$,
    if for each 0-1 assignment $\sigma$ that satisfies $f^n_m$, $\sigma(\bar{y})$ is sorted and $|\sigma(\bar{y})|_1 \geq \min(|\sigma(\bar{x})|_1,m)$.
  \end{definition}

  To encode a cardinality constraint $x_1 + \dots + x_n \leq k$ we build a GSN $f^n_{k+1}(\bar{x},\bar{y},L)$
  and we add to it a singleton clause $\neg y_{k+1}$ which practically
  sets the variable $y_{k+1}$ to false. Such encoding we will call a {\em standard encoding}.
  
  \begin{definition}[standard encoding]\label{def:se}
    Let $k,n \in \nat$, where $k \leq n$ and let $\bar{x}=\tuple{x_1,\dots,x_n}$
    and $\bar{y}=\tuple{y_1,\dots,y_{k+1}}$ be sequences of Boolean variables. A Boolean formula
    $\phi^n_k(\bar{x},\bar{y},L) = f^n_{k+1}(\bar{x},\bar{y},L) \wedge \neg y_{k+1}$
    is a {\em standard encoding} of the constraint $x_1 + \dots + x_n \leq k$, if
    $f^n_{k+1}(\bar{x},\bar{y},L)$ is a generalized selection network of order $(n,k+1)$ (for some sequence $L$ of layers).
  \end{definition}

  We use the following convention regarding equivalences of Boolean variables:
  the symbol ``$=$'' is used as an assignment operator or value equivalence.
  We use the symbol ``$\equiv$'' as equivalence of variables, that is, if $a \equiv b$ then $a$ is an alias for $b$ (and vice-versa).
  We define $V[\phi]$ as the set of Boolean variables in formula $\phi$.

  Let $n,k,r \in \nat$, $n_1,k_1,\dots,n_r,k_r \in \nat$. Let $\bar{x}=\tuple{x_1,\dots,x_n}$,
  $\bar{y}=\tuple{y_1,\dots,y_{k+1}}$, $\bar{x}^i=\tuple{x^i_1,\dots,x^i_{n_i}}$, $\bar{y}^i=\tuple{y^i_1,\dots,y^i_{k_i}}$
  be sequences of Boolean variables, for $1 \leq i \leq r$. For the rest of the chapter assume
  that $f^n_{k+1}(\bar{x},\bar{y},L)$ is a GSN of order $(n,k+1)$, where
  $L = \tuple{L^{(n_1,k_1)}_1(\bar{x}^1, \bar{y}^1, S_1), \dots, L^{(n_r,k_r)}_r(\bar{x}^r, \bar{y}^r, S_r)}$ is
  a sequence of layers, $S_i$ is a sequence of selectors (for $1 \leq i \leq r$).
  
  We will now define a notion that captures a structure of propagation through the network for a single variable.

  \begin{definition}[path]
    A {\em path} is a sequence of Boolean variables $\tuple{z_1,\dots,z_p}$ such that $\forall_{1\leq i \leq p}$
    $z_i \in V[f^n_{k+1}]$ and for all $1 \leq i < p$ there exists an $m$-selector of order $n'$ in $f^n_{k+1}$ with inputs
    $\tuple{a_1,\dots,a_{n'}}$ and outputs $\tuple{b_1,\dots,b_m}$ such that $z_i \in \tuple{a_1,\dots,a_{n'}}$ and
    $z_{i+1} \in \tuple{b_1,\dots,b_m}$.
  \end{definition}

  \begin{definition}[propagation path]
    Let $x$ be an undefined variable. A path $\bar{z}_x = \tuple{z_1,\dots,z_p}$ is a {\em propagation path}, if
    $z_1 \equiv x$ and $p$ is the largest integer such that $\tuple{z_2,\dots,z_p}$
    is a sequence of variables that would be set to $1$ by UP, if we set $z_1=1$.
  \end{definition}

  \begin{lemma}\label{lma:pl}
    Let $1 \leq i \leq n$ and let $z_{x_i}=\tuple{z_1,\dots,z_p}$ be the propagation path, for some $p \in \nat$.
    Then (i) each $z_j$ is an input to a layer $L^{(n_j,k_j)}$, for $1 \leq j < r$, and (ii) $z_p \equiv z_r \equiv y_1$.
  \end{lemma}

  \begin{proof}
    We prove (i) by induction on $j$. If $j=1$ then $z_1 \equiv x_i$ is the input to the first layer,
    by the definition of generalized network. Take any $j \geq 1$ and assume that $z_j$ is an input to a layer $L^{(n_j,k_j)}$.
    From the definition of a layer, inputs of selectors of $L^{(n_j,k_j)}$ are disjoint, therefore $z_j$ is an input to a unique
    selector $s^{n'}_{m'} \in S_j$, for some $n' \geq m'$. By Lemma \ref{lma:pp}, if $z_j=1$ then $z_{j+1}$ -- the output of $s^{n'}_{m'}$ -- is set to $1$.
    Since $j<r$, $z_{j+1}$ is an input to layer $L^{(n_{j+1},k_{j+1})}$, by the definition of generalized network. This ends the inductive step, therefore (i) is true.

    Using (i) we know that $z_{r-1}$ is an input to layer $L^{(n_{r-1},k_{r-1})}$. Using similar argument as in the inductive step in previous paragraph, we conclude that
    $z_r \equiv y_t$, for some $1 \leq t \leq k+1$. Let $\sigma$ be a partial assignment which fixed variables $\tuple{z_1,\dots,z_p}$ by unit propagation.
    We can extend $\sigma$ so that every other variable in $f^n_{k+1}$ is set to $0$,
    then by Definition \ref{def:gsn} the output $\sigma(\bar{y})$ is sorted and $|\sigma(\bar{y})|_1 \geq 1$, therefore we conclude that $t=1$, so (ii) is true.
  \end{proof}
  
  \begin{lemma}\label{lma:red_n}
    Let $1 \leq i \leq n$. If we set $x_i=1$ in a partial assignment $\sigma$ (where the rest of the variables are undefined),
    then unit propagation will set $y_1=1$. Furthermore,
    $f^n_{k+1}$ can be reduced to $f^{n-1}_{k}$ by $\sigma$, where $f^{n-1}_{k}$ is a generalized selection network of order $(n-1,k)$.
  \end{lemma}

  \begin{proof}
    First part is a simple consequence of Lemma \ref{lma:pl}.
    Take partial assignment $\sigma$ (after unit propagation set $y_1=1$). The assigned variables are
    exactly the ones from the propagation path $z_{x_i}=\tuple{z_1,\dots,z_r}$. Therefore, by Lemma \ref{lma:pl},
    we need to consider clauses from a single selector of each layer. Take each $1 \leq j \leq r$ and $s^{n'}_{m'} \in S_j$ such
    that $z_j$ is an input of $s^{n'}_{m'}$. By Lemma \ref{lma:red_s}, $s^{n'}_{m'}$ can be reduced to a smaller selector $s^{n'-1}_{m'-1}$.
    Thus $f^n_{k+1}$ can be reduced to $f^{n-1}_{k}$ with inputs $\{x_1,\dots,x_n\} \setminus \{x_i\}$ and outputs $\{y_2,\dots,y_{k+1}\}$,
    which is a generalized selection network of order $(n-1,k)$.
  \end{proof}
  
  \begin{lemma}\label{lma:fp}
    Let $1 \leq i \leq k+1$ and let $\pi \, : \, \{1,\dots,i\} \rightarrow \{1,\dots,n\}$ be a 1-1 function.
    Assume that inputs $x_{\pi(1)},\dots,x_{\pi(i)}$ are set to $1$, and the rest of the variables
    are undefined. Unit propagation will set variables $y_1,\dots,y_{i}$ to $1$.
  \end{lemma}
  
  \begin{proof}
    By setting $1$ to variables $x_{\pi(1)},\dots,x_{\pi(i)}$ in an (arbitrarily) ordered way and repeated application of Lemma \ref{lma:red_n}.
  \end{proof}

  The process of propagating $1$'s given in Lemma \ref{lma:fp} we call a {\em forward propagation}.
  We are ready to prove the main result.

  \begin{theorem}
    Consider the standard encoding $\phi^n_{k}(\bar{x},\bar{y},L) = f^n_{k+1}(\bar{x},\bar{y},L) \wedge \neg y_{k+1}$.
    Assume that $k$ inputs are set to $1$ and forward propagation has been performed which set all $y_1,\dots,y_{k}$ to $1$.
    Then the unit propagation will set all undefined input variables to $0$.
  \end{theorem}

  \begin{proof}
    First, reduce $f^n_{k+1}$ to $f^{n'}_1$ ($n'=n-k$), by repeated application of Lemma \ref{lma:red_n}.
    Let new input variables be $\tuple{x'_1,\dots,x'_{n'}}$ and let the single output be $y'_1 \equiv y_{k+1}$.
    Since the standard encoding contains a singleton clause $\neg y_{k+1}$, then $y'_1 = y_{k+1} = 0$.
    Let $\tuple{L'^{(n'_1,k'_1)}_1, \dots, L^{(n'_r,k'_r)}_r}$ be the sequence of layers in $f^{n'}_1$,
    where $n'_1,k'_1,\dots,n'_r,k'_r \in \nat$. Let $0 \leq i < r$, we prove the following statement by induction on $i$:

    \[
      S(i) = \text{unit propagation sets all input variables of layer $L^{(n'_{r-i},k'_{r-i})}_{r-i}$ to $0$.} 
    \]
    
    When $i=0$, then we consider the last layer $L^{(n'_r,k'_r)}_r$. Since it has only one output, namely $y'_1$,
    it consists of a single selector $s^{n'_r}_1$. We know that $y'_1 = 0$, then by Lemma \ref{lma:pp} all input variables are set to $0$ by UP,
    therefore $S(0)$ holds. Take any $i \geq 0$ and assume that $S(i)$ is true. That means that all input variables
    of $L^{(n'_{r-i},k'_{r-i})}_{r-i}$ are set to $0$. From the definition of generalized network,
    those inputs are the outputs of the layer $L^{(n'_{r-i-1},k'_{r-i-1})}_{r-i-1}$. Therefore each selector of
    layer $L^{(n'_{r-i-1},k'_{r-i-1})}_{r-i-1}$ has all outputs set to 0. By Lemma \ref{lma:pp} all input variables of this selectors are set to 0 by UP, and therefore
    all inputs of $L^{(n'_{r-i-1},k'_{r-i-1})}_{r-i-1}$ are set to 0. Thus we conclude that $S(i+1)$ holds. This completes the induction step.

    We know that $S(r-1)$ is true, which means that all input variables of layer $L^{(n'_{1},k'_{1})}_{1}$ are set to 0,
    those are exactly the input variables $\tuple{x'_1,\dots,x'_{n'}}$ of $f^{n'}_1$. Those variables are previously undefined
    input variables of $f^n_{k+1}$, which completes the proof.
  \end{proof}
  
  \section{Summary}

  In this chapter we have given definitions and conventions used in the rest of the thesis.
  We have presented several comparator network models and we conclude that for the purpose of
  this thesis the procedural representation is sufficient to present all our encodings.
  We deconstructed a comparator into a set of clauses and showed how 0-1 values are being propagated therein.

  We have also defined a model based on layers of selectors and we showed the first rigorous
  proof that any standard encoding based on generalized selection networks
  preserves arc-consistency. We believe that this result will relieve future researchers of this topic from the burden
  of proving that their encodings are arc-consistent, which was usually a long and technical endeavor.

\chapter[Encodings of Boolean Cardinality Constraints]{Encodings of Boolean \\ Cardinality Constraints}\label{ch:history}

    \def\nqueenssolution{Qd4, Qe2, Qf8}
    \setchessboard{smallboard,labelleft=false,labelbottom=false,showmover=false,setpieces=\nqueenssolution}

    \begin{tikzpicture}[remember picture,overlay]
      \node[anchor=east,inner sep=0pt] at (current page text area.east|-0,3cm) {\chessboard};
    \end{tikzpicture}

The practical importance of encoding cardinality constraints into SAT resulted in a large number
of research papers published in the last 20 years. In this chapter we review some of the most significant
methods found in the literature for three types of constraints: {\em at-most-k}, {\em at-most-one} and {\em Pseudo-Boolean}.
Although Pseudo-Boolean constraints are more general than cardinality constraints, many methods for efficiently
encoding cardinality constraints were first developed for Pseudo-Boolean constraints.

\section{At-Most-k Constraints} \label{sec:atmk}

Boolean cardinality constraints -- which are the main focus of this thesis -- are also called at-most-k constraints
in the literature \cite{frisch2010sat}, which is understandable, since we require that {\bf at most}
$k$ out of $n$ propositional literals can be true, and constraints with other relations (at least, or exactly)
can be easily reduced to the "at most" case (Observation \ref{obs:card}).
Here we present the most influential ideas in the topic of translating such constraints into
propositional formulas. The overview is presented in the chronological order.

\paragraph{Binary adders.} Warners \cite{warners1998linear} considered encoding based on using adders where numbers are represented in binary.
The original method was devised for PB-constraints and involved using binary addition of terms of the LHS (left-hand side) of the constraint
$a_1x_1 + a_2x_2 + \dots + a_nx_n \leq k$ using {\em adders} and {\em multiplicators}, then comparing the resulting sum to the binary representation
of $k$. On the highest level, the algorithm uses a divide-and-conquer strategy, partitioning the terms of the LHS into two sub-sums, recursively encoding them into binary numbers, then encoding the summation of those two numbers using an adder. An adder is simply a formula that models the classical addition in columns:

\begin{center}
  {\renewcommand{\arraystretch}{2}%
  \begin{tabular}{c@{\;\;}c@{\;\;}c@{\;\;}c@{\;\;}c@{\;\;}c@{\;\;}c}
    & $p_{M-1}^A$ & $p_{M-2}^A$ & $\dots$ & $p_{2}^A$ & $p_{1}^A$ & $p_{0}^A$ \\
$+$ & $p_{M-1}^B$ & $p_{M-2}^B$ & $\dots$ & $p_{2}^B$ & $p_{1}^B$ & $p_{0}^B$ \\
\hline
$p_{M}^C$ & $p_{M-1}^C$ & $p_{M-2}^C$ & $\dots$ & $p_{2}^C$ & $p_{1}^C$ & $p_{0}^C$ \\
\end{tabular}}
\end{center}

\noindent In the above we compute $A+B=C$, where $p_{i}^X$ is a propositional variable representing $i$-th bit of number $X$, for $1 \leq i \leq M$.
We assume that $M$ is a constant bounding the number of bits in $C$. The addition is done column-by-column, from right to left, taking into
consideration the potential carries (propositional variables $c_{i,j}$'s): 

\[
  \left(p_0^C \iff (p_0^A \iff \neg p_0^B)\right) \wedge \left(c_{0,1} \iff (p_0^A \wedge p_0^B)\right) \wedge \bigwedge_{j=1}^M \left(p_j^C \iff (p_j^A \iff p_j^B \iff c_{j-1,j})\right) \; \wedge
\]

\[
  \bigwedge_{j=1}^{M-1} \left(c_{j,j+1} \iff (p_j^A \wedge p_j^B) \vee (p_j^A \wedge c_{j-1,j}) \vee (p_j^B \wedge c_{j-1,j})\right) \wedge \left(c_{M-1,M} \iff p_{M}^C\right).
\]

\noindent Notice that $(p_j^A \iff p_j^B \iff c_{j-1,j})$ is true if and only if $1$ or $3$ variables from the set $\{p_j^A,p_j^B,c_{j-1,j}\}$
are true.

For the base case, we have to compute the multiplication of $a_i$ and $x_i$. This is also done in a straightforward way.
In the following, let $B_{a_i}$ be the set containing indices of all 1's in the binary representation of $a_i$. We get:

\[
  \bigwedge_{k \in B_{a_i}} \left( p_k^{a_i} \iff x_i \right) \wedge \bigwedge_{k \not\in B_{a_i}} \neg p_k^{a_i}
\]

\noindent Finally, we have to enforce the constraint $(\leq k)$, which is done like so:

\[
  \bigwedge_{i \not\in B_k}\left(p_i^{LHS} \Rightarrow \neg \bigwedge_{j \in B_k: \, j>i} p_j^{LHS}\right),
\]

\noindent where $p_0^{LHS},p_1^{LHS},\dots$ are propositional variables representing the binary number of the sum of the LHS.

\begin{examplebox}
\begin{example}\label{ex:warners}
  Let $k=26$, as in the example from \cite{warners1998linear}. Then $B_k = \{1,3,4\}$. In order to enforce $LHS \leq k$,
  we add the following clauses:  

  \small
  \[
    \left(p^{LHS}_0 \Rightarrow \neg(p_1^{LHS} \wedge p_3^{LHS} \wedge p_4^{LHS})\right) \; \wedge \; 
    \left(p^{LHS}_2 \Rightarrow \neg(p_3^{LHS} \wedge p_4^{LHS})\right).
  \]
  \normalsize 

  \noindent We can check that this indeed enforces the constraint for some sample values of the LHS. For example, if $LHS = 25$, then
  $\tuple{p^{LHS}_0,p^{LHS}_1,p^{LHS}_2,p^{LHS}_3,p^{LHS}_4} = \tuple{1,0,0,1,1}$, and we can see that the formula above is satisfied. 
  However, for $LHS = 30$, we get $\tuple{p^{LHS}_0,p^{LHS}_1,p^{LHS}_2,p^{LHS}_3,p^{LHS}_4} = \tuple{0,1,1,1,1}$, and the second implication
  evaluates to false.
\end{example}
\end{examplebox}

In Lemma 2 of \cite{warners1998linear} it is shown that number of variables and clauses of the proposed encoding
is bounded by $2n(1+ \log(a_{max}))$ and $8n(1+ 2\log(a_{max}))$, respectively ($a_{max}=\max\{a_i\}$).
Therefore, in case of cardinality constraints, the encoding uses $O(n)$ variables and clauses. This encoding is small,
but does not preserve arc-consistency.

\paragraph{Totalizers.} Bailleux and Boufkhad \cite{bailleux2003efficient} presented an encoding based on the idea of a totalizer.
The totalizer is a binary tree, where the leaves are the constraint literals $x_i$'s. With each inner node the number $s$
is associated which represents the sum of the leaves in the corresponding sub-tree. The number $s$ is represented in unary,
by $s$ auxiliary variables. The encoding is arc-consistent and uses $O(n\log n)$ variables and $O(n^2)$ clauses. Here we briefly
explain the idea of a totalizer as described in \cite{bailleux2003efficient}.

We begin with defining unary representation of a number $v$ such that $0 \leq v \leq n$.
An integer $v$ can be modeled by a set $V=\{v_1,v_2,\dots,v_n\}$ of $n$ propositional variables.
Each possible value of $v$ is encoded as a complete instantiation of $V$, such that if $v=x$, then $x$ 1's follow $(n-x)$ 0's, i.e.,
$v_1=1,v_2=1,\dots,v_x=1,v_{x+1}=0,\dots,v_n=0$. A partial instantiation of $V$ is said to be {\em pre-unary} if for each $v_i=1$,
$v_j=1$ for any $j<i$ and for each $v_i=0$, $v_j=0$ for any $j,i \leq j \leq n$. In other words, $V$ has its prefix set with $1$'s and
its suffix with $0$'s (prefix and suffix cannot overlap in this context).

The advantage of using the unary representation is that the integer can be specified as belonging to an interval.
The inequality $x \leq v \leq y$ is specified by the partial pre-unary instantiation of $V$ that fixes to 1 any $v_i$
such that $i \leq x$ and fixes to 0 any $v_j$ such that $j \geq y+1$.

\begin{examplebox}
\begin{example}
  Following an example from \cite{bailleux2003efficient}, consider $n=6$ and a partial instantiation such that $v_1=v_2=1$,
  $v_5=v_6=0$ and $v_3,v_4$ are free. Then the corresponding integer $v$ is such that $2 \leq v \leq 4$. 
\end{example}
\end{examplebox}

The totalizer is a CNF formula defined on 3 sets of variables: $X={x_1,\dots,x_n}$ -- inputs, $Y={y_1,\dots,y_n}$ -- outputs, and a set S
of {\em linking variables}. These sets can be described by a binary tree built as follows. We start from an isolated root node labeled $n$ and we
proceed iteratively: to each node labeled by $m>1$, we connect two children labeled by $\floor{m/2}$ and $\ceil{m/2}$, respectively. This
produces a binary tree with $n$ leaves labeled 1. Next, each variable in $X$ is allocated to a leaf in a bijective way. Set $Y$ is allocated
to the root node. To each internal node labeled by an integer $m$, a set of $m$ new variables is allocated that is used to represent an unary
value between $1$ and $m$. All those internal variables produce the set $S$.

We now define an encoding that ensures that $m=m_1+m_2$ in any complete instantiation of the variables belonging to
some sub-tree of a totalizer labeled $m$ with two children labeled $m_1$ and $m_2$. Let $S=\{s_1,\dots,s_m\}$,
$S_1=\{s_1^1,\dots,s_{m_1}^1\}$ and $S_2=\{s_1^2,\dots,s_{m_2}^2\}$ be the sets of variables related to $m$, $m_1$ and $m_2$,
respectively. We add the following set of clauses:

\[
  \bigwedge_{\substack{0 \leq a \leq m_1 \\ 0 \leq b \leq m_2 \\ 0 \leq c \leq m \\ a+b=c}} \left((s_a^1 \wedge s_b^2 \Rightarrow s_c) \wedge (s_{c+1} \Rightarrow s_{a+1}^1 \vee s_{b+1}^2)\right),
\]

\noindent where $s^1_0=s^2_0=s_0=1,s^1_{m_1+1}=s^2_{m_2+1}=s_{m+1}=0$.

\begin{tcolorbox}[colback=white,colframe=black,sharp corners,enhanced jigsaw,breakable,break at=7cm]
\begin{example}
  Borrowing an example from \cite{bailleux2003efficient}, for $n=5$, $X=\{x_1,x_2,x_3,x_4,x_5\}$ and $Y=\{y_1,y_2,y_3,y_4,y_5\}$ the following
  tree is obtained by the totalizer procedure:

  \begin{center}
    \begin{tikzpicture}[level/.style={sibling distance=60mm/#1}]
      \node [text width=4em, text centered] (z){$(Y,5)$}
      child {node [text width=4em, text centered] (a) {$(\{s^1_1, s^1_2\},2)$}
        child {node [text width=4em, text centered] (b) {$(\{x_1\},1)$}}
        child {node [text width=4em, text centered] (g) {$(\{x_2\},1)$}}
      }
      child {node [text width=4em, text centered] (j) {$(\{s^2_1, s^2_2,s^2_3\},3)$}
        child {node [text width=4em, text centered] (k) {$(\{x_3\},1)$}}
        child {node [text width=6em, text centered] (l) {$(\{s^3_1, s^3_2\},2)$}
          child {node [text width=4em, text centered] (o) {$(\{x_4\},1)$}}
          child {node [text width=4em, text centered] (p) {$(\{x_5\},1)$}}
        }
      };
    \end{tikzpicture}
  \end{center}

  \noindent Here, the set of linking variables is $S=\{s^1_1,s^1_2,s^2_1,s^2_2,s^2_3,s^3_1,s^3_2\}$.
  Let us encode the unary addition on a node labeled 3 in the totalizer above (clauses with constants 0's and 1's are already simplified):

  \small
  \begin{align*}
    & (s^3_1 \impl s^2_1) \wedge (s^3_2 \impl s^2_2) \wedge (x_3 \impl s^2_1) \wedge (x_3 \wedge s^3_1 \impl s^2_2) \wedge (x_3 \wedge s^3_2 \impl s^2_3)\; \wedge \\
    & (s^2_1 \impl x_3 \vee s^3_1) \wedge (s^2_2 \impl x_3 \vee s^3_2) \wedge (s^2_3 \impl x_3) \wedge (s^2_2 \impl s^3_1) \wedge (s^2_3 \impl s^3_2).
  \end{align*}
  \normalsize

  \noindent Notice that the first row of clauses represents the relation $c \geq a + b$ and the second row of clauses represents the relation
  $c \leq a + b$, where $c$, $a$ and $b$ are possible values of the unary representations of nodes labeled $3$ and its children, respectively.
\end{example}
\end{tcolorbox}

B{\"u}ttner and Rintanen \cite{buttner2005satisfiability} made an improvement to the encoding of Bailleux and Boufkhad \cite{bailleux2003efficient}
by noticing that counting up to $k+1$ suffices to enforce the constraint. Therefore, they reduced the number of variables and clauses used in each node
of the totalizer. Their encoding improves the previous result for small values of $k$ as it requires $O(nk)$ variables and $O(nk^2)$ clauses. They
also proposed a novel encoding based on encoding the injective mapping between the true $x_i$ variables and $k$ elements. This idea requires
$O(nk)$ variables and clauses, but is not arc-consistent.

\paragraph{Counters.} Sinz \cite{sinz2005towards} proposed two encodings based on counters.
The first uses a sequential counter where numbers are represented in unary and the second
uses a parallel counter with numbers represented in binary.
The first encoding uses $O(nk)$ variables and clauses and the second encoding uses $O(n)$
variables and clauses. Only the encoding based on a sequential counter is arc-consistent.
We present here the construction of a sequential counter.

\begin{figure}[ht!]
  \centering
  \includegraphics[scale=1.3]{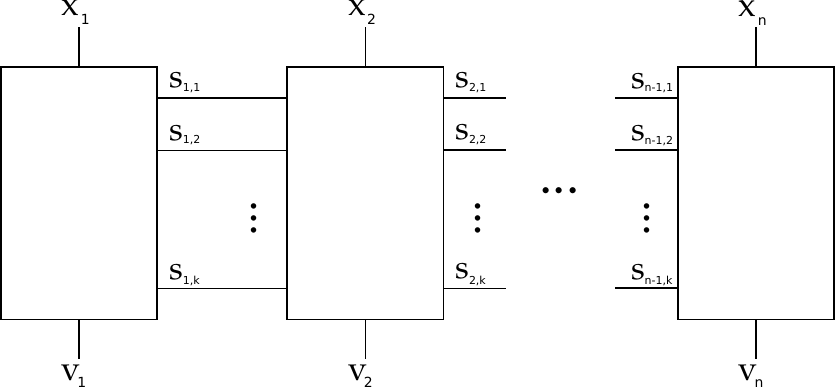}
  \caption{Schema of a sequential counter}
  \label{fig:seqc}
\end{figure}

The general idea is to build a count-and-compare hardware circuit and then translate this circuit to a CNF.
Such counter is presented in Figure \ref{fig:seqc} and consists of $n$ sub-circuits, each computing a partial sum $s_i = \sum_{j=1}^ix_j$.
The values of all $s_i$'s are represented as unary numbers, i.e., we have a sequence of variables $\tuple{s_{i,1},\dots,s_{i,k}}$ for each $s_i$.
Additional variable $v_i$ (the overflow bit) is related to each
sub-circuit and is set to true if the partial sum $s_i$ is greater than $k$. To convert the circuit to a CNF,
we first build a series of implications for the partial sum bits $s_{i,j}$'s ($j$-th bit of $s_i$) and the overflow bits $v_i$.
We can then simplify the formula noting that all overflow bits have to be zero, in order to enforce the constraint $\leq k$.
The resulting set of clauses is as follows:

\begin{align*}
  (x_1 \impl s_{1,1}) &\wedge (x_n \impl \neg s_{n-1,k}) \wedge \bigwedge_{1<j\leq k} \left(\neg s_{1,j} \right) \; \wedge \\
      \bigwedge_{1 < i < n} \bigg( &(x_i \impl s_{i,1}) \wedge (s_{i-1,1} \impl s_{i,1}) \wedge (x_i \impl \neg s_{i-1,k}) \; \wedge \\ 
        & \bigwedge_{1 < j \leq k} \Big( (x_i \wedge s_{i-1,j-1} \impl s_{i,j}) \wedge (s_{i-1,j} \impl s_{i,j}) \Big) \bigg).
\end{align*}

\paragraph{Sorting Networks.} One of the most influential ideas was introduced by E{\'e}n and S\"orensson \cite{minisatp}.
They proposed an encoding of Pseudo-Boolean constraints based on the odd-even sorting networks.
A PB-constraint is decomposed into a number of interconnected sorting networks,
where each sorter represents an adder of digits in a mixed radix base. Detailed explanation of this technique is done in Section \ref{sec:pb}.
In case of cardinality constraints a single sorting network is required and we have already explained the details of this technique
in the previous chapter. This encoding (using sorting networks) requires $O(n \log^2 n)$ variables and clauses and is arc-consistent.

It is worth noting that although we reduce cardinality constraints to the form $x_1 + x_2 + \dots + x_n \leq k$, and we use
3-clause representation for each comparator ($\{a \Rightarrow c, b \Rightarrow c, a \wedge b \Rightarrow d\}$) 
and assert the output $\neg y_{k+1}$, we do this only to simplify the presentation. In practice, when dealing with
other types of cardinality constraints, one should do the following:

\begin{itemize}
  \item For $x_1 + x_2 + \dots + x_n \geq k$ encode each comparator with the set of clauses $\{d \Rightarrow a, d \Rightarrow b, c \Rightarrow a \vee b\}$ and add a clause $y_{k}$.
  \item For $x_1 + x_2 + \dots + x_n = k$ encode each comparator with 6-clause representation $\{a \Rightarrow c, b \Rightarrow c, a \wedge b \Rightarrow d, d \Rightarrow a, d \Rightarrow b, c \Rightarrow a \vee b\}$ and assert both $y_{k}$ and $\neg y_{k+1}$.
\end{itemize}

The situation where reduction from one inequality to the other is beneficial is when $k > \floor{n/2}$, for example, given
$x_1 + x_2 + \dots + x_n \geq k$ we reduce it to $\neg x_1 + \neg x_2 + \dots + \neg x_n \leq n-k$ (following Observation \ref{obs:card}).
In the resulting cardinality constraint $n-k < \ceil{n/2}$. This does not make the encoding smaller in case of sorting networks, but
in case of selection networks this might vastly reduce the size of the resulting CNF.

\paragraph{Selection Networks.} Further improvements in encoding cardinality constraints
are based on the aforementioned idea of E{\'e}n and S\"orensson \cite{minisatp}.
Basically, in order to make more efficient encodings, more efficient sorting networks are required. It was observed that we do not need to
sort the entire input sequence, but only the first $k+1$ largest elements. Hence, the use of {\em selection networks} allowed to achieve
the complexity $O(n \log^2 k)$ in terms of the number of variables and clauses. In the last years several selection networks were proposed
for encoding cardinality constraints and experiments proved their efficiency. They were based
mainly on the odd-even or pairwise comparator networks. Codish and Zazon-Ivry
\cite{codish2010pairwise} introduced Pairwise Selection Networks that used the concept of Parberry's
Pairwise Sorting Network \cite{parberry1992pairwise}. Their construction was later
improved (we show this result in Chapter \ref{ch:cp15}).
Ab\'io, As\'in, Nieuwenhuis, Oliveras and
Rodr\'iguez-Carbonell \cite{asin2011cardinality,abio2013parametric} defined encodings that implemented
selection networks based on the odd-even sorting networks.
In \cite{abio2013parametric} the authors proposed a mixed
parametric approach to the encodings, where so called {\em Direct Cardinality Networks} are chosen for small
sub-problems and the splitting point is optimized when large problems are divided into two
smaller ones. They proposed to minimize the function $\lambda \cdot num\_vars +
num\_clauses$ in the encodings. The constructed encodings are small and efficient.

It's also worth noting that using encodings based on selection networks give an
extra edge in solving optimization problems for which we need to solve a sequence
of problems that differ only in the decreasing bound of a cardinality constraint.
In this setting we only need to add one more clause $\neg y_{k}$ for a new value
of $k$, and the search can be resumed keeping all previous clauses as it is. This
works because if a comparator network is a $k$-selection network, then it is also
a $k'$-selection network for any $k'<k$. This property is called {\em incremental
strengthening} and most state-of-the-art SAT-solvers provide a user interface for it.

\section{At-Most-One Constraints}

Much research has been done on cardinality constraints for small values of $k$.
The special case of at-most-k constraint is when $k=1$, which results in at-most-one constraint (AMO, in short),
which could be viewed as the simplest type of constraint and at the same time, the most useful one. 
It is due to the fact that AMO constraints are most widely used constraints during the process
of translating a practical problem into a propositional satisfiability instance.
We reference some of those encodings here.

For convenience, we denote $AMO(X)$ and $ALO(X)$
to be at-most-one and at-least-one clauses for the set of propositional variables $X=\{x_1,\dots,x_n\}$,
respectively, and we define $EO(X) = AMO(X) \wedge ALO(X)$.
We use a running example $AMO(x_1,\dots,x_8)$ to illustrate the encodings.

\paragraph{Binomial Encoding.} The simplest encoding is the {\em binomial} encoding, sometimes also called the {\em naive} encoding.
It is referenced in many papers, for example in \cite{frisch2010sat}. The idea of this encoding is to express that all possible combinations
of two variables are not simultaneously assigned to true. This requires $\binom{n}{2}$ clauses:

\[
  \bigwedge_{i=1}^{n-1}\bigwedge_{j=i+1}^{n} (\neg x_i \vee \neg x_j).
\]

The encoding does not require any additional variables, but the quadratic number of clauses
makes it impractical for large values of $n$. Nevertheless, due to its simplicity, it is
widely used in practice. Notice that we used this encoding to translate the 4-Queens Puzzle
to SAT in Section \ref{sec:sat:app}.

\begin{examplebox}
\begin{example}
  In the running example, the binomial encoding produces the following set of clauses:

  \small
  \begin{align*}
    & (\neg x_1 \vee \neg x_2) \wedge (\neg x_1 \vee \neg x_3) \wedge (\neg x_1 \vee \neg x_4) \wedge (\neg x_1 \vee \neg x_5) \wedge
    (\neg x_1 \vee \neg x_6) \wedge (\neg x_1 \vee \neg x_7) \wedge (\neg x_1 \vee \neg x_8) \\
    & (\neg x_2 \vee \neg x_3) \wedge (\neg x_2 \vee \neg x_4) \wedge (\neg x_2 \vee \neg x_5) \wedge
    (\neg x_2 \vee \neg x_6) \wedge (\neg x_2 \vee \neg x_7) \wedge (\neg x_2 \vee \neg x_8) \\
    & (\neg x_3 \vee \neg x_4) \wedge (\neg x_3 \vee \neg x_5) \wedge
    (\neg x_3 \vee \neg x_6) \wedge (\neg x_3 \vee \neg x_7) \wedge (\neg x_3 \vee \neg x_8) \\
    & (\neg x_4 \vee \neg x_5) \wedge
    (\neg x_4 \vee \neg x_6) \wedge (\neg x_4 \vee \neg x_7) \wedge (\neg x_4 \vee \neg x_8) \\
    & (\neg x_5 \vee \neg x_6) \wedge (\neg x_5 \vee \neg x_7) \wedge (\neg x_5 \vee \neg x_8) \\
    & (\neg x_6 \vee \neg x_7) \wedge (\neg x_6 \vee \neg x_8) \\
    & (\neg x_7 \vee \neg x_8) \\
  \end{align*}
  \normalsize
\end{example}
\end{examplebox}

\paragraph{Binary Encoding.} The {\em binary} encoding \cite{frisch2006solving} uses $\ceil{\log n}$ auxiliary variables
$\{b_1,\dots,b_{\ceil{\log n}}\}$ to reduce the number of clauses to $n\log n$. The idea is to create a mapping between
each label of the variables $\{x_1,\dots,x_n\}$ to its binary representation using the auxiliary variables ($b_j$ represent j-th bit of the number,
for $1 \leq j \leq \ceil{\log n}$) so that the truth assignment of one input variable $x_i$ implies that the rest of the variables evaluate to false:

\[
  \bigwedge_{i=1}^{n}\bigwedge_{j=1}^{\ceil{\log n}} (x_i \Rightarrow B(i,j)),
\]

\noindent where $B(i,j) \equiv b_j$, if $j$-th bit of $i-1$ (represented in binary) is 1, otherwise $B(i,j) \equiv \neg b_j$.

\begin{examplebox}
\begin{example}
  In the running example, the binary encoding produces the following set of clauses:

  \small
  \begin{align*}
    & (x_1 \Rightarrow \neg b_1) \wedge (x_1 \Rightarrow \neg b_2) \wedge (x_1 \Rightarrow \neg b_3) \, \wedge \\
    & (x_2 \Rightarrow b_1) \wedge (x_2 \Rightarrow \neg b_2) \wedge (x_2 \Rightarrow \neg b_3) \, \wedge \\
    & (x_3 \Rightarrow \neg b_1) \wedge (x_3 \Rightarrow b_2) \wedge (x_3 \Rightarrow \neg b_3) \, \wedge \\
    & (x_4 \Rightarrow b_1) \wedge (x_4 \Rightarrow b_2) \wedge (x_4 \Rightarrow \neg b_3) \, \wedge \\
    & (x_5 \Rightarrow \neg b_1) \wedge (x_5 \Rightarrow \neg b_2) \wedge (x_5 \Rightarrow b_3) \, \wedge \\
    & (x_6 \Rightarrow b_1) \wedge (x_6 \Rightarrow \neg b_2) \wedge (x_6 \Rightarrow b_3) \, \wedge \\
    & (x_7 \Rightarrow \neg b_1) \wedge (x_7 \Rightarrow b_2) \wedge (x_7 \Rightarrow b_3) \, \wedge \\
    & (x_8 \Rightarrow b_1) \wedge (x_8 \Rightarrow b_2) \wedge (x_8 \Rightarrow b_3)
  \end{align*}
  \normalsize
\end{example}
\end{examplebox}

\paragraph{Commander Encoding.}  In the {\em commander} encoding \cite{klieber2007efficient}
one splits the input variables into $m$ disjoint sets $\{G_1,\dots,G_m\}$
and introduce $m$ auxiliary commander variables $\{c_1,\dots,c_m\}$, one for each set. The constraint is enforced
by adding clauses so that exactly one variable from $G_i \cup {\neg c_i}$ is true and at most one of the commander variables are true:

\[
  \bigwedge_{i=1}^m EO(\{\neg c_i\} \cup G_i) \wedge AMO(c_1,\dots,c_m)
\]

\noindent where $EO(\{\neg c_i\} \cup G_i) = AMO(\{\neg c_i\} \cup G_i) \wedge ALO(\{\neg c_i\} \cup G_i)$ by definition. ALO part can be easily
translated into a single clause and AMO parts can be encoded either recursively or by another AMO encoding (like the binomial encoding).

\begin{examplebox}
\begin{example}
  In the running example, if we set $m=4$ and we divide the input set $X=\{x_1,\dots,x_8\}$ into subsets $G_1=\{x_1,x_2\}$, 
  $G_2=\{x_3,x_4\}$, $G_3=\{x_5,x_6\}$ and $G_4=\{x_7,x_8\}$, and we use the binomial encoding for the AMO parts, then the 
  commander encoding produces the following set of clauses:

  \small
  \begin{align*}
    & (c_1 \vee \neg x_1) \wedge (c_1 \vee \neg x_2) \wedge (\neg x_1 \vee \neg x_2) \wedge (\neg c_1 \vee x_1 \vee x_2) \, \wedge \\
    & (c_2 \vee \neg x_3) \wedge (c_2 \vee \neg x_4) \wedge (\neg x_3 \vee \neg x_4) \wedge (\neg c_2 \vee x_3 \vee x_4) \, \wedge \\
    & (c_3 \vee \neg x_5) \wedge (c_3 \vee \neg x_6) \wedge (\neg x_5 \vee \neg x_6) \wedge (\neg c_3 \vee x_5 \vee x_6) \, \wedge \\
    & (c_4 \vee \neg x_7) \wedge (c_4 \vee \neg x_8) \wedge (\neg x_7 \vee \neg x_8) \wedge (\neg c_4 \vee x_7 \vee x_8) \\
    & \wedge \\
    & (\neg c_1 \vee \neg c_2) \wedge (\neg c_1 \vee \neg c_3) \wedge (\neg c_1 \vee \neg c_4) \wedge (\neg c_2 \vee \neg c_3) \wedge
    (\neg c_2 \vee \neg c_4) \wedge (\neg c_3 \vee \neg c_4).
  \end{align*}
  \normalsize
\end{example}
\end{examplebox}

\paragraph{Product Encoding.} Chen \cite{chen2010new} proposed the {\em product} encoding, where the idea is to arrange the input variables
into a 2-dimensional array and to enforce that in at most one column and at most one row the variable can be set to true. Let $p,q \in \nat$
such that $p \times q \geq n$. We introduce the row variables $R=\{r_1,\dots,r_p\}$ and column variables $C=\{c_1,\dots,c_q\}$. We map
the inputs into a 2-dimensional array and enforce the constraint in the following way:

\[
  AMO(R) \wedge AMO(C) \wedge \bigwedge_{1 \leq i \leq p, \, 1 \leq j \leq q}^{1 \leq k \leq n, \, k=(i-1)q + j} (x_k \Rightarrow r_i \wedge c_j),
\]

\noindent where $AMO(R)$ and $AMO(C)$ can be computed recursively or by another encoding.

\begin{examplebox}
\begin{example}
  In the running example, if we set $p=q=3$, then the arrangement of the variables can be illustrated as follows:

  \Large
  \begin{center}
    \begin{tabular}{ccccc}
                           & $c_1$                      & $c_2$                      & $c_3$                      &  \\ \cline{2-4}
\multicolumn{1}{l|}{$r_1$} & \multicolumn{1}{l|}{$x_1$} & \multicolumn{1}{l|}{$x_2$} & \multicolumn{1}{l|}{$x_3$} &  \\ \cline{2-4}
\multicolumn{1}{l|}{$r_2$} & \multicolumn{1}{l|}{$x_4$} & \multicolumn{1}{l|}{$x_5$} & \multicolumn{1}{l|}{$x_6$} &  \\ \cline{2-4}
\multicolumn{1}{l|}{$r_3$} & \multicolumn{1}{l|}{$x_7$} & \multicolumn{1}{l|}{$x_8$} & \multicolumn{1}{l|}{}      &  \\ \cline{2-4}
    \end{tabular}
  \end{center}
  \normalsize

  \noindent If we use the binomial encoding for $AMO(R)$ and $AMO(C)$, then the product encoding produces the following set of clauses:

  \small
  \begin{align*}
    & (\neg r_1 \vee \neg r_2) \wedge (\neg r_1 \vee \neg r_3) \wedge (\neg r_2 \vee \neg r_3) \, \wedge \\
    & (\neg c_1 \vee \neg c_2) \wedge (\neg c_1 \vee \neg c_3) \wedge (\neg c_2 \vee \neg c_3) \\
    & \wedge \\
    & (x_1 \Rightarrow r_1 \wedge c_1) \wedge (x_2 \Rightarrow r_1 \wedge c_2) \wedge (x_3 \Rightarrow r_1 \wedge c_3) \, \wedge \\
    & (x_4 \Rightarrow r_2 \wedge c_1) \wedge (x_5 \Rightarrow r_2 \wedge c_2) \wedge (x_6 \Rightarrow r_2 \wedge c_3) \, \wedge \\
    & (x_7 \Rightarrow r_3 \wedge c_1) \wedge (x_8 \Rightarrow r_3 \wedge c_2).
  \end{align*}
  \normalsize
\end{example}
\end{examplebox}

\paragraph{Bimander Encoding.} Recently, hybrid approaches have emerged, for example the {\em bimander encoding} \cite{nguyen2015new}
borrows ideas from both {\em binary} and {\em commander} encodings, and the experiments show that the new encoding
is very competitive compared to other state-of-art encodings. The encoding is obtained as follows: we partition a set of
input variables $X=\{x_1,\dots,x_n\}$ into $m$ disjoint subsets $\{G_1,\dots,G_m\}$ such that each subset consists of
$g = \ceil{n/m}$ variables. This step is similar to the commander encoding, but instead of using $m$ commander variables,
we introduce auxiliary variables $\{b_1,\dots,b_{\ceil{\log m}}\}$, just like in the binary encoding. The new variables
take over the role of commander variables in the new encoding. The bimander encoding produces the following set of clauses:

\[
  \bigwedge_{i=1}^{m} AMO(G_i) \wedge \bigwedge_{i=1}^{m}\bigwedge_{h=1}^{g}\bigwedge_{j=1}^{\ceil{\log m}} (x_{i,h} \Rightarrow B(i,j)),
\]

\noindent where $x_{i,h}$ is the $h$-th element in $G_i$ and $B(i,j)$ is defined the same as in binary encoding.

\begin{examplebox}
\begin{example}
  In the running example, if we set $m=3$, thus obtaining $G_1=\{x_1,x_2,x_3\}$, $G_2=\{x_4,x_5,x_6\}$ and $G_3=\{x_7,x_8\}$,
  and we use the binomial encoding for the AMO part, then the bimander encoding produces the following set of clauses:

  \small
  \begin{align*}
    & (\neg x_1 \vee \neg x_2) \wedge (\neg x_1 \vee \neg x_3) \wedge (\neg x_2 \vee \neg x_3) \, \wedge \\
    & (\neg x_4 \vee \neg x_5) \wedge (\neg x_4 \vee \neg x_6) \wedge (\neg x_5 \vee \neg x_6) \, \wedge \\
    & (\neg x_7 \vee \neg x_8) \\
    & \wedge \\
  & (x_1 \Rightarrow \neg b_1) \wedge (x_1 \Rightarrow \neg b_2) \wedge (x_2 \Rightarrow \neg b_1) \wedge (x_2 \Rightarrow \neg b_2) \, \wedge \\
  & (x_3 \Rightarrow \neg b_1) \wedge (x_3 \Rightarrow \neg b_2) \wedge (x_4 \Rightarrow b_1) \wedge (x_4 \Rightarrow \neg b_2) \, \wedge \\
  & (x_5 \Rightarrow b_1) \wedge (x_5 \Rightarrow \neg b_2) \wedge (x_6 \Rightarrow b_1) \wedge (x_6 \Rightarrow \neg b_2) \, \wedge \\
  & (x_7 \Rightarrow \neg b_1) \wedge (x_7 \Rightarrow b_2) \wedge (x_8 \Rightarrow \neg b_1) \wedge (x_8 \Rightarrow b_2).
  \end{align*}
  \normalsize
\end{example}
\end{examplebox}

\section{Pseudo-Boolean Constraints} \label{sec:pb}

The current trend in encoding cardinality constraints involve comparator networks. The experiments show vast superiority
over other approaches. Nevertheless, some methods for encoding PB-constraints are worth mentioning here
(several were already referenced in Section \ref{sec:atmk}), as PB-constraints are a superset
of cardinality constraints. For example, E{\'e}n and S\"orensson \cite{minisatp} developed a PB-solver called \textsc{MiniSat+},
where the solver chooses between three techniques to generate SAT encodings for Pseudo-Boolean constraints.
These convert the constraint to: a BDD structure, a network of binary adders, a network of sorters.
The network of adders is the most concise encoding, but it can have poor propagation properties and often
leads to longer computations than the BDD based encoding. We introduce two techniques that are the basis for
what is considered to be the current state-of-the-art in PB-solving.

\paragraph{Reduced Ordered BDDs.} Recent development in PB-solvers show superiority of encodings
based on {\em Binary Decision Diagrams} (BDDs). The main advantage of BDD-based encodings is that
the resulting size of the formula is not dependent on the size of the coefficients of a PB-constraint.
The first to apply BDDs in the context of encoding PB-constraints were Bailleux, Boufkhad and Roussel \cite{bailleux}.
In the worst case, the size of the resulting CNF formula of their BDD encoding is exponential with respect
to the size of the encoded PB-constraint, but when applied to cardinality constraints,
the encoding is arc-consistent and uses $O(n^2)$ variables and clauses.

Ab{\'\i}o et al. \cite{abio2012} show a construction of {\em Reduced Ordered BDDs} (ROBDDs),
which produce arc-consistent, efficient encoding for PB-constraints. Here we briefly describe their method.
A Reduced Ordered BDD for a PB-constraint $a_1x_1 + a_2x_2 + \dots + a_nx_n \leq k$ is obtained as follows.
An ordering of the variables is established, suppose that it is $\tuple{x_1,x_2,\dots,x_n}$, for convenience.
We build a directed graph with a root node $x_1$. A node has two children: false child and true child. False child
represent the PB-constraint assuming $x_1=0$ (i.e., $a_2x_2 + a_3x_3 + \dots a_nx_n \leq k$),
and its true child represents $a_2x_2 + a_3x_3 \leq k-a_1$. The process is repeated until we reach the last variable.
Then, a constraint of the form $0 \leq K$ is the true node (1) if $K \geq 0$, and the false node (0) if $K < 0$.
This results in what is called an Ordered BDD. For obtaining a Reduced Ordered BDD, two reductions are applied (until fix-point):
removing nodes with identical children and merging isomorphic subtrees. This reduces the size of the initial BDD. We encode BDDs
into CNFs by introducing an auxiliary variable $a$ for every node. If the select variable of the node is $x$ and the auxiliary variables for
the false and true child are $f$ and $t$, respectively, then add the {\em if-then-else} clauses:

\begin{center}
\begin{tabular}{c@{\hskip 1in}c@{\hskip 1in}c}
 {$\!\begin{aligned}
    \neg x \wedge \neg f &\Rightarrow \neg a \\
    \neg x \wedge f &\Rightarrow a
 \end{aligned}$}
 &
 {$\!\begin{aligned}
    x \wedge \neg t &\Rightarrow \neg a \\
    x \wedge t &\Rightarrow a
 \end{aligned}$}
 &
 {$\!\begin{aligned}
    \neg f \wedge \neg t &\Rightarrow \neg a \\
    f \wedge t &\Rightarrow a
 \end{aligned}$}
\end{tabular}
\end{center}

\begin{examplebox}
\begin{example}
  This example is taken from Section 2 of \cite{abio2012}. Consider a PB-constraint $2x_1 + 3x_2 + 5x_3 \leq 6$
  and the ordering $\tuple{x_1,x_2,x_3}$. The Ordered BDD for this constraint looks like in the left figure:

  \begin{center}
    \begin{tabular}{ccc}
        \includegraphics[scale=0.9]{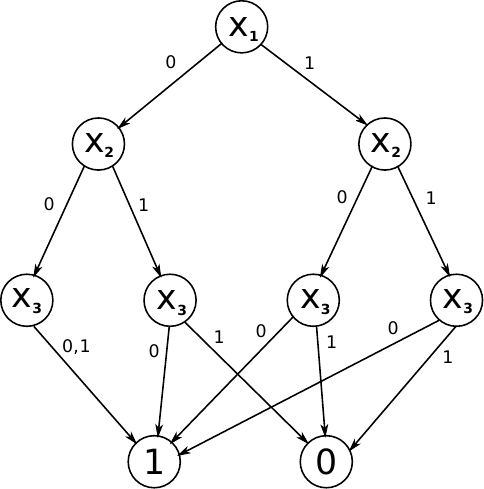} &
        ~~~~~~~~~ &
        \includegraphics[scale=0.9]{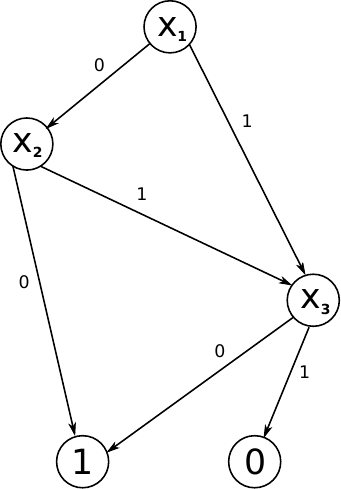}
    \end{tabular}
    \end{center}


  \noindent The root node has $x_1$ as selector variable. Its false child represent the PB-constraint assuming $x_1=0$ (i.e., $3x_2 + 5x_3 \leq 6$),
  and its true child represents $2 + 3x_2 + 5x_3 \leq 6$, that is, $3x_2 + 5x_3 \leq 4$.
  The two children have the next variable ($x_2$) as selector, and the process is repeated until we reach the last variable. Then,
  a constraint of the form $0 \leq K$ is the true node (1 on the graph) if $K \geq 0$, and the false node (0 on the graph) if $K < 0$.
  The Reduced Ordered BDD for this constraint is presented in the right figure above.
  
\end{example}
\end{examplebox}

In \cite{abio2012} authors show how to produce polynomial-sized ROBDDs and how to encode them
into SAT with only 2 clauses per node, and present experimental results that confirm that
their approach is competitive with other encodings and state-of-the-art Pseudo-Boolean solvers.
They present a proof that there are PB-constraints that admit no polynomial-size ROBDD,
regardless of the variable order, but they also show how to overcome the possible exponential
blowup of BDDs by carefully decomposing the coefficients of a given PB-constraint.

For further improvements, one can look at the work of Sakai and Nabeshima \cite{sakai2015},
where they extend the ROBDD construction to support constraints in the band form:
$l \leq \tuple{\text{Linear term}}\leq h$. They also propose an incremental SAT-solving strategy
of binary/alternative search for minimizing values of a given goal function and their experiments
show significant speed-up in SAT-solver runtime.

\begin{figure}[t!]
  \centering
  \includegraphics[scale=1.0]{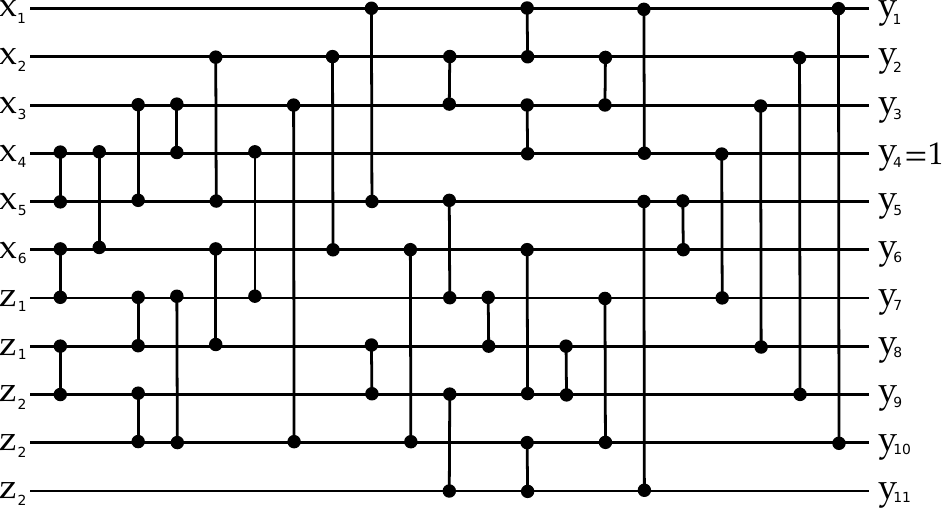}
  \caption{Sorting network with 11 inputs. Fourth output is set to 1 in order to assert the $\geq 4$ constraints. The network uses 35 comparators,
    which is almost optimal (the current best known lower bound is 33 \cite{codish2014twenty}).}
  \label{fig:sortmixed}
\end{figure}

\paragraph{Sorting Networks.} We revisit the concept of using sorting networks,
which have been successfully applied to encode cardinality constraints.
To demonstrate how sorters can be used to translate PB-constraints, consider the following
example from \cite{minisatp}:

\[
  x_1 + x_2 + x_3 + x_4 + x_5 + x_6 + 2z_1 + 3z_2 \geq 4
\]

\noindent The sum of coefficients is 11. We build a sorting network of size 11, feeding $z_1$ into two of the inputs,
$z_2$ into three of the inputs, and all the signals $x_i$ into one input each. To assert the constraint, one just asserts
the fourth output bit of the sorter, like in Figure \ref{fig:sortmixed}.

The shortcoming of this approach is that the resulting size of a CNF after transformation of the sorting network can get exponential if
the coefficients get bigger. Consider an example from \cite{abio2012}: both constraints $3x_1 + 2x_2 + 4x_3 \leq 5$ and
$30001x_1 + 19999x_2 + 39998x_3 \leq 50007$ are equivalent. The Boolean function they represent can be expressed, for example, by
the clauses $\bar{x}_1 \vee \bar{x}_3$ and $\bar{x}_2 \vee \bar{x}_3$. But clearly, a sorting network for the left constraint
would be smaller.

To remedy this situation the authors of \textsc{MiniSat+} propose a method to decompose the constraint into a number of
interconnected sorting networks, where sorters play the role of adders on unary numbers in a {\em mixed radix representation}.

In the classic base $r$ radix system, positive integers are represented as finite sequences of digits $\mathbf{d} = \tuple{d_0,\dots,d_{m-1}}$
where for each digit $0 \leq d_i < r$, and for the most significant digit, $d_{m-1} > 0$. The integer value associated with
$\mathbf{d}$ is $v=d_0 + d_1r + d_2r^2+\dots + d_{m-1}r^{m-1}$. A mixed radix system is a generalization where a base
$\mathbf{B}$ is a sequence of positive integers $\tuple{r_0,\dots,r_{m-1}}$.
The integer value associated with $\mathbf{d}$ is $v = d_0w_0 + d_1w_1 + d_2w_2 + \dots + d_{m-1}w_{m-1}$
where $w_0=1$ and for $i\geq 0$, $w_{i+1} = w_ir_i$. For example, the number $\tuple{2,4,10}_{\mathbf{B}}$ in base $\mathbf{B}=\tuple{3,5}$ is interpreted as
$2 \times \mathbf{1} + 4 \times \mathbf{3} + 10 \times \mathbf{15} = 164$ (values of $w_i$'s in boldface).

The decomposition of a PB-constraint into sorting networks is roughly as follows: first, 
find a "suitable" finite base $\mathbf{B}$ for the set of coefficients,
for example, in \textsc{MiniSat+} base is chosen so that the sum of all the digits of
the coefficients written in that base, is as small as possible. Then for
each element $r_i$ of $\mathbf{B}$ construct a sorting network where the inputs
to $i$-th sorter are those digits $\mathbf{d}$ (from the coefficients) where $d_i$
is non-zero, plus the potential carry bits from the $(i-1)$-th sorter.

\begin{tcolorbox}[colback=white,colframe=black,sharp corners,enhanced jigsaw,breakable,break at=6cm/20cm]
\begin{example}
  We show a construction of a sorting network system using an example from \cite{codish2011}, where authors show
  a step-by-step process of translating a PB-constraint $\psi = 2x_1+2x_2+2x_3+2x_4+5x_5+18x_6 \geq 23$.
  Let $\mathbf{B}=\tuple{2,3,3}$ be the considered mixed radix base. The representation of the coefficients of $\psi$
  in base $\mathbf{B}$ is illustrated as a $6 \times 4$ matrix:

  \begin{equation*}
      \begin{pmatrix}
      0 & 1 & 0 & 0 \\
      0 & 1 & 0 & 0 \\
      0 & 1 & 0 & 0 \\
      0 & 1 & 0 & 0 \\
      1 & 2 & 0 & 0 \\
      0 & 0 & 0 & 1 \\
     \end{pmatrix}
  \end{equation*}

  \noindent The rows of the matrix correspond to the representation of the coefficients in base $\mathbf{B}$.
  Weights of the digit positions of base $\mathbf{B}$ are $\bar{w}=\tuple{1,2,6,18}$. Thus, the decomposition
  of the LHS of $\psi$ is:

  \[
    \mathbf{1} \cdot (x_5) + \mathbf{2} \cdot (x_1 + x_2 + x_3 + x_4 + 2x_5) + \mathbf{6} \cdot (0) + \mathbf{18} \cdot (x_6)
  \]

  Now we construct a series of four sorting networks in order to encode the sums at each digit position of $\bar{w}$.
  Given values for the variables, the sorted outputs from these networks represent unary numbers $d_1$,$d_2$,$d_3$,$d_4$
  such that the LHS of $\psi$ takes the value $\mathbf{1} \cdot d_1 + \mathbf{2} \cdot d_2 + \mathbf{6} \cdot d_3 + \mathbf{18} \cdot d_4$.

  The final step is to encode the carry operation from each digit position to the next. The first three outputs must
  represent valid digits (in unary) for $\mathbf{B}$. In our example the single potential violation to this is $d_2$, which
  is represented in 6 bits. To this end we add two components to the encoding: (1) each third output of the second network
  is fed into the third network as carry input; and (2) a {\em normalizer} $R$ is added to encode that the output of the second network is to
  be considered modulo 3. The full construction is illustrated below:

  \begin{center}
    \includegraphics[scale=1.3]{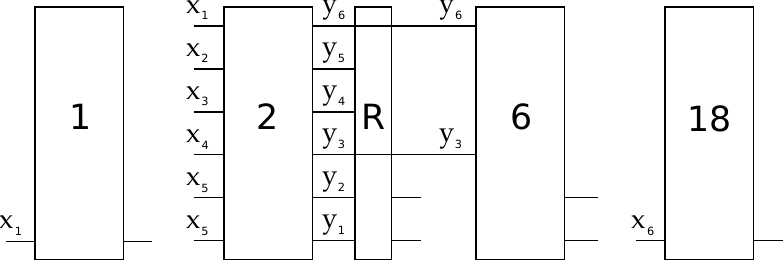}
  \end{center}

In the end, to enforce the constraint, we have to add clauses representing the relation $\geq 23$ (in base $\mathbf{B}$).
It is done by lexicographical comparison of bits representing LHS to bits representing $23_{(\mathbf{B})}$. See \cite{minisatp}
for a detailed description of the algorithm.
\end{example}
\end{tcolorbox}

On a final note, some research has been done on finding optimal mixed radix base for the aforementioned construction.
For example, Codish et al. \cite{codish2011} present an algorithm which scales to find an optimal
base consisting of elements with values up to $1,000,000$ and they consider several measures of optimality for finding the base.
They show experimentally that in many cases finding a better base leads also to better SAT-solving time.

\section{Summary}

The list of encodings presented here is not exhaustive, as many more encodings have been proposed in the past
for different types of constraints. For at-most-one constraints
one can look into the log encoding \cite{walsh2000sat}, ladder encoding \cite{gent2004new}, and
generalizations of the bimander encoding \cite{barahona2014representative}. For at-most-k constraints
there exists, for example, the partial sum encoding \cite{aloul2002generic} and perfect hashing encoding \cite{ben2012perfect}.

Here we present the comparison of the encodings introduced in the previous sections. All at-most-k constraints
can be reduced to at-most-one constraints by setting $k=1$. On the other hand, some encodings of at-most-one constraints
have generalized constructions for the at-most-k constraints, for example, Firsch and Giannaros \cite{frisch2010sat} give generalizations
for binary, commander and product encodings. For the binomial encoding the number of clauses grows significantly when considering the at-most-k constraint.
In the worst case of $k=\ceil{n/2}-1$ it requires $O(2^n/\sqrt{n/2})$ clauses. We summarize the encodings in Table \ref{tbl:encodings}.

\begin{table}[ht!]\renewcommand*{\arraystretch}{1.0}
\centering
\begin{tabular}{>{\centering}p{19mm}|>{\centering}p{7mm}|>{\centering}p{29mm}|>{\centering}p{29mm}|>{\centering}p{35mm}|p{4mm}}
Method                    & Type     & Origin                         & New vars. & Clauses & AC \\ \hline
\multirow{2}{*}{binomial} & $\leq 1$ & \multirow{2}{*}{folklore}      & 0     & $\binom{n}{2}$   & yes             \\ \cline{2-2} \cline{4-6}
                          & $\leq k$ &                                & 0     & $\binom{n}{k+1}$ & yes             \\ \hline
\multirow{2}{*}{binary}   & $\leq 1$ & Frisch et al. \cite{frisch2006solving} & $O(\log n)$ & $O(n \log n)$ & yes          \\ \cline{2-6}
                          & $\leq k$ & \makecell{Firsch \& \\ Giannaros \cite{frisch2010sat}} & $O(kn)$ & $O(kn \log n)$ & no           \\ \hline
\multirow{2}{*}{commander}& $\leq 1$ & \makecell{Kwon \& \\ Klieber \cite{klieber2007efficient}} & $n/2$ & $3.5n$ & yes \\ \cline{2-6}
                          & $\leq k$ & \makecell{Firsch \& \\ Giannaros \cite{frisch2010sat}} & $kn/2$ & $\left(\binom{2k+2}{k+1}+\binom{2k+2}{k-1}\cdot n/2\right)$ & yes \\ \hline
\multirow{2}{*}{product}  & $\leq 1$ & Chen \cite{chen2010new}                   & $2\sqrt{n}+O(\sqrt[4]{n})$ & $2n + 4\sqrt{n} + O(\sqrt[4]{n})$ & yes \\ \cline{2-6}
                          & $\leq k$ & \makecell{Firsch \& \\ Giannaros \cite{frisch2010sat}} & $(k+1)O(\sqrt[k]{n})$ & $(k+1)(n + O(k\sqrt[k]{n}))$ & yes \\ \hline
bimander                  & $\leq 1$ & \makecell{Mai \& \\ Nguyen \cite{nguyen2015new}} & $\ceil{\log m}$ & $n^2/2m + n\ceil{\log m} - n/2$ & yes \\ \hline
\multirow{2}{*}{adders}   & $\leq k$ & \multirow{2}{*}{Warners \cite{warners1998linear}} & $2n$ & $8n$ & no \\ \cline{2-2} \cline{4-6}
                          & PB       &                                                   & $2n(1+\log(a_{\max}))$ & $8n(1+2\log(a_{\max}))$ & no \\ \hline
totalizers                & $\leq k$ & \makecell{B{\"u}ttner \& \\ Rintanen \cite{buttner2005satisfiability}} & $O(kn)$ & $O(k^2n)$ & yes \\ \hline
\multirow{2}{*}{seq. counter}& $\leq 1$ & \multirow{3}{*}{Sinz \cite{sinz2005towards}} & $n-1$ & $3n-4$ & yes \\ \cline{2-2} \cline{4-6}
                             & $\leq k$ &                                              & $k(n-1)$ & $2nk + n - 3k - 1$ & yes \\ \cline{1-2} \cline{4-6}
par. counter                 & $\leq k$ &                                              & $2n-2$  & $7n-3\floor{\log n}-6$  & no  \\ \hline
\multirow{2}{*}{BDDs}        & $\leq k$ & Bailleux et al. \cite{bailleux} & $O(n^2)$ & $O(n^2)$ & yes \\ \cline{2-6}
                             & PB       & Ab{\'\i}o et al. \cite{abio2012} & $O(n^3\log(a_{\max}))$ & $O(n^3\log(a_{\max}))$ & yes \\ \hline
\multirow{2}{*}{sort. net.}  & $\leq k$ & \multirow{2}{*}{\makecell{E{\'e}n \& \\ S\"orensson \cite{minisatp}}} & $O(n\log^2 n)$ & $O(n\log^2 n)$ & yes \\ \cline{2-2} \cline{4-6}
                             & PB       &   & $O((\sum a_i) \log^2 (\sum a_i))$ & $O((\sum a_i) \log^2 (\sum a_i))$ & yes \\ \hline
sel. net.                    & $\leq k$ & \cite{codish2010pairwise,karpinski2015smaller,asin2011cardinality,abio2013parametric} & $O(n\log^2 k)$ & $O(n\log^2 k)$ & yes \\
\end{tabular}
\caption{Comparison of different encodings for at-most-one, at-most-k and Pseudo-Boolean constraints. We report on the number of new variables that needs to be introduced, the number of generated clauses and whether an encoding achieves some form of arc-consistency.}
\label{tbl:encodings}
\end{table}

One can also compare different encodings based on other measures. For example, Chen \cite{chen2010new} reports that his product
AMO encoding is better than sequential AMO encoding and binary AMO encoding in terms of total number of literals appearing
in the clauses. Chen's product encoding requires $4n + 8\sqrt{n} + O(\sqrt[4]{n})$ literals, while sequential encoding and binary
encoding requires $6n-8$ and $2n\log n$ literals, respectively.

For at-most-one constraints one can also take a look at a very interesting, recent, theoretical result by Ku{\v{c}}era et al. \cite{kuvcera2018lower}.
In their paper authors show a lower bound for the number of clauses the encoding for AMO constraints needs to have in order to preserve the
{\em complete propagation} property -- a generalization of arc-consistency in which we not only require consistency enforced on the input variables
(as in Definition \ref{def:gac}) but for the auxiliary variables as well. The lower bound is $2n + \sqrt{n} - O(1)$
and the product encoding is the closest to that barrier.

\part{Pairwise Selection Networks}

\chapter[Pairwise Bitonic Selection Networks]{Pairwise Bitonic \\ Selection Networks}\label{ch:cp15}

    \def\nqueenssolution{Qd4, Qe2, Qf8, Qa5}
    \setchessboard{smallboard,labelleft=false,labelbottom=false,showmover=false,setpieces=\nqueenssolution}

    \begin{tikzpicture}[remember picture,overlay]
      \node[anchor=east,inner sep=0pt] at (current page text area.east|-0,3cm) {\chessboard};
    \end{tikzpicture}

It has already been observed that using selection networks instead of sorting
networks is more efficient for the encoding of cardinality constraints. Codish and
Zazon-Ivry \cite{codish2010pairwise} introduced Pairwise Cardinality Networks, which are
networks derived from pairwise sorting networks that express cardinality
constraints. Two years later, same authors \cite{zazonpairwise} reformulated the 
definition of Pairwise Selection Networks and proved that their sizes are never 
worse than the sizes of corresponding Odd-Even Selection Networks. To show the 
difference they plotted it for selected values of $n$ and $k$.

In this chapter we give a new construction of smaller selection networks that are 
based on the pairwise selection ones and we prove that the construction is correct. 
We estimate also the size of our networks and compute the difference in sizes 
between our selection networks and the corresponding pairwise ones. The difference 
can be as big as $n\log n / 2$ for $k = n/2$.

To simplify the presentation we assume that $n$ and $k$ are powers of $2$.
The networks in this chapter are presented such that the inputs can be over
any totally ordered set $X$. In the context of encoding Boolean constraints we would
like to set $X=\{0,1\}$, but the proofs in this chapter are general enough to work with
any $X$.

\section{Pairwise Selection Network}\label{sec:psn}

Here we present the basis for our constructions in this chapter (and the next chapter). It is called
the {\em Pairwise Selection Network} and it was created by Codish and Zazon-Ivry \cite{zazonpairwise}.
This class of networks uses a component called a {\em splitter}.

\begin{definition}[splitter]
  A comparator network $f^n$ is a {\em splitter} if for any sequence $\bar{x}~\in~X^n$,
  if $\bar{y} = f^n(\bar{x})$, then $\bar{y}_{\leftt}$ weakly dominates $\bar{y}_{\rightt}$.
\end{definition}

\begin{observation}
  The splitter (on $n$ inputs) -- denoted as $split^n$ from now on -- can be constructed by comparing
  inputs $i$ and $i + n/2$, for $i=1..n/2$ (see Figure \ref{fig:splitters}a).
\end{observation}

  \begin{algorithm}[t!]
    \caption{$pw\_sel^{n}_k$}\label{net:pw}
    \begin{algorithmic}[1]
      \Require {$\bar{x} \in X^{n}$; $n$ and $k$ are powers of $2$; $1 \leq k \leq n$}
      \Ensure{The output is top $k$ sorted and is a permutation of the inputs}
      \If {$k=1$}
        \Return $max^n(\bar{x})$
      \EndIf
      \If {$k=n$}
        \Return $oe\_sort^n(\bar{x})$
      \EndIf
      \State $\bar{y} \gets split^n(\bar{x})$
      \State $\bar{l} \gets pw\_sel^{n/2}_{\min(n/2,k)}(\bar{y}_{\leftt})$
      \State $\bar{r} \gets pw\_sel^{n/2}_{\min(n/2,k/2)}(\bar{y}_{\rightt})$
      \State \Return $pw\_merge^{n}_k(\bar{l} :: \bar{r})$
    \end{algorithmic}
  \end{algorithm}

  \begin{algorithm}[t!]
    \caption{$pw\_merge^{n}_k$}\label{net:pwmerge}
    \begin{algorithmic}[1]
      \Require {$\bar{l} :: \bar{r} \in X^{n}$; $|l|=|r|$; $\bar{l}$ is top $k$ sorted,
        $\bar{r}$ is top $k/2$ sorted and $\pref(k/2,\bar{l}) \succeq_{w} \pref(k/2,\bar{r})$;
        $n$ and $k$ are powers of $2$; $1 \leq k < n$}
      \Ensure{The output is top $k$ sorted and is a permutation of the inputs}
      \If {$n \leq 2$ {\bf or} $k=1$}
        \Return $\zip(\bar{l},\bar{r})$
      \EndIf
      \State $\bar{y}  \gets pw\_merge^{n/2}_{k/2}(\bar{l}_{\odd} :: \bar{r}_{\odd})$
      \State $\bar{y}' \gets pw\_merge^{n/2}_{k/2}(\bar{l}_{\even} :: \bar{r}_{\even})$
      \State $\bar{z}  \gets \zip(\bar{y}, \bar{y}')$, $z'_1=z_1$, $\bar{z}'_{2k..n}=\bar{z}_{2k..n}$
      \ForAll {$i \in \{1, \dots, k-1\}$}
        $\tuple{z'_{2i},z'_{2i+1}} \gets sort^2(z_{2i},z_{2i+1})$
      \EndFor
      \State \Return $\bar{z}'$
    \end{algorithmic}
  \end{algorithm}

  The construction is presented in Algorithm \ref{net:pw}. The sub-procedures used are: 
  $max^n$ -- select maximum element out of $n$ inputs, and
  $pw\_merge^n_k$ -- a {\em Pairwise Merging Network} (Algorithm \ref{net:pwmerge}).
  If $k=n$ we need to sort the input sequence, therefore we use the
  odd-even sorting network (Algorithm \ref{ex:net:oddeven}) in this case.
  The last step of Algorithm \ref{net:pw} produces a top $k$ sorted sequence given
  the outputs of the recursive calls.

  Notice that since we introduced a splitter (Step 3), in the recursive calls
  we need to select $k$ top elements from the first half of $\bar{y}$, but only $k/2$
  elements from the second half. The reason: $r_{k/2+1}$ cannot be one of the first
  $k$ largest elements of $\bar{l} :: \bar{r}$. First, $r_{k/2+1}$ is smaller
  than any one of $\tuple{r_1,\ldots,r_{k/2}}$ (by the definition of top $k$ sorted sequence),
  and second, $\tuple{l_1,\ldots,l_{k/2}}$ weakly dominates $\tuple{r_1,\ldots,r_{k/2}}$,
  so $r_{k/2+1}$ is smaller than any one of $\tuple{l_1,\ldots,l_{k/2}}$.
  From this argument we make the following observation:

  \begin{observation}
    If $\bar{l} \in X^{n/2}$ is top $k$ sorted,
    $\bar{r} \in X^{n/2}$ is top $k/2$ sorted and $\tuple{l_1,\ldots,l_{k/2}}$
    weakly dominates $\tuple{r_1,\ldots,r_{k/2}}$, then $k$ largest elements of $\bar{l} :: \bar{r}$ are in
    $\tuple{l_1,\ldots,l_k} :: \tuple{r_1,\ldots,r_{k/2}}$.
    \label{obs:second_step}
  \end{observation}

  We would like to note, that the number of comparators used in the merger is:
  $|pw\_merge^{n}_k|=k\log k - k + 1$. The detailed proof of correctness of network $pw\_sel^n_k$ can
  be found in Section 6 of \cite{zazonpairwise}. The networks in \cite{zazonpairwise}
  are given in the functional representation.

  \begin{theorem}
    Let $n,k \in \nat$, where $1 \leq k \leq n$ and let $n$ and $k$ be powers of 2.
    Then $|pw\_sel^n_k| \leq |oe\_sel^n_k|$, where $oe\_sel^n_k$ is the Odd-Even Selection Network,
    for which $|oe\_sel^n_k|= (n/4)(\log^2 k + 3\log k + 4) - k \log k - 1$.
  \end{theorem}

  \begin{proof}
    See Theorems 11 and 14 of \cite{zazonpairwise}.
  \end{proof}

  \section{Bitonic Selection Network}

  We now present the construction of the {\em Bitonic Selection Network}.
  We use it to estimate the sizes of our improved pairwise selection network from the next section.
  We begin with a useful property of splitters and bitonic sequences
  proved by Batcher:

  \begin{lemma}\label{lma:split_bit}
    If $\bar{b} \in X^n$ is bitonic and $\bar{y}=split^n(\bar{b})$,
    then $\bar{y}_{\leftt}$ and $\bar{y}_{\rightt}$ are bitonic and
    $\bar{y}_{\leftt} \succeq \bar{y}_{\rightt}$.
  \end{lemma}
  
  \begin{proof}
    See Appendix B of \cite{batcher1968sorting}.
  \end{proof}

  \begin{definition}[bitonic splitter]\label{def:bitonic_splitter}
    A comparator network $f^n$ is a bitonic splitter 
    if for any two sorted sequences $\bar{x},\bar{y}~\in~X^{n/2}$,
    if $\bar{z} = bit\_split^n(\bar{x}::\bar{y})$, then (1) $\bar{z}_{\leftt} \succeq \bar{z}_{\rightt}$
    and (2) $\bar{z}_{\leftt}$ and $\bar{z}_{\rightt}$ are bitonic.
  \end{definition}

  \begin{observation}
    We can construct a bitonic splitter $bit\_split^n$ by joining inputs $\tuple{i,n-i+1}$,
    for $i=1..n/2$, with a comparator (see Figure \ref{fig:splitters}c).
    Notice that this is a consequence of Lemma \ref{lma:split_bit}, because
    given two sorted sequences $\tuple{x_1,\dots,x_n}$ and $\tuple{y_1,\dots,y_n}$,
    a sequence $\tuple{x_1,\dots,x_n} :: \tuple{y_n,\dots,y_1}$ is bitonic.
    Size of a bitonic splitter is $|bit\_split^n| = n/2$.
  \end{observation}

  \begin{figure}[t!]
    \centering
    \subfloat[]{
      \includegraphics[width=0.21\textwidth]{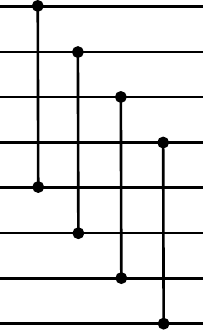}
    }
    ~~~~~~~~~
    \subfloat[]{
      \includegraphics[width=0.21\textwidth]{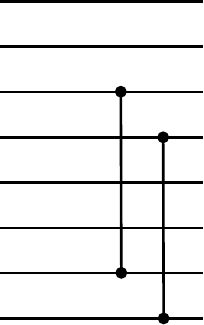}
    }
    ~~~~~~~~~
    \subfloat[]{
      \includegraphics[width=0.21\textwidth]{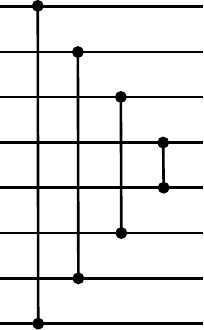}
    }
    \caption{a) splitter; b) half-splitter; c) bitonic splitter}
    \label{fig:splitters}
  \end{figure}

  We now present the procedure for construction of the Bitonic Selection Network.
  We use the odd-even sorting network $oe\_sort$ and the network $bit\_merge$ 
  (also by Batcher \cite{batcher1968sorting}) for sorting bitonic sequences, as black-boxes.
  As a reminder: $bit\_merge^n$ consists of two steps, first we use $\bar{y} = split^n(\bar{x})$,
  then recursively compute $bit\_merge^{n/2}$ for $\bar{y}_{\leftt}$ and $\bar{y}_{\rightt}$
  (base case, $n=2$, consists of a single comparator). In Figure \ref{fig:bit-vs-hbit}a we present $bit\_merge^{16}$.
  Size of this network is: $|bit\_merge^n|=n\log n/2$.
  Bitonic Selection Network $bit\_sel^n_k$ is constructed by the procedure given in Algorithm \ref{net:bit}.

  \begin{algorithm}[ht!]
    \caption{$bit\_sel^n_k$}\label{net:bit}
    \begin{algorithmic}[1]
      \Require {$\bar{x} \in X^{n}$; $n$ and $k$ are powers of $2$; $1 \leq k \leq n$}
      \Ensure{The output is top $k$ sorted and is a permutation of the inputs}
      \State {\bf let} $l=n/k$ {\bf and} $\bar{r} = \tuple{}$
      \ForAll {$i \in \{0, \dots, l-1\}$}
        $B_{i+1} \gets oe\_sort^k(\tuple{x_{ik+1},\dots,x_{(i+1)k}})$
      \EndFor
      \While {$l>1$}
        \ForAll {$i \in \{1,3,\ldots,l-1\}$}
          \State $\bar{y}^i \gets bit\_split^{2k}(B_i :: B_{i+1})$
          \State $B_{\lceil i/2 \rceil}' \gets bit\_merge^k(\bar{y}^i_{\leftt})$
          \State $\bar{r} \gets \bar{r} :: \bar{y}^i_{\rightt}$ \Comment{ residue elements}
          \State {\bf let} $l=l/2$ {\bf and} relabel $B_i'$ to $B_i$, for $1 \leq i \leq l$.
        \EndFor
      \EndWhile
      \State \Return $B_1 :: \bar{r}$
    \end{algorithmic}
  \end{algorithm}

  First, we partition input $\bar{x}$ into $l$ consecutive blocks, each of size $k$,
  then we sort each block with $oe\_sort^k$, obtaining $B_1,\ldots,B_l$. Then, we collect blocks
  into pairs $\tuple{B_1,B_2},\ldots,\tuple{B_{l-1},B_{l}}$ and perform a bitonic splitter on each of them.
  By Lemma \ref{lma:split_bit} $k$ largest elements in $\bar{y}^i$ are in $\bar{y}^i_{\leftt}$,
  and $\bar{y}^i_{\leftt}$ is bitonic, therefore we can use bitonic merger (Step 6) to sort it.
  The algorithm continues until one block remains.

  \begin{theorem}
    A comparator network $bit\_sel^n_k$ is a selection network.
  \end{theorem}

  \begin{proof}
    Let $\bar{x} \in X^n$ be the input to $bit\_sel^n_k$. After Step 2 we get sorted sequences
    $B_1,\ldots,B_l$, where $l=n/k$. Let $l_m$ be the value of $l$ after $m$ iterations of the loop in Step 3.
    Let $B^m_1,\ldots,B^m_{l_m}$ be the blocks after $m$ iterations. We prove
    by induction that:

    \begin{center}
    {\em $P(m)$: if $B_1,\ldots,B_{l}$ are sorted and are containing $k$ largest elements of $\bar{x}$,
      then after $m$-th iteration of the loop in Step 3: $l_m=l/2^m$, $B^m_1,\ldots,B^m_{l_m}$
      are sorted and are containing $k$ largest elements of $\bar{x}$.}
    \end{center}

    \noindent If $m=0$, then $l_0=l$, so $P(m)$ holds. We show that
    $\forall_{m\geq 0}$ $(P(m) \Rightarrow P(m+1))$. Consider $(m+1)$-th iteration of the while loop.
    By the induction hypothesis $l_m=l/2^m$, $B^m_1,\ldots,B^m_{l_m}$
    are sorted and are containing $k$ largest elements of $\bar{x}$.  We show that
    $(m+1)$-th iteration does not remove any element from $k$ largest elements of $\bar{x}$.
    To see this, notice that if $\bar{y}^i = bit\_split^{2k}(B^m_i :: B^m_{i+1})$
    (for $i \in \{1,3,\ldots,l_m-1\}$), then $\bar{y}^i_{\leftt} \succeq \bar{y}^i_{\rightt}$ and
    that $\bar{y}^i_{\leftt}$ is bitonic (by Definition \ref{def:bitonic_splitter}).
    Because of those two facts, $\bar{y}^i_{\rightt}$ is discarded and $\bar{y}^i_{\leftt}$
    is sorted using $bit\_merge^k$. After this, $l_{m+1}=l_m/2=l/2^{m+1}$ and 
    blocks $B^{m+1}_1,\ldots,B^{m+1}_{l_{m+1}}$ are sorted. Thus $P(m+1)$ is true.

    Since $l=n/k$, then by $P(m)$ we see that the while loop terminates after $m=\log \frac{n}{k}$
    iterations and that $B_1$ is sorted and contains $k$ largest elements of $\bar{x}$.
  \end{proof}

  \begin{figure}[ht]
    \begin{center}
      \includegraphics[scale=0.5]{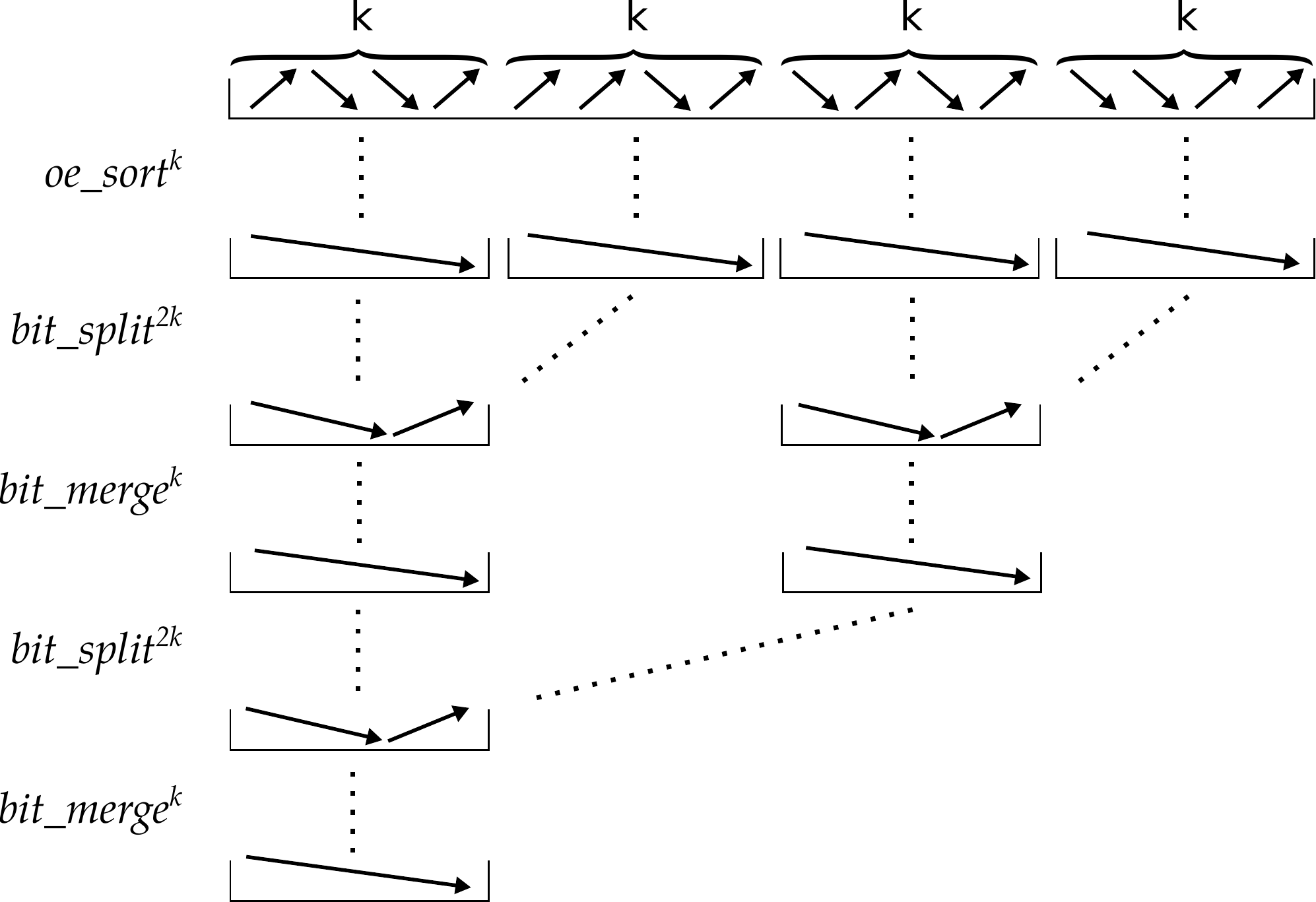}
    \end{center}
    \caption{A Bitonic Selection Network -- a construction diagram}
    \label{fig:bit-sel}
  \end{figure}

  A construction diagram of the Bitonic Selection Network is shown in Figure \ref{fig:bit-sel}. The size of the Bitonic Selection Network is:

  \begin{align}
    |bit\_sel^n_k| &= \frac{n}{k}|oe\_sort^k| + \left(\frac{n}{k}-1\right)(|bit\_split^{2k}|+|bit\_merge^k|) \notag \\ 
    &= \frac{1}{4}n\log^2k + \frac{1}{4}n\log k + 2n - \frac{1}{2}k\log k - k - \frac{n}{k} \label{eq:bit}
  \end{align}

  \noindent The $|oe\_sort^k|$ was already shown in Example \ref{ex:oe_func}. The rest is a straightforward calculation.
  
  In Figure \ref{fig:oe-sel} we present a Bitonic Selection Network for $n=8$ and $k=2$.

  \begin{figure}[ht]
    \begin{center}
      \includegraphics[width=0.30\textwidth]{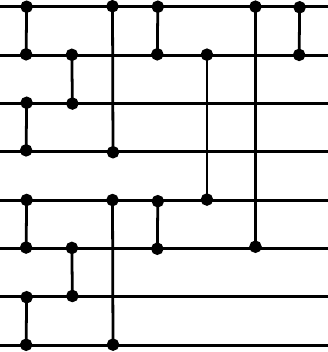}
    \end{center}
    \caption{A Bitonic Selection Network for $n=8$ and $k=2$}\label{fig:oe-sel}
  \end{figure}

\section{Pairwise Bitonic Selection Network}

  As mentioned in Section \ref{sec:psn}, only the first $k/2$ elements from the second half
  of the input are relevant when we get to the merging step in $pw\_sel^n_k$.
  We exploit this fact to create a new, smaller merger.
  We use the concept of bitonic sequences, therefore we call the new merger
  $pw\_bit\_merge^n_k$ and the new selection network $pw\_bit\_sel^n_k$
  (the {\em Pairwise Bitonic Selection Network}). The network $pw\_bit\_sel^n_k$ is generated
  by substituting the last step of $pw\_sel^n_k$ with $pw\_bit\_merge^n_k$.
  The new merger is presented as Algorithm \ref{net:pw_merge}.

  \begin{algorithm}[t!]
    \caption{$pw\_bit\_merge^n_k$}\label{net:pw_merge}
    \begin{algorithmic}[1]
      \Require {$\bar{l} :: \bar{r} \in X^{n}$; $\bar{l}$ is top $k$ sorted,
        $\bar{r}$ is top $k/2$ sorted;
        $\pref(k/2,\bar{l}) \succeq_{w} \pref(k/2,\bar{r})$; $k$ is a power of $2$}
      \Ensure{The output is top $k$ sorted and is a permutation of the inputs}
      \State $\bar{y} \gets bit\_split^k(l_{k/2+1},\ldots,l_k,r_1,\ldots,r_{k/2})$
      \State $\bar{b} \gets \tuple{l_1,\ldots,l_{k/2}} :: \tuple{y_1,\ldots,y_{k/2}}$
      \State $\bar{p} \gets \suff(k/2,\bar{y}) :: \suff(n/2-k,\bar{l}) :: \suff(n/2-k/2,\bar{r})$ \Comment{ residue elements}
      \State \Return $bit\_merge^k(\bar{b}) :: \bar{p}$
    \end{algorithmic}
  \end{algorithm}

  \begin{theorem}\label{thm:pw_merge}
    The output of Algorithm \ref{net:pw_merge} consists of sorted $k$
    largest elements from input $\bar{l} :: \bar{r}$, assuming
    that $\bar{l} \in X^{n/2}$ is top $k$ sorted and $\bar{r} \in X^{n/2}$
    is top $k/2$ sorted and $\tuple{l_1,\ldots,l_{k/2}}$ weakly dominates $\tuple{r_1,\ldots,r_{k/2}}$.
  \end{theorem}

  \begin{proof}
    We have to prove two things: (1) $\bar{b}$ is bitonic and (2) $\bar{b}$ consists of
    $k$ largest elements from $\bar{l} :: \bar{r}$.

    (1) Let $j$ be the last index in the sequence $\tuple{k/2+1,\ldots,k}$, for which $l_j > r_{k-j+1}$.
    If such $j$ does not exist, then $\tuple{y_1,\ldots,y_{k/2}}$ is non-decreasing,
    hence $\bar{b}$ is bitonic (non-decreasing). Assume that $j$ exists, then
    $\tuple{y_{j-k/2+1},\ldots,y_{k/2}}$ is non-decreasing and $\tuple{y_1,\ldots,y_{k-j}}$
    is non-increasing. Adding the fact that $l_{k/2} \geq l_{k/2+1} = y_1$ proves,
    that $\bar{b}$ is bitonic (v-shaped).

    (2) By Observation \ref{obs:second_step}, it is sufficient to prove
    that $\bar{b} \succeq \tuple{y_{k/2+1},\ldots,y_k}$.
    Since $\forall_{k/2<j\leq k}$ $l_{k/2} \geq l_j \geq \min\{l_j, r_{k-j+1}\}=y_{3k/2-j+1}$,
    then $\tuple{l_1,\ldots,l_{k/2}} \succeq \tuple{y_{k/2+1},\ldots,y_k}$
    and by Definition \ref{def:bitonic_splitter}:
    $\langle y_1,$ $\ldots,$ $y_{k/2}\rangle \succeq \tuple{y_{k/2+1},\ldots,y_k}$. 
    Therefore $\bar{b}$ consists of $k$ largest elements from $\bar{l} :: \bar{r}$.

    The bitonic merger in Step 4 receives a bitonic sequence, so it outputs a sorted sequence,
    which completes the proof.
  \end{proof}

  The first step of improved pairwise merger is illustrated in Figure \ref{fig:pw_bit}.
  We use $k/2$ comparators in the first step and $k\log k/2$
  comparators in the last step. We get a merger of size $k\log k/2 + k/2$, which is better
  than the previous approach ($k\log k - k + 1$). In the following we show
  that we can do even better and eliminate the $k/2$ term.

  \begin{figure}[ht]
    \begin{center}
      \includegraphics[width=\textwidth]{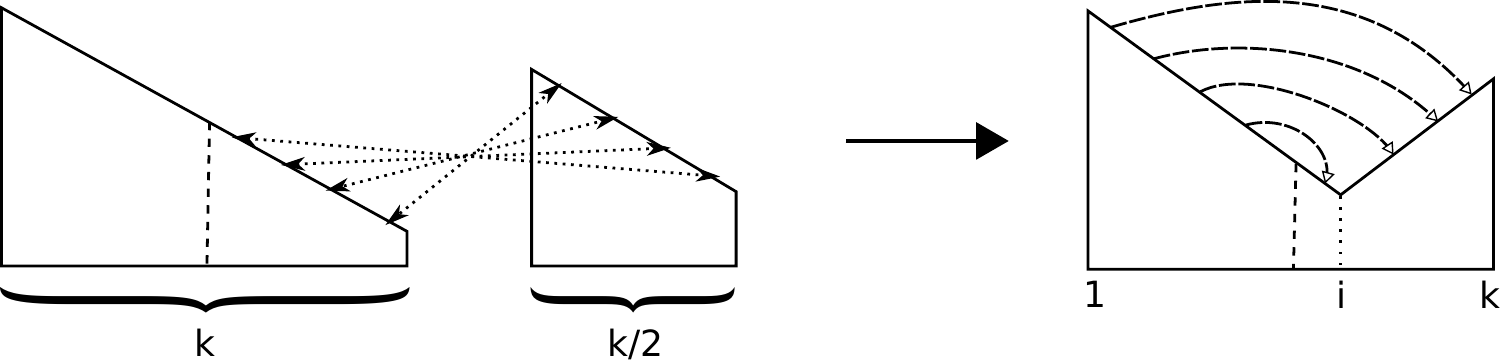}
    \end{center}
    \caption{Constructing a bitonic sequence. Arrows on the right picture show directions of inequalities.
      Sequence on the right is v-shape s-dominating at point $i$.}
    \label{fig:pw_bit}
  \end{figure}

  The main observation is that the result of the first step of $pw\_bit\_merge$ operation 
  $\tuple{b_1,\ldots,b_k}$ is not only bitonic, but what we call {\em v-shape s-dominating}.

  \begin{definition}[s-domination]\label{def:sd}
    A sequence $\bar{b} =\tuple{b_1,\ldots,b_k}$ is s-dominating if 
    $\forall_{1\leq j \leq k/2}$ $b_j \geq b_{k-j+1}$.
  \end{definition} 

  \begin{lemma}\label{lma:sd}
    If $\bar{b} =\tuple{b_1,\ldots,b_k}$ is v-shaped and s-dominating, then (1) $\bar{b}$ is non-increasing
    or (2) $\exists_{k/2<i<k} \; b_i < b_{i+1}$.
  \end{lemma}

  \begin{proof}
    Assume that $\bar{b}$ is not non-increasing. Then $\exists_{1\leq j<k} \; b_j < b_{j+1}$.
    Assume that $j\leq k/2$. Since $\bar{b}$ is v-shaped, $b_{j+1}$ must be in non-decreasing
    part of $\bar{b}$. If follows that $b_j < b_{j+1} \leq \ldots \leq b_{k/2} \leq \ldots \leq b_{k-j+1}$.
    That means that $b_j<b_{k-j+1}$. On the other hand, $\bar{b}$ is s-dominating, thus $b_j \geq b_{k-j+1}$
    -- a contradiction.
  \end{proof}
  
  We say that a sequence $\bar{b}$ is {\em v-shape s-dominating at point $i$}
  if $i$ is the smallest index greater than $k/2$ such that $b_i < b_{i+1}$ or $i=k$
  for a non-increasing sequence.

  \begin{lemma}
    Let $\bar{b}=\tuple{b_1,\ldots,b_k}$ be v-shape s-dominating at point $i$, then
    $\tuple{b_1,\ldots,b_{k/4}} \succeq \langle b_{k/2+1},$ $\ldots,$ $b_{3k/4}\rangle$.
    \label{lma:transit}
  \end{lemma}

  \begin{proof}
    If $\bar{b}$ is non-increasing, then the lemma holds. From Lemma \ref{lma:sd}: $k/2<i<k$. 
    If $i > 3k/4$, then by Definition \ref{def:bit}: $b_1 \geq \ldots \geq b_{3k/4} \geq \ldots \geq b_i$,
    so lemma holds. If $k/2 < i \leq 3k/4$, then by Definition \ref{def:bit}:  $b_1 \geq \ldots \geq b_i$,
    so $\tuple{b_1,\ldots,b_{k/4}} \succeq \tuple{b_{k/2+1},\ldots,b_{i}}$.
    Since $b_i < b_{i+1} \leq \ldots \leq b_{3k/4}$, it suffices to prove that $b_{k/4} \geq b_{3k/4}$.
    By Definition \ref{def:sd} and \ref{def:bit}: $b_{k/4} \geq b_{3k/4+1} \geq b_{3k/4}$.
  \end{proof}

  \begin{definition}[half-splitter]
    A {\em half-splitter} is a comparator network constructed by comparing inputs 
    $\tuple{k/4+1,3k/4+1},\ldots,\tuple{k/2,k}$ (normal splitter with first $k/4$ comparators removed; see Figure \ref{fig:splitters}b).
    We call it $half\_split^k$.
    \label{def:hs}
  \end{definition}

  \begin{lemma}
    If $\bar{b}$ is v-shape s-dominating, then $half\_split^k(\bar{b})=split^k(\bar{b})$.
    \label{lma:hs_s}
  \end{lemma}

  \begin{proof}
    Directly from Lemma \ref{lma:transit}.
  \end{proof}

  \begin{lemma}\label{lma:dom}
    Let $\bar{b}=\tuple{b_1,\ldots,b_k}$ be v-shape s-dominating. Let $\bar{w}=half\_split^k(\bar{b})$. 
    The following statements are true: (1) $\bar{w}_{\leftt}$ is v-shape s-dominating;
    (2) $\bar{w}_{\rightt}$ is bitonic; (3) $\bar{w}_{\leftt}$ $\succeq \bar{w}_{\rightt}$.
  \end{lemma}

  \begin{proof}
    (1) Let $\bar{y}=\bar{w}_{\leftt}$. First, we show that $\bar{y}$ is v-shaped.
    If $\bar{y}$ is non-increasing, then it is v-shaped. Otherwise, let $j$ be the first index
    from the range $\{1,\ldots,k/2\}$, where $y_{j-1}<y_j$. Since $y_j=\max\{b_j,b_{j+k/2}\}$ and
    $y_{j-1} \geq b_{j-1} \geq b_j$, thus $b_j < b_{j+k/2}$.
    Since $\bar{b}$ is v-shaped, element $b_{j+k/2}$ must be in non-decreasing part of $\bar{b}$.
    It follows that $b_{j} \geq \ldots \geq b_{k/2}$ and $b_{j+k/2} \leq \ldots \leq b_k$.
    From this we can see that $\forall_{j\leq j' \leq k/2}$ $y_{j'}=\max\{b_{j'},b_{j'+k/2}\} = b_{j'+k/2}$,
    so $y_j \leq \ldots \leq y_{k/2}$. Therefore $\bar{y}$ is v-shaped.

    Next, we show that $\bar{y}$ is s-dominating. Consider any $j$, where $1 \leq j \leq k/4$.
    By Definition \ref{def:bit} and \ref{def:sd}: $b_j \geq b_{k/2-j+1}$ and $b_j \geq b_{k-j+1}$,
    therefore $y_j = b_j \geq \max\{b_{k/2-j+1},b_{k-j+1}\} = y_{k/2-j+1}$,
    thus proving that $\bar{y}$ is s-dominating. Concluding: $\bar{y}$ is v-shape s-dominating.
    
    (2) Let $\bar{z}=\bar{w}_{\rightt}$. By Lemma \ref{lma:hs_s}:
    $\bar{z}=split^k(\bar{b})_{\rightt}$. We know that $\bar{b}$ is a special case of bitonic sequence,
    therefore using Lemma \ref{lma:split_bit} we get that $\bar{z}$ is bitonic.

    (3) By Lemma \ref{lma:hs_s}: $\bar{w}=split^k(\bar{b})$.
    We know that $\bar{b}$ is a special case of bitonic sequence, therefore using
    Lemma \ref{lma:split_bit} we get $\bar{w}_{\leftt} \succeq \bar{w}_{\rightt}$.
  \end{proof}

  Using $half\_split$ and Batcher's $bit\_merge$ and successively applying Lemma \ref{lma:dom} to the
  resulting v-shape s-dominating half of the output, we have all the tools needed to construct the improved
  pairwise merger using half-splitters, which we present as Algorithms \ref{net:pw_merge2}
  and \ref{net:hbit_merge}.

  \begin{algorithm}[t!]
    \caption{$pw\_hbit\_merge^n_k$}\label{net:pw_merge2}
    \begin{algorithmic}[1]
      \Require {$\bar{l} :: \bar{r} \in X^{n}$; $\bar{l}$ is top $k$ sorted,
        $\bar{r}$ is top $k/2$ sorted;
        $\pref(k/2,\bar{l}) \succeq_{w} \pref(k/2,\bar{r})$; $k$ is a power of $2$}
      \Ensure{The output is top $k$ sorted and is a permutation of the inputs}
      \State $\bar{y} \gets bit\_split^k(l_{k/2+1},\ldots,l_k,r_1,\ldots,r_{k/2})$
      \State $\bar{b} \gets \tuple{l_1,\ldots,l_{k/2}} :: \tuple{y_1,\ldots,y_{k/2}}$
      \State $\bar{p} \gets \suff(k/2,\bar{y}) :: \suff(n/2-k,\bar{l}) :: \suff(n/2-k/2,\bar{r})$ \Comment{ residue elements}
      \State \Return $half\_bit\_merge^k(\bar{b}) :: \bar{p}$
    \end{algorithmic}
  \end{algorithm}

  \begin{algorithm}[t!]
    \caption{$half\_bit\_merge^k$}\label{net:hbit_merge}
    \begin{algorithmic}[1]
      \Require {$\bar{b} \in X^{k}$; $\bar{b}$ is v-shaped s-dominating, $k$ is a power of $2$}
      \Ensure{The output is sorted and is a permutation of the inputs}
      \If {$k=2$}
        \Return $\bar{b}$
      \EndIf
      \State $\bar{b'} \gets half\_split(b_1,\ldots,b_k)$
      \State $\bar{l'} \gets half\_bit\_merge^{k/2}(\bar{b'}_{\leftt})$
      \State $\bar{r'} \gets bit\_merge^{k/2}(\bar{b'}_{\rightt})$
      \State \Return $\bar{l'}::\bar{r'}$
    \end{algorithmic}
  \end{algorithm}

  In Figure \ref{fig:bit-vs-hbit} the difference between bitonic and half-bitonic merger is shown for $n=16$.
  The following theorem states that the construction of $pw\_hbit\_merge^n_k$ is correct.
  
  \begin{figure}[ht]
    \centering
    \subfloat[]{
      \includegraphics[width=0.45\textwidth]{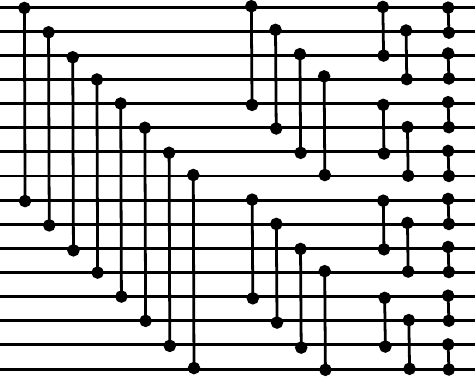}
    }
    ~
    \subfloat[]{
      \includegraphics[width=0.435\textwidth]{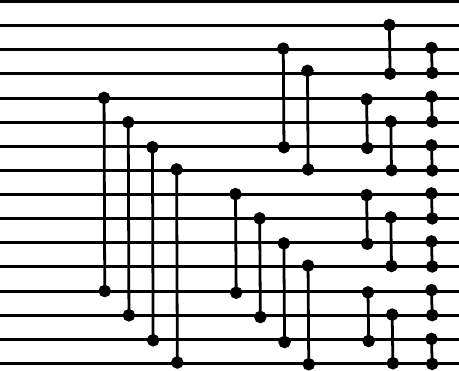}
    }
    \caption{a) bitonic merging network; b) half-bitonic merging network; $n=16$}
    \label{fig:bit-vs-hbit}
  \end{figure}

  \begin{theorem}\label{thm:pw_bit_merge}
    The output of Algorithm \ref{net:pw_merge2} consists of sorted $k$ largest elements from input
    $\bar{l} :: \bar{r}$, assuming that $\bar{l} \in X^{n/2}$ is top $k$ sorted
    and $\bar{r} \in X^{n/2}$ is top $k/2$ sorted and $\tuple{l_1,\ldots,l_{k/2}}$
    weakly dominates $\tuple{r_1,\ldots,r_{k/2}}$. Also, $|pw\_hbit\_merge^n_k|=k \log k/2$.
  \end{theorem}

  \begin{proof}
    Since Step 1 in Algorithm \ref{net:pw_merge2} is the same as in Algorithm \ref{net:pw_merge},
    we can reuse the proof of Theorem \ref{thm:pw_merge} to deduce, that $\bar{b}$ is v-shaped and is
    containing $k$ largest elements from $\bar{l} :: \bar{r}$. Also, since $\forall_{1\leq j \leq k/2}$
    $l_j \geq l_{k-j+1}$ and $l_j \geq r_j$, then $b_j = l_j \geq \max\{l_{k-j+1},r_j\} = b_{k-j+1}$,
    so $\bar{b}$ is s-dominating.
    
    We prove by the induction on $k$, that if $\bar{b}$ is v-shape s-dominating,
    then $half\_bit\_merge^k(\bar{b})$ is sorted.
    For the base case, consider $k=2$ and a v-shape s-dominating sequence $\tuple{b_1,b_2}$.
    By Definition \ref{def:sd} this sequence is already sorted and we are done. For the induction step,
    consider $\bar{b'} = half\_split^k(\bar{b})$. By Lemma \ref{lma:dom} we get that $\bar{b'}_{\leftt}$
    is v-shape s-dominating and $\bar{b'}_{\rightt}$ is bitonic. Using the induction
    hypothesis we sort $\bar{b'}_{\leftt}$ and using bitonic merger we sort $\bar{b'}_{\rightt}$.
    By Lemma \ref{lma:dom}: $\bar{b'}_{\leftt} \succeq \bar{b'}_{\rightt}$, which completes
    the proof of correctness. 

    As mentioned in Definition \ref{def:hs}: $half\_split^k$ is just $split^k$ with the first $k/4$
    comparators removed. So $half\_bit\_merge^k$ is just $bit\_merge^k$ with some of the comparators removed.
    Let us count them: in each level of recursion step we take half of comparators from $split^k$ and
    additional one comparator from the base case ($k=2$). We sum them together to get:

    \begin{equation*}
      1 + \sum_{i=0}^{\log k - 2}\frac{k}{2^{i+2}} = 1 + \frac{k}{4}\left(\sum_{i=0}^{\log k - 1}\left(\frac{1}{2}\right)^i - \frac{2}{k}\right) = 1 + \frac{k}{4}\left(2 - \frac{2}{k} - \frac{2}{k} \right) = \frac{k}{2}
    \end{equation*}

   \noindent Therefore, we have:
 
    \[
      |pw\_hbit\_merge^n_k| = k/2 + k \log k/2 - k/2 = k \log k/2
    \]

  \end{proof}
  
  The only difference between $pw\_sel$ and our $pw\_hbit\_sel$ is the use of improved merger
  $pw\_hbit\_merge$ rather than $pw\_merge$. By Theorem \ref{thm:pw_bit_merge},
  we can conclude that $|pw\_merge^n_k|$ $\geq$ $|pw\_hbit\_merge^n_k|$, so it follows that:

  \begin{corollary}
    For $1 \leq k \leq n$, $|pw\_hbit\_sel^n_k| \leq |pw\_sel^n_k|$.
  \end{corollary}

  \section{Sizes of New Selection Networks}

  In this section we estimate the size of $pw\_hbit\_sel^n_k$. To this end we show that
  the size of $pw\_hbit\_sel^n_k$ is upper-bounded by the size of $bit\_sel^n_k$ and
  use this fact in our estimation. We also compute the exact difference between
  sizes of $pw\_sel^n_k$ and $pw\_hbit\_sel^n_k$ and show that it can be as big as $n\log n/2$.
  
  We have the recursive formula for the number of comparators of $pw\_hbit\_sel^n_k$:

  \begin{equation}
    |pw\_hbit\_sel^n_k| = \left\{ 
    \begin{array}{l l}
      |pw\_hbit\_sel^{n/2}_k| + |pw\_hbit\_sel^{n/2}_{k/2}|+ &  \\ 
      + |split^n| + |pw\_hbit\_merge^k| & \quad \text{if $k<n$}\\
      |oe\_sort^k| & \quad \text{if $k=n$} \\
      |max^n| & \quad \text{if $k=1$} \\
    \end{array} \right.
  \label{eq:pw_cp15}
  \end{equation}

  \begin{lemma}
    For $1 \leq k < n$ (both powers of 2), $|pw\_hbit\_sel^n_k| \leq |bit\_sel^n_k|$.
    \label{lma:xyz}
  \end{lemma}

  \begin{proof}
    Let $aux\_sel^n_k$ be the comparator network that is generated by substituting recursive calls in $pw\_hbit\_sel^n_k$
    by calls to $bit\_sel^n_k$. Size of this network (for $1<k<n$) is:

    \begin{equation}
      |aux\_sel^n_k| = |bit\_sel^{n/2}_k| + |bit\_sel^{n/2}_{k/2}| + |split^n| + |pw\_hbit\_merge^k|
      \label{eq:aux}
    \end{equation}

    \noindent Lemma \ref{lma:xyz} follows from Lemma \ref{lma:main1} and Lemma \ref{lma:main2} below, where we show that:

    \[
      |pw\_hbit\_sel^n_k| \leq |aux\_sel^n_k| \leq |bit\_sel^n_k|
    \]
  \end{proof}

  \begin{lemma}
    For $1 < k < n$ (both powers of 2), $|aux\_sel^n_k| \leq |bit\_sel^n_k|$.
    \label{lma:main1}
  \end{lemma}

  \begin{proof}
    We compute both values from Eq. \ref{eq:bit} and Eq. \ref{eq:aux}:

    \begin{align*}
      |aux\_sel^n_k|&= \frac{1}{4}n\log^2k + \frac{5}{2}n - \frac{1}{4}k\log k - \frac{5}{4}k - \frac{3n}{2k} \\
      |bit\_sel^n_k|&= \frac{1}{4}n\log^2k + \frac{1}{4}n\log k + 2n - \frac{1}{2}k\log k - k - \frac{n}{k}
    \end{align*}

    \noindent We simplify both sides to get the following inequality:

    \[
      n - \frac{1}{2}k - \frac{n}{k} \leq \frac{1}{2}(n-k)\log k
    \]

    \noindent which can be easily proved by induction.
  \end{proof}

  \begin{lemma}
    For $1 \leq k < n$ (both powers of 2), $|pw\_hbit\_sel^n_k| \leq |aux\_sel^n_k|$.
    \label{lma:main2}
  \end{lemma}

  \begin{proof}
    By induction. For the base case, consider $1=k<n$. If follows by definitions that
    $|pw\_hbit\_sel^n_k|=|aux\_sel^n_k|=n-1$.
    For the induction step assume that for each $(n',k') \prec (n,k)$
    (in lexicographical order) the lemma holds, we get:

    \[
    \begin{tabular}{l r}
      \multicolumn{2}{l}{$|pw\_hbit\_sel^n_k|$} \\
      \multicolumn{2}{l}{$=|pw\_hbit\_sel^{n/2}_{k/2}| + |pw\_hbit\_sel^{n/2}_k|+|split^n|+|pw\_hbit\_merge^k|$} \\
      & \small{\bf (by the definition of $pw\_hbit\_sel$)} \\
      \multicolumn{2}{l}{$\leq|aux\_sel^{n/2}_{k/2}| + |aux\_sel^{n/2}_k|+|split^n|+|pw\_hbit\_merge^k|$} \\
      & \small{\bf (by the induction hypothesis)} \\
      \multicolumn{2}{l}{$\leq|bit\_sel^{n/2}_{k/2}| + |bit\_sel^{n/2}_k|+|split^n|+|pw\_hbit\_merge^k|$} \\
      & \small{\bf (by Lemma \ref{lma:main1})} \\
      \multicolumn{2}{l}{$= |aux\_sel^n_k|$} \\
      & \small{\bf (by the definition of $aux\_sel$)} \\
    \end{tabular}
    \]
  \end{proof}

  Let $N=2^n$ and $K=2^k$. We compute the upper bound for $P(n,k)=|pw\_hbit\_sel^N_K|$ using $B(n,k)=|bit\_sel^N_K|$.
  First, we prove a technical lemma below. The value $P(n,k,m)$ denotes the number of comparators used in the network
  $pw\_hbit\_sel^N_K$ after $m$ levels of recursion (of Eq. \ref{eq:pw_cp15}). Notice that if $0<k<n$, then:

  \begin{equation}\label{eq:reccc}
    P(n,k) = P(n-1,k) + P(n-1,k-1) + k2^{k-1} + 2^{n-1}
  \end{equation}
 
  \noindent Term $k2^{k-1}$ corresponds to $|pw\_hbit\_merge^K|$ and $2^{n-1}$ to $|split^N|$.

  \begin{lemma}
    Let:
    
    \[
      P(n,k,m) = \sum_{i=0}^{m-1}\sum_{j=0}^i\binom{i}{j}\left((k-j)2^{k-j-1}+2^{n-i-1}\right) + \sum_{i=0}^{m}\binom{m}{i}P(n-m,k-i).
    \]

    \noindent Then $\forall_{0 \leq m \leq \min(k,n-k)}$ $P(n,k,m)=P(n,k)$.
    \label{lma:unfold}
  \end{lemma}

  \begin{proof}
    By induction on $m$. If $m=0$, then $P(n,k,0)=P(n,k)$. Choose any $m$ such that $0 \leq m < \min(k,n-k)$ and assume that
    $P(n,k,m)=P(n,k)$. We show that $P(n,k,m+1)=P(n,k)$. We have:

    \begin{align*}
      P(n,k,m+1) &= \sum_{i=0}^{(m-1)+1}\sum_{j=0}^i\binom{i}{j}\left((k-j)2^{k-j-1}+2^{n-i-1}\right) + \underbrace{\sum_{i=0}^{m+1}\binom{m+1}{i}P(n-m-1,k-i)}_{(\ref{eq:997})} \\
                 &= \sum_{i=0}^{m-1}\sum_{j=0}^i\binom{i}{j}\left((k-j)2^{k-j-1}+2^{n-i-1}\right) + \sum_{i=0}^m\binom{m}{i}\left((k-i)2^{k-i-1}+2^{n-m-1}\right) \\
                 &\quad+ \sum_{i=0}^{m}\binom{m}{i}\left(P(n-m-1,k-i) + P(n-m-1,k-i-1)\right) \\
                 &\stackrel{(\ref{eq:reccc})}{=} \sum_{i=0}^{m-1}\sum_{j=0}^i\binom{i}{j}\left((k-j)2^{k-j-1}+2^{n-i-1}\right) + \sum_{i=0}^{m}\binom{m}{i}P(n-m,k-i) \\
                 &= P(n,k,m) \stackrel{IH}{=} P(n,k)
    \end{align*}

    \begin{align}
      \sum_{i=0}^{m+1}&\binom{m+1}{i}P(n-m-1,k-i) \label{eq:997} \\
      &= P(n-m-1,k) + P(n-m-1,k-m-1) + \sum_{i=1}^{m}\left(\binom{m}{i} + \binom{m}{i-1}\right)P(n-m-1,k-i) \nonumber \\
      &= \left(P(n-m-1,k) + \sum_{i=1}^{m}\binom{m}{i}P(n-m-1,k-i)\right) \nonumber \\
       &\quad + \left(P(n-m-1,k-m-1) + \sum_{i=0}^{m-1}\binom{m}{i}P(n-m-1,k-i-1)\right) \nonumber \\
      &= \sum_{i=0}^{m}\binom{m}{i}\left(P(n-m-1,k-i) + P(n-m-1,k-i-1)\right) \nonumber
    \end{align}

  \end{proof}

  \begin{lemma}
    $P(n,k,m) \leq 2^{n-2}\!\left(\!\!\left(k-\frac{m}{2}\right)^2\! + k + \frac{7m}{4} + 8 \!\right)\!\! + 2^k\!\left(\frac{3}{2}\right)^m\!\left(\frac{k}{2} - \frac{m}{6}\right) - 2^k(k+1) - 2^{n-k}\left(\frac{3}{2}\right)^m$.
    \label{lma:bigineq}
  \end{lemma}

\begin{proof}
  The first inequality below is a consequence of Lemma \ref{lma:unfold} and \ref{lma:xyz}. We also use the following known equations:
  $\sum_{k=0}^n\binom{n}{k}x^{k-1}k = n (1 + x)^{n-1}$, 
  $\sum_{k=0}^n\binom{n}{k}k^2 =  n(n+1)2^{n-2}$,
  $\sum_{k=0}^{n-1}x^{k-1}k = \frac{(1-x)(-nx^{n-1})+(1-x^n)}{(1-x)^2}$.
  
    \begin{align*}
      P(n,k,m) &\leq \underbrace{\sum_{i=0}^{m-1}\sum_{j=0}^i\binom{i}{j}\left((k-j)2^{k-j-1}+2^{n-i-1}\right)}_{(\ref{eq:1})} + \underbrace{\sum_{i=0}^{m}\binom{m}{i}B(n-m,k-i)}_{(\ref{eq:2})}\\
         &= 
         \left(2^{k}\left(\frac{3}{2}\right)^{m}\left(k+1-\frac{m}{3}\right)-2^k(k+1)+m2^{n-1}\right)
          \\
         &\quad+ 2^{n-2}\left( k^2 - km + \frac{m(m-1)}{4} + k + 8 \right) \\
         &\quad+ 2^k\left(\frac{3}{2}\right)^m\left(-\frac{k}{2} + \frac{m}{6} - 1\right) 
         - 
         2^{n-k}\left(\frac{3}{2}\right)^m\\
 &= 2^{n-2}\left( \left(k-\frac{m}{2}\right)^2 + k + \frac{7m}{4} + 8 \right) + 2^k\left(\frac{3}{2}\right)^m\left(\frac{k}{2} - \frac{m}{6}\right) \\
  &\quad- 2^k(k+1) - 2^{n-k}\left(\frac{3}{2}\right)^m
    \end{align*}

    \begin{align}
      \sum_{i=0}^{m-1}\sum_{j=0}^i&\binom{i}{j}\left((k-j)2^{k-j-1}+2^{n-i-1}\right) \label{eq:1} \\
  &= \underbrace{\sum_{i=0}^{m-1}\sum_{j=0}^{i}\binom{i}{j}(k-j)2^{k-j-1}}_{(\ref{eq:1.1})}
  + \underbrace{\sum_{i=0}^{m-1}\sum_{j=0}^{i}\binom{i}{j}2^{n-i-1}}_{(\ref{eq:1.2})} \nonumber \\
  &= \left(2^{k}\middle(\frac{3}{2}\middle)^{m}(k+1-\frac{m}{3})-2^k(k+1)\right) + 
  (m2^{n-1}) \nonumber
    \end{align}

    \begin{align}
      \sum_{i=0}^{m-1}\sum_{j=0}^{i}\binom{i}{j}&(k-j)2^{k-j-1}
      = k2^{k-1}\sum_{i=0}^{m-1}\sum_{j=0}^{i}\binom{i}{j}2^{-j}- 2^{k-1}\sum_{i=0}^{m-1}\sum_{j=0}^{i}\binom{i}{j}2^{-j}j \label{eq:1.1} \\
      &= k2^{k-1}\sum_{i=0}^{m-1}\left(\frac{3}{2}\right)^i - 2^{k-1}\frac{1}{2}\sum_{i=0}^{m-1}\left(\frac{3}{2}\right)^{i-1}i \nonumber \\
      &= k2^{k-1}2\left(\left(\frac{3}{2}\right)^m-1\right) - 2^{k-1}\left(2 - \left(\frac{3}{2}\right)^{m-1}(3-m)\right) \nonumber \\
      &= 2^k\left(\frac{3}{2}\right)^m\left(k + 1 - \frac{m}{3}\right) - 2^k(k+1) \nonumber
    \end{align}

    \begin{equation}
      \sum_{i=0}^{m-1}\sum_{j=0}^{i}\binom{i}{j}2^{n-i-1}=\sum_{i=0}^{m-1}\left(2^{n-i-1}\sum_{j=0}^{i}\binom{i}{j}\right)=\sum_{i=0}^{m-1}2^{n-i-1}2^i=m2^{n-1}
      \label{eq:1.2}
    \end{equation}

    \begin{align}
      &\sum_{i=0}^{m}\binom{m}{i}(2^{n-m-2}(k-i)^2 + 2^{n-m-2}(k-i) - 2^{k-i-1}(k-i) \nonumber \\
      &\quad\quad\quad - 2^{n-m-k+i} + 2^{n-m+1}-2^{k-i}) \label{eq:2} \\
      &= \sum_{i=0}^{m}\binom{m}{i}2^{n-m-2}(k-i)^2 +\sum_{i=0}^{m}\binom{m}{i}2^{n-m-2}(k-i) -\sum_{i=0}^{m}\binom{m}{i}2^{k-i-1}(k-i)\nonumber \\
      &\quad- \sum_{i=0}^{m}\binom{m}{i}2^{n-m-k+i}+\sum_{i=0}^{m}\binom{m}{i}2^{n-m+1} - \sum_{i=0}^{m}\binom{m}{i}2^{k-i}\nonumber \\
      &= 2^{n-m-2}(k^22^m-km2^m+m(m+1)2^{m-2})+2^{n-m-2}(k2^m-m2^{m-1}) \nonumber \\
      &\quad- 2^{k-1}\left(k\left(\frac{3}{2}\right)^m-2^{-m}3^{m-1}m\right)- 2^{n-m-k}3^m + 2^{n+1} - 2^k\left(\frac{3}{2}\right)^m \nonumber \\
      &= 2^{n-2}\left( k^2 - km + \frac{m(m-1)}{4} + k + 8 \right) \nonumber \\
      &\quad+ 2^k\left(\frac{3}{2}\right)^m\left(-\frac{k}{2} + \frac{m}{6} - 1\right) - 2^{n-k}\left(\frac{3}{2}\right)^m \nonumber
    \end{align}

\end{proof}

\begin{theorem}
  For $m= \min(k,n-k)$, $P(n,k) \leq 2^{n-2}\left(\!\!\left(k - \frac{m}{2} - \frac{7}{4}\right)^2\! + \frac{9k}{2} + \frac{79}{16}\!\right)$ $+ 2^k\left(\frac{3}{2}\right)^m\!\left(\frac{k}{2} - \frac{m}{6}\right) - 2^k(k+1) - 2^{n-k}\left(\frac{3}{2}\right)^m$.
  \label{thm:upper-bound}
\end{theorem}

\begin{proof}
  Directly from Lemmas \ref{lma:unfold} and \ref{lma:bigineq}.
\end{proof}

  We now present the {\em size difference} $SD(n,k)$ between Pairwise Selection Network and our network.
  Merging step in $pw\_sel^N_K$ costs $2^{k}k - 2^k + 1$ and in $pw\_hbit\_sel^N_K$: $2^{k-1}k$,
  so the difference is given by the following equation:

  \begin{equation}
    SD(n,k) = \left\{ 
    \begin{array}{l l}
      0 & \quad \text{if $n=k$} \\
      0 & \quad \text{if $k=0$} \\
      2^{k-1}k - 2^k + 1 + & \\
      + SD(n-1,k) + SD(n-1,k-1) & \quad \text{if $0<k<n$}
    \end{array} \right.
  \label{eq:pw_size_diff}
  \end{equation}

  \begin{theorem}
    Let $S_{n,k}=\sum_{j=0}^k\binom{n-k+j}{j}2^{k-j}$. Then:

    \[
      SD(n,k) = \binom{n}{k}\frac{n+1}{2} - S_{n,k}\frac{n-2k+1}{2} - 2^k(k-1)-1
    \]
  \end{theorem}

  \begin{proof}
    By straightforward calculation one can verify that $S_{n,0} = 1$, $S_{n,n}
    = 2^{n+1} - 1, S_{n-1,k-1} = \frac{1}{2}(S_{n,k} - \binom{n}{k})$ and
    $S_{n-1,k-1} + S_{n-1,k} = S_{n,k}$. It follows that the theorem is true
    for $k = 0$ and $k = n$. We prove the theorem by induction on pairs
    $(k,n)$. Take any $(k,n)$, $0 < k < n$, and assume that theorem holds for
    every $(k',n') \prec (k,n)$ (in lexicographical order). Then we have:

    \begin{align*}
      SD(n,k)&= 2^{k-1}k - 2^k + 1 + SD(n-1,k) + SD(n-1,k-1) \\
      &= 2^{k-1}k - 2^k + 1 + \binom{n-1}{k}\frac{n}{2} + \binom{n-1}{k-1}\frac{n}{2} - 2^k(k-1)-1 \\
      &\quad -2^{k-1}(k-2)-1 - (S_{n-1,k}\frac{n-2k}{2} + S_{n-1,k-1}\frac{n-2k+2}{2}) \\
      &= \binom{n}{k}\frac{n}{2} - S_{n,k}\frac{n-2k}{2} - S_{n-1,k-1} - 2^k(k-1)-1\\
      &= \binom{n}{k}\frac{n+1}{2}- S_{n,k}\frac{n-2k+1}{2} - 2^k(k-1)-1
    \end{align*}
  \end{proof}
  
\begin{corollary}
  $|pw\_sel^N_{N/2}| - |pw\_hbit\_sel^N_{N/2}| = N\frac{\log N - 4}{2} + \log N + 2$, for $N=2^n$, where $n \in \nat$.
\end{corollary}

\section{Summary}

We have constructed a new family of selection networks, which are based on the
pairwise selection ones, but require less comparators to merge subsequences. The
difference in sizes grows with $k$ and is equal to $n\frac{\log n - 4}{2} + \log n
+ 2$ for $k = n/2$. Less comparators means less variables and clauses generated when
translating cardinality constraints into CNFs, which is the goal for proposing new (smaller) networks.

  The shortcoming of our new networks is that the size of
  an input sequence to the merging procedure is required to be a power of 2, since the construction
  uses modified version of bitonic merging networks \cite{batcher1968sorting} which require $k$ to be the power of 2.
  One could replace $k$ with the nearest power of 2 (by padding the input sequences),
  but it would have a negative impact on the encoding efficiency.

\chapter[Generalized Pairwise Selection Networks]{Generalized Pairwise \\ Selection Networks}\label{ch:mw}

    \def\nqueenssolution{Qd4, Qe2, Qf8, Qa5, Qc1}
    \setchessboard{smallboard,labelleft=false,labelbottom=false,showmover=false,setpieces=\nqueenssolution}

    \begin{tikzpicture}[remember picture,overlay]
      \node[anchor=east,inner sep=0pt] at (current page text area.east|-0,3cm) {\chessboard};
    \end{tikzpicture}

  Here we show a construction of the {\bf $m$-Wise Selection Network} -- the first generalized selection network of this thesis.
  As in the previous chapter, we use the pairwise approach. The idea is to split inputs into $m$ columns, perform a pre-processing on them
  (generalized form of a splitter), recursively select elements in the columns, 
  and then return a top $k$ sorted sequence from the selected elements using a dedicated merging network.
  This network works for any values of $n$ and $k$, in contrast to the networks from the previous chapter.
  We show a complete construction for $m=4$ and perform a theoretical comparison with the Pairwise Selection Network.

  \section{m-Wise Sequences}

  In the Pairwise Sorting Network \cite{parberry1992pairwise}, the first step is to sort pairs and split the input into two equally
  sized sequences $\bar{x}$ and $\bar{y}$, such that for any $i$, $x_i \geq y_i$ (weak domination).
  Then $\bar{x}$ and $\bar{y}$ are sorted independently and finally merged. Even after sorting,
  the property $x_i \geq y_i$ is maintained (see \cite{parberry1992pairwise}). For example,
  the pair of sequences $\tuple{110,100}$ are pairwise sorted. We would like to extend this notion
  to cover larger number of sequences, suppose $m \geq 2$. For example, the tuple
  $\tuple{111,110,100,000}$ is 4-wise (sorted). Furthermore, the sequences we define may not be of equal size.
  It is because in the Pairwise Selection Network we merge two sorted sequences of size
  $\min(n/2, k)$ and $\min(n/2, k/2)$, and in our construction, the i-th sequence (to be merged)
  is of size at most $\floor{k/i}$, where $1 \leq i \leq m$. But it can happen that for some $i$
  we have $\floor{n/m} \leq \floor{k/i}$, i.e., there are at most as many elements in a sequence as the
  number of largest elements to be returned. Therefore we should consider the minimum of the two values.
  The variable $c$ in the following definition serves this purpose.
  
  \begin{definition}[m-wise sequences]\label{def:mws}
  Let $c, k, m \in \nat$, $1 \le k$ and $k/m \le c$. Moreover, let $k_i=\min(c,
  \floor{k/i})$ and $\bar{x}^i \in X^{k_i}$, $1 \leq i \leq m$. The tuple
  $\tuple{\bar{x}^1, \dots, \bar{x}^m}$ is $m$-wise of order $(c,k)$ if:
  \begin{enumerate}
    \item $\forall_{1 \leq i \leq m}$ $\bar{x}^i$ is sorted,
    \item $\forall_{1 \leq i \leq m-1}$ $\forall_{1 \leq j \leq k_{i+1}}$ $x^i_j \geq x^{i+1}_j$.
  \end{enumerate}
  \end{definition}

  \begin{observation} \label{obs:mws}
    Let a tuple $\tuple{\bar{x}^1, \dots, \bar{x}^m}$ be $m$-wise of order 
    $(c,k)$. Then:
    \begin{enumerate}
      \item $|\bar{x}^1| \ge |\bar{x}^2| \ge  \dots \ge |\bar{x}^m| = \floor{k/m}$,
      \item $\sum_{i=1}^{m} |\bar{x}^i| \ge k$,
      \item if $|\bar{x}^{i-1}| > |\bar{x}^i| + 1$ then $|\bar{x}^i| = 
      \floor{k/i}$. 
    \end{enumerate}
  \end{observation}
  
  \begin{proof}
    The first statement is obvious. To prove the second one let $k_i = |\bar{x}^i|$, $1
    \le i \le m$. Then, if $k_i=c$ for each $1 \leq i \leq m$, we have $\sum_{i=1}^{m} k_i
    = m c \geq m k/m = k$ and we are done. Let $1 \leq i \leq m$ be the first index such
    that $k_i\neq c$, therefore $k_i=\floor{k/i} < c$. Thus, for each $1 \leq j < i$: $k_j
    = c > k_i$, therefore $k_j \geq k_i + 1$. From this we get that $\sum_{j=1}^{m} k_j
    \geq \sum_{j=1}^{i} k_j \geq i\floor{k/i} + i - 1 \geq k$.

    The third one can be easily proved by contradiction.
    Assume that $k_i = c$, then from the first property $k_{i-1} \geq k_i = c$. By Definition \ref{def:mws}, $k_{i-1} = \min(c,
    \floor{k/(i-1)}) \leq c$, so $k_{i-1}=c$, a contradiction. 
  \end{proof}

  \begin{definition}[m-wise merger]\label{def:mwm}
    A comparator network $f^s_{(c,k)}$ is an {\em m-wise merger of order $(c,k)$},
    if for each $m$-wise tuple $T = \tuple{\bar{x}^1, \dots, \bar{x}^m}$ of order
    $(c,k)$, such that $s = \sum_{i=1}^{m} |\bar{x}^i|$, $f^s_{(c,k)}(T)$ is top
    $k$ sorted.
  \end{definition}

  \section{m-Wise Selection Network}

  Now we present the algorithm for constructing the $m$-Wise Selection Network
  (Algorithm \ref{net:mw_sel}). In Algorithm \ref{net:mw_sel} we
  use $mw\_merge^{s}_{(c,k)}$, that is, an $m$-Wise Merger of order $(c,k)$, as a black box.
  We give detailed constructions of $m$-Wise Merger for $m=4$ in the next section.

  The idea is as follows: first, we split the input sequence into $m$ columns of non-increasing sizes (lines 2--5)
  and we sort rows using sorters (lines 6--8). Then we recursively run the selection
  algorithm on each column (lines 9--11), where at most $\floor{k/i}$ items are selected
  from the $i$-th column. In obtained outputs,
  selected items are sorted and form prefixes of the columns. The
  prefixes are padded with zeroes (with $\bot$'s) in order to get the input sizes required by the $m$-wise
  property (Definition \ref{def:mws}) and, finally, they are passed to the merging procedure
  (line 12--13).

\begin{examplebox}
\begin{example}\label{ex:mw_sel}
  In Figure \ref{fig:mw_sel_example} we present a sample run of Algorithm \ref{net:mw_sel}. The input is
  a sequence $1111010010000010000101$, with parameters $\tuple{n_1,n_2,n_3,n_4,k}=\tuple{8,7,4,3,6}$.
  The first step (Figure \ref{subfig:mw:a}) is to arrange the input in columns (lines 2--5).
  In this example we get
  $\bar{x}^1=\tuple{11110100}$, $\bar{x}^2=\tuple{1000001}$, $\bar{x}^3=\tuple{0000}$
  and $\bar{x}^4=\tuple{101}$. Next we sort rows using $1,2,3$ or $4$-sorters (lines 6--8),
  the result is visible in Figure \ref{subfig:mw:b}. We make recursive calls in
  lines 9--11 of the algorithm. Items selected recursively in this step are
  marked in Figure \ref{subfig:mw:c}. Notice that in $i$-th column we only need
  to select $\floor{k/i}$ largest elements. This is because of the initial
  sorting of the rows. Next comes the merging step, which is selecting $k$
  largest elements from the results of the previous step. How exactly those elements
  are obtained and output depends on the implementation. We choose the convention that the
  resulting $k$ elements must be placed in the row-major order in our column
  representation of the input (see Figure \ref{subfig:mw:d}).
\end{example}
\end{examplebox}

\begin{figure}
  \centering
  \subfloat[split into columns\label{subfig:mw:a}]{\includegraphics[width=0.22\textwidth]{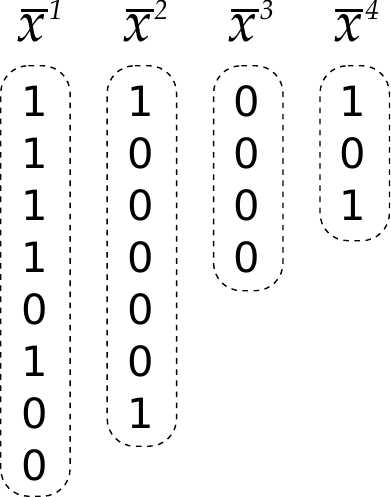}}~~~~~~~~~~~~~~~~~~~~
  \subfloat[sort rows\label{subfig:mw:b}]{\includegraphics[width=0.22\textwidth]{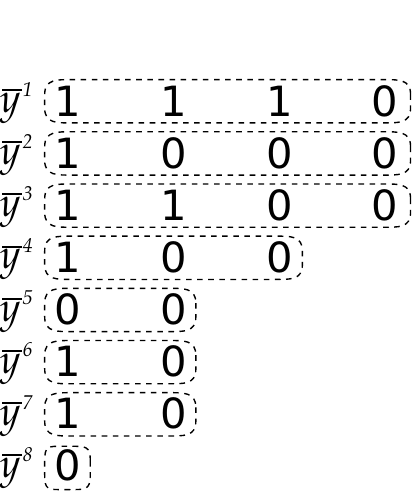}} \\
  \subfloat[select recursively\label{subfig:mw:c}]{\includegraphics[width=0.22\textwidth]{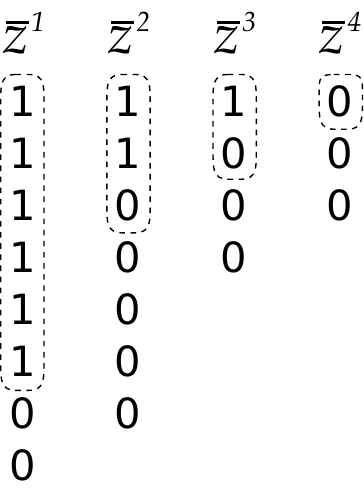}}~~~~~~~~~~~~~~~~~~~~
  \subfloat[merge columns\label{subfig:mw:d}]{\includegraphics[width=0.22\textwidth]{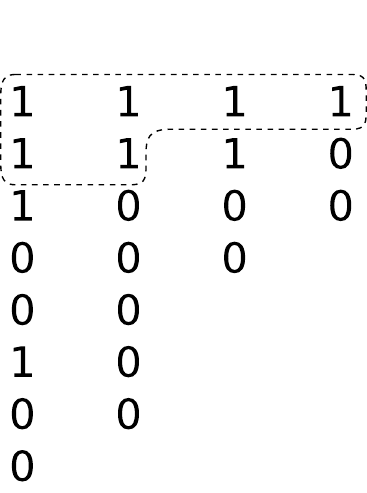}}
  \caption{Sample run of the 4-Wise Selection Network}
  \label{fig:mw_sel_example}
\end{figure}

  \begin{algorithm}[t]
    \caption{$mw\_sel^n_k$}\label{net:mw_sel}
    \begin{algorithmic}[1]
      \Require {$\bar{x} \in X^n$; $n_1, \dots , n_m \in \nat$ 
        where $n > n_1 \geq \dots \geq n_m$ and $\sum n_i = n$; $1 \le k \le n$}
        \Ensure{The output is top $k$ sorted and is a permutation of the inputs}
      \If {$k = 1$}
        \Return $max^n(\bar{x})$
      \EndIf 
      \State $\mathit{offset} = 1$
      \ForAll {$i \in \{1,\dots,m\}$} \Comment{Splitting the input into columns.}
        \State $\bar{x}^i \gets \tuple{x_{\mathit{offset}},\dots,x_{\mathit{offset} + n_i - 1}}$
        \State $\mathit{offset} += n_i$
      \EndFor
      \ForAll {$i \in \{1,\dots,n_1\}$} \Comment{Sorting rows.}
        \State $m' = \max\{j : n_j \geq i\}$
        \State $\bar{y}^i \gets \sort^{m'}(\tuple{x^1_{i},\dots,x^{m'}_{i}})$
      \EndFor
      \ForAll {$i \in \{1,\dots,m\}$} \Comment{Recursively selecting items in columns.}
        \State $k_i=\min(n_1,\floor{k/i})$; ~~$l_i=\min(n_i,\floor{k/i})$
        \Comment{$k_i \ge l_i$}
        \State $\bar{z}^i \gets mw\_sel^{n_i}_{l_i}(\bar{y}^1_i,\dots,\bar{y}^{n_i}_i)$
      \EndFor
      \State $s = \sum_{i=1}^{m} k_i$; ~~$c = k_1$;
        ~~$\overline{out} = \suff(l_1+1,\bar{z}^1) :: \dots :: \suff(l_m+1,\bar{z}^m)$
      \State $\overline{res} \gets mw\_merge^{s}_{(c,k)}(\tuple{
      \pref(l_1,\bar{z}^1)::\bot^{k_1-l_1}, \dots, \pref(l_m,\bar{z}^m)}::\bot^{k_m-l_m})$
      \State \Return $\remove(\bot,\overline{res}) :: \overline{out}$
    \end{algorithmic}
  \end{algorithm}
  
  \begin{lemma}\label{lma:mw_num}
    Let $\bar{t}^i = \tuple{\bar{y}^1_i,\dots,\bar{y}^{n_i}_i}$, where $\bar{y}^j$, $j=1,
    \dots, n_1$, is the result of Step 8 in Algorithm \ref{net:mw_sel}. For each $1 \leq
    i < m$: $|\bar{t}^i|_1 \geq |\bar{t}^{i+1}|_1$.
  \end{lemma}

  \begin{proof}
    Take any $1 \leq i < m$. Consider element $y^j_i$ (for some $1 \leq j \leq n_i$).
    Since $\bar{y}^j$ is sorted, we have
    $y^j_i \geq y^j_{i+1}$, therefore if $y^j_{i+1}=1$ then $y^j_i=1$.
    Thus $|\bar{t}^i|_1 \geq |\bar{t}^{i+1}|_1$. 
  \end{proof}

  \begin{corollary}\label{crly:mw_num}
    For each $1 \leq i \leq m$, let $\bar{z}^i$ be the result of Step 11
    in Algorithm \ref{net:mw_sel}.
    Then for each $1 \leq i < m$: $|\bar{z}^i|_1 \geq |\bar{z}^{i+1}|_1$.
  \end{corollary}

  \begin{proof}
    For $1 \leq i \leq m$, let $\bar{t}^i$ be the same as in Lemma \ref{lma:mw_num}.
    Take any $1 \leq i < m$. Comparator network only permutes its input, therefore
    $|\bar{z}^i|_1 = |mw\_sel^{n_i}_{s_i}(\bar{t}^i)|_1 = |\bar{t}^i|_1 \geq
    |\bar{t}^{i+1}|_1 = |mw\_sel^{n_{i+1}}_{s_{i+1}}(\bar{t}^{i+1})|_1 = |\bar{z}^{i+1}|_1$.
    The inequality comes from Lemma \ref{lma:mw_num}. 
  \end{proof}
  
  \begin{theorem}\label{thm:mw_sel}
    Let $n,k \in \nat$, such that $k \leq n$. Then $mw\_sel^n_k$ is a $k$-selection network.
  \end{theorem}

  \begin{proof}
    We prove by induction that for each $n,k \in \nat$ such that $1 \leq k \leq n$ and
    each $\bar{x} \in \bool^n$: $mw\_sel^n_k(\bar{x})$ is top $k$ sorted. If $1=k \leq n$
    then $mw\_sel^n_k = max^n$, so the theorem is true. For the induction step assume that
    $n \ge k \ge 2$, $m \ge 2$ and for each $(n^{*},k^{*}) \prec (n,k)$ (in
    lexicographical order) the theorem holds. We have to prove the following two properties:

    \begin{enumerate}
      \item The tuple $\tuple{\pref(l_1,\bar{z}^1) :: \bot^{k_1-l_1}, \dots, \pref(l_m,\bar{z}^m) :: 
      \bot^{k_m-l_m}}$ is $m$-wise of order $(k_1, k)$.
      \item The sequence $\bar{w} = \pref(l_1,\bar{z}^1)::\dots ::\pref(l_m,\bar{z}^m)$ contains 
      $k$ largest elements from $\bar{x}$.
    \end{enumerate}
    
    Ad.~1): Observe that for any $i$, $1 \leq i \leq m$, we have $n_i \leq n_1 < n$ and
    $l_i \leq l_1 \leq k$. Thus, $(n_i,l_i) \prec (n,k)$ and $\bar{z}_i$ is top $l_i$
    sorted due to the induction hypothesis. Therefore, $\pref(l_i,\bar{z}^i)$ is sorted and
    so is $\pref(l_i,\bar{z}^i) :: \bot^{k_i-l_i}$. In this way, we prove that
    the first property of Definition \ref{def:mws} is satisfied. To prove the second one,
    fix $1 \leq j \leq l_{i+1}$ and assume that $z^{i+1}_j = 1$. We are going to
    show that $z^i_j = 1$. Since $\bar{z}^i$ and $\bar{z}^{i+1}$ are both top $l_{i+1}$
    sorted (by the induction hypothesis) and since $|\bar{z}^i|_1 \geq |\bar{z}^{i+1}|_1$
    (from Corollary \ref{crly:mw_num}), we have $|\pref(l_{i+1},\bar{z}^i)|_1 \geq
    |\pref(l_{i+1},\bar{z}^{i+1})|_1$, so $z^i_j=1$.
    
    Ad. 2): It is easy to observe that $\bar{z} = \bar{z}^1 :: \dots :: \bar{z}^m$ is a
    permutation of the input sequence $\bar{x}$. If all 1's in $\bar{z}$ are in $\bar{w}$,
    we are done. So assume that there exists $z^i_j=1$ for some $1 \leq i \leq m$, $l_i <
    j \leq n_i$. We show that $|\bar{w}|_1 \geq k$. From the induction hypothesis we
    get $\pref(l_i,z^i) \succeq \tuple{z^i_j}$, which implies that $|\pref(l_i,z^i)|_1 =
    l_i$. From (1) and the second property of Definition \ref{def:mws} it is clear that
    all $z^s_t$ are 1's, where $1 \leq s \leq i$ and $1 \leq t \leq l_i$, therefore
    $|\pref(l_i,\bar{z}^1):: \dots :: \pref(l_i,\bar{z}^i)|_1 = i \cdot l_i$. Moreover,
    since $j > l_i$, we have $l_i=\floor{k/i}$; otherwise we would have $j >n_i$. If
    $i=1$, then $|\bar{w}|_1 \geq 1 \cdot l_1 = k$ and (2) holds.
    
    Otherwise, $l_1 > l_i$, so from the definition of $l_i$, $l_1 \geq \dots \geq l_i$, hence there
    exists $r \geq 1$ such that $\forall_{r < i' \leq i}$ $l_r > l_{i'} = l_i$. Notice
    that since $|\pref(l_i,z^i)|_1=l_i$ and $z^i_j=1$ where $j>l_i$ we get $|z^i|_1 \geq
    l_i + 1$. From Corollary \ref{crly:mw_num} we have that for $1 \leq r' \leq r$:
    $|\bar{z}^{r'}|_1 \geq l_i + 1$. Using $l_{r'} \geq l_r > l_i$ and the induction
    hypothesis we get that each $\bar{z}^{r'}$ is top $l_{i}+1$ sorted, therefore
    $|\pref(l_{i}+1,\bar{z}^{r'})|_1 = l_i+1$. We finally have that $|\bar{w}|_1 \geq
    |\pref(l_{i}+1,\bar{z}^1):: \dots :: \pref(l_{i}+1,\bar{z}^r) :: \pref(l_i,\bar{z}^{r+1})
    :: \dots :: \pref(l_i,\bar{z}^{i})|_1 = r (l_{i}+1) + (i-r)l_i = r (l_{r+1}+1) +
    (i-r)l_{r+1} = i \cdot l_{r+1} + r \geq (r+1)\floor{k/(r+1)} + r \geq k$. In the
    second to last inequality, we use the facts: $l_{r+1}=l_i=\floor{k/i} < n_i \leq n_{r+1}$
    from which $l_{r+1}=\floor{k/(r+1)}$ follows.

    From the statements (1) and (2) we can conclude that $mw\_merge^{s}_{(n,k)}$ returns
    the $k$ largest elements from $\bar{x}$, which completes the proof. 
  \end{proof}

  \section{4-Wise Merging Network}

  In this section the merging algorithm for four columns is presented.
  The input to the merging procedure is $R = \tuple{\pref(n_1,\bar{y}^1), \dots,\pref(n_m,\bar{y}^m)}$,
  where each $\bar{y}^i$ is the output of the recursive call in
  Algorithm \ref{net:mw_sel}. The main observation is the following: since $R$ is $m$-wise, if
  you take each sequence $\pref(n_i,\bar{y}^i)$ and place them side by side, in columns, from
  left to right, then the sequences are sorted in rows and columns.
  The goal of the networks is to put the $k$ largest
  elements in top rows. It is done by sorting slope lines with decreasing slope rate, in lines 1--14
  (similar idea can be found in \cite{zhao1998efficient}).
  The algorithm is presented in Algorithm \ref{net:4mw_merge}. The pseudo-code looks non-trivial, but it
  is because we need a separate sub-case every time we need to use either $\sort^2$,  $\sort^3$ or
  $\sort^4$ operation, and this depends on the sizes of columns and the current slope.

  After slope-sorting phase some elements might not be in the desired row-major order,
  therefore the correction phase is needed, which is the goal of the sorting operations
  in lines 15--18. Figure \ref{fig:4col} shows the order relations of elements
  after $i=\ceil{\log n_1}$ iterations of the {\bf while} loop (disregarding the upper index $i$, for clarity).
  Observe that the order should be possibly corrected between $z_{j-1}$ and $w_{j+1}$ and then
  the 4-tuples $\tuple{x_j, w_j, z_{j-1}, y_{j-1}}$ and $\tuple{x_{j+1}, w_{j+1}, z_j, y_j}$
  should be sorted to get the row-major order. Lines 17--18 addresses certain corner cases of
  the correction phase.

  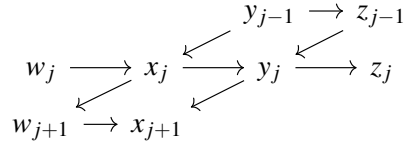
\begin{figure}[ht!]
    \centering 
    \begin{tikzcd}[column sep=small,row sep=tiny]
      ~ & ~ & y_{j-1} \arrow{dl} \arrow{r} & z_{j-1} \arrow{dl}\\
      w_j \arrow{r} & x_j \arrow{dl} \arrow{r} & y_j \arrow{dl} \arrow{r} & z_j \\
      w_{j+1} \arrow{r} & x_{j+1} &  & 
    \end{tikzcd}
    \caption{The order relation among elements of neighboring rows. Arrows shows non-increasing 
      order. Relations that follow from transitivity are not shown.}
    \label{fig:4col}
  \end{figure}

  An input to Algorithm \ref{net:4mw_merge} must be 4-wise of order $(c,k)$
  and the output should be top $k$ sorted. Using the 0-1 principle,
  we can assume that sequences (in particular, inputs) are binary.
  Thus, the algorithm gets as input four sorted 0-1 columns with the additional property
  that the numbers of 1's in successive columns do not increase. Nevertheless, the
  differences between them can be quite big. The goal of each iteration of the main
  loop in Algorithm \ref{net:4mw_merge} is to decrease the maximal possible difference
  by the factor of two. Therefore, after the main loop, the differences are bounded by
  one.
  
\begin{algorithm}[!t]
\caption{$4w\_merge^s_{(c,k)}$}\label{net:4mw_merge}
\begin{algorithmic}[1]
  \Require { 4-wise tuple $\tuple{\bar{w}, \bar{x}, \bar{y}, \bar{z}}$ of order 
  $(c,k)$, where $s = |\bar{w}| + |\bar{x}| + |\bar{y}| + |\bar{z}|$ and $c = |\bar{w}|$.}
  \Ensure{The output is top $k$ sorted and is a permutation of the inputs}
  \State $k_1 = |\bar{w}|$; ~~$k_2 = |\bar{x}|$; ~~$k_3 = |\bar{y}|$; ~~$k_4 = 
  |\bar{z}|$ \Comment{ $k_4=\floor{k/4}$, see def. of $m$-wise tuple.}
  \State $d \gets \min\{l \in \nat \, | \, 2^{l} \geq k_1\}$
  \State $h \gets 2^{d}$; ~~$i \gets 0$; 
  ~~$\tuple{\bar{w}^0,\bar{x}^0,\bar{y}^0,\bar{z}^0} 
  = \tuple{\bar{w}, \bar{x}, \bar{y}, \bar{z}}$; ~~$h_0 \gets h$
  \While{$h_i > 1$} \Comment{ Define the $(i+1)$-th stage of $4w\_merge^s_{(n,k)}$.}
    \State $i \gets i + 1$; ~~$h_{i} \gets h_{i-1}/2$
    \State $\tuple{\bar{w}^i,\bar{x}^i,\bar{y}^i,\bar{z}^i} = 
      \tuple{\bar{w}^{i-1}, \bar{x}^{i-1}, \bar{y}^{i-1}, \bar{z}^{i-1}}$
    \ForAll {$j \in \{1,\dots,\min(k_3-h_i, k_4)\}$}
      \If {$j + 3h_i \leq k_1 \textbf{~and~} j + 2h_i \leq k_2$}
        $\sort^4(z^{i}_j, y^{i}_{j+h_i}, x^{i}_{j + 2h_i}, w^{i}_{j + 3h_i})$
      \ElsIf {$j + 2h_i \leq k_2$}
        $\sort^3(z^{i}_j, y^{i}_{j+h_i}, x^{i}_{j + 2h_i})$
        \Comment{$w^{i}_{j + 3h_i}$ is not defined.}
      \Else
        ~$\sort^2(z^{i}_j, y^{i}_{j+h_i})$
        \Comment{Both $x^{i}_{j + 2h_i}$ and $w^{i}_{j + 3h_i}$ are not defined.}
      \EndIf
    \EndFor
    \ForAll {$j \in \{1,\dots,\min(k_2-h_i, k_3, h_i)\}$}
      \If {$j + 2h_i \leq k_1$}
        $\sort^3(y^{i}_{j}, x^{i}_{j + h_i}, w^{i}_{j + 2h_i})$
      \Else
        ~$\sort^2(y^{i}_{j}, x^{i}_{j + h_i})$
        \Comment{$w^{i}_{j + 2h_i}$ is not defined.}
      \EndIf
    \EndFor
    \ForAll {$j \in \{1,\dots,\min(k_1-h_i, k_2, h_i)\}$}
        $\sort^2(x^{i}_{j}, w^{i}_{j + h_i})$
    \EndFor
  \EndWhile
  \Comment{Define two more stages to correct local disorders.}
  \ForAll {$j \in \{1,\dots,\min(k_1-2, k_4)\}$}
    $\sort^2(z^{i}_{j}, w^{i}_{j+2})$
  \EndFor
  \ForAll {$j \in \{1,\dots,\min(k_2-1, k_4)\}$}
    $\sort^4(y^{i}_{j}, z^{i}_{j}, w^{i}_{j+1}, x^{i}_{j+1})$
  \EndFor
  \If {$k_1 > k_4 \textbf{~and~}  k_2 = k_4$}
    $\sort^3(y^{i}_{k_4}, z^{i}_{k_4}, w^{i}_{k_4+1})$
    \Comment{$x^{i}_{k_4+1}$ is not defined.}
  \EndIf
  \If {$k \bmod 4 = 3 \textbf{~and~} k_1 > k_3$}
    $\sort^2(y^{i}_{k_4+1}, w^{i}_{k_4+2})$
    \Comment{$y^{i}_{k_4+1}$ must be corrected.}
  \EndIf
  \State \Return $\zip(\bar{w}^i,\bar{x}^i,\bar{y}^i,\bar{z}^i)$
  \Comment{Returns the columns in row-major order.}
\end{algorithmic}
\end{algorithm}

  \begin{examplebox}
  \begin{example}
    A sample run of slope sorting phase of Algorithm \ref{net:4mw_merge} is
    presented in Figure \ref{fig:mw_merge_example}. The arrows represent the sorting order.
    Notice how the 1's are being {\em pushed} towards upper-right side. In the end, the differences
    between the number of 1's in consecutive columns are bounded by one.
  \end{example}
  \end{examplebox}

  \begin{observation} \label{obs:abc}
    Let $k_1, \dots, k_4$ be as defined in Algorithm \ref{net:4mw_merge}.
    For each $i$, $0 \le i \le \ceil{\log(k_1)}$, let $c_i = \min(k_3,
    \floor{k/4}+h_i)$, $b_i = \min(k_2, c_i+h_i)$ and  $a_i = \min(k_1, b_i+h_i)$, where
    $h_i$ is as defined in Algorithm \ref{net:4mw_merge}.
    Then we have: $k_1 \ge a_i \ge b_i \ge c_i \ge \floor{k/4}$ and the
    inequalities: $a_i - b_i \le h_i$, $b_i - c_i \le h_i$ and $c_i - \floor{k/4}
    \le h_i$ are true.
  \end{observation}

  \begin{proof}
    From Observation \ref{obs:mws}.(1) we have $k_1 \ge k_2 \ge k_3 \ge
    \floor{k/4}$. It follows that $c_i \ge \floor{k/4}$ and $b_i \ge
    \min(k_3,\floor{k/4}+h_i) = c_i$ and $k_1 \ge a_i \ge \min(k_2,c_i+h_i) = b_i$.
    Moreover, one can see that $c_i - \floor{k/4} = \min(k_3-\floor{k/4},h_i) \le
    h_i$, $b_i - c_i = \min(k_2-c_i, h_i) \le h_i$ and $a_i - b_i = \min(k_1-b_i,
    h_i) \le h_i$, so we are done. 
  \end{proof}

  \begin{figure}[t]
    \centering
    \subfloat[\label{subfig:mwm:a}]{\includegraphics[width=0.22\textwidth]{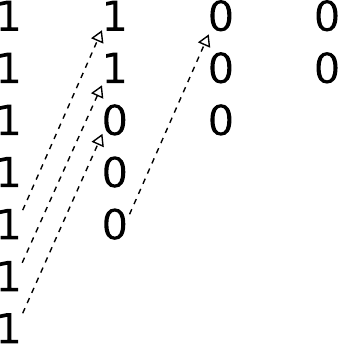}}~~
    \subfloat[\label{subfig:mwm:b}]{\includegraphics[width=0.22\textwidth]{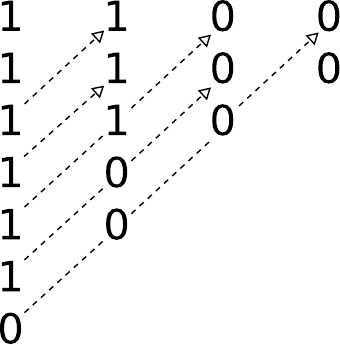}}~~
    \subfloat[\label{subfig:mwm:c}]{\includegraphics[width=0.22\textwidth]{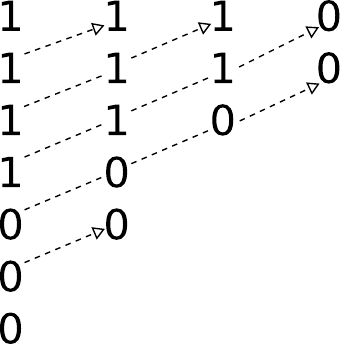}}~~
    \subfloat[\label{subfig:mwm:d}]{\includegraphics[width=0.22\textwidth]{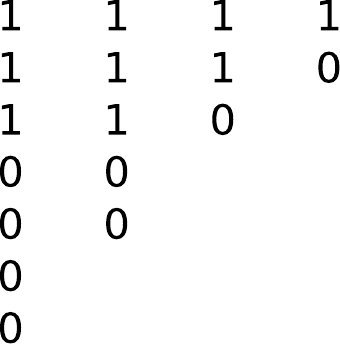}}
    \caption{Sample run of the slope-sorting phase of the 4-Wise Merging Network}
    \label{fig:mw_merge_example}
  \end{figure}
  
  \begin{lemma} \label{lem:check}
    Let $a_i, b_i, c_i$, $0 \le i \le \ceil{\log(k_1)}$, be as defined in Observation
    \ref{obs:abc}. Let us define $\check{w}^i = \pref(a_i,\bar{w}^i)$, $\check{x}^i =
    \pref(b_i,\bar{z}^i)$, $\check{y}^i = \pref(c_i,\bar{y}^i)$ and $\check{z}^i =
    \bar{z}^i$. Then only the items in $\check{w}^i$, $\check{x}^i$, $\check{y}^i$ and
    $\check{z}^i$ take part in the $i$-th iteration of sorting operations in lines 8-14 of
    Algorithm \ref{net:4mw_merge}.
  \end{lemma}

  \begin{proof}
    An element of the vector $\bar{z}^i$, that is,  $z^i_j$ is sorted in one of the lines
    8--10 of Algorithm \ref{net:4mw_merge}, thus for all such $j$ that $1 \le j \le k_4 =
    \floor{k/4}$. An element $y^i_j$ of $\bar{y}^i$ is sorted in lines 8--10 of the first
    inner loop and in lines 12--13 of the second inner loop. In the first three of those
    lines we have the bounds on $j$: $1 + h_i \le j \le \min(k_3,k_4+h_i)$ and in the last
    three lines - the bounds: $1 \le j \le \min(k_3,h_i)$. The sum of these two intervals
    is $1 \le j \le min(k_3, \floor{k/4}+h_i) = c_i$. Similarly, we can analyze the
    operations on $x^i_j$ in lines 8--10, 12--13 and 14. The range of $j$ is the sum of the
    following disjoint three intervals: $[1+2h_i, \min(k_2, k_3+h_i, k_4+2h_i)]$, $[1+h_i,
      \min(k_3+h_i,2h_i)]$ and $[1, \min(k_2, h_i)]$ that give us the interval $[1, \min(k_2,
      k_3+h_1, k_4+2h_i)] = [1, b_i]$ as their sum. The analysis of the range of $j$ used in
    the operations on $w^i_j$ in Algorithm \ref{net:4mw_merge} can be done in the same way. 
  \end{proof}

  \begin{lemma} \label{lma:4wsel}
    Let $\check{w}^i$, $\check{x}^i$, $\check{y}^i$ and $\check{z}^i$ be as defined in Lemma
    \ref{lem:check}. Then for all $i$, $0 \leq i \leq \ceil{\log_2 k_1}$, after the $i$-th
    iteration of the while loop in Algorithm \ref{net:4mw_merge} the sequences $\check{w}^i,
    \check{x}^i, \check{y}^i, \check{z}^i$ are of the form $1^{p_i}0^*$, $1^{q_i}0^*$,
    $1^{r_i}0^*$ and $1^{s_i}0^*$, respectively, and:
    
    \begin{eqnarray}
      p_i \ge q_i \ge r_i \ge s_i, \label{equ:1}\\
      p_i - q_i \le h_i \textrm{~and~} q_i - r_i \le h_i 
      \textrm{~and~} r_i - s_i \le h_i, \label{equ:2}\\
      p_i + q_i + r_i + s_i \ge \min(k, p_0 + q_0 + r_0 + s_0). \label{equ:3}
    \end{eqnarray}
  \end{lemma}

  \begin{proof}
    By induction. At the beginning we have $h_0 \geq k_1 \geq k_2 \geq k3 \geq k4 =
    \floor{k/4}$ and therefore $\check{w}^0 = \bar{w}$, $\check{x}^0 = \bar{x}$,
    $\check{y}^0 = \bar{y}$ and $\check{z}^0 = \bar{z})$. All four sequences $\bar{w},
    \bar{x}, \bar{y}, \bar{z}$ are sorted (by Definition \ref{def:mws}), thus, they are of
    the form $1^{p_0}0^*$, $1^{q_0}0^*$, $1^{r_0}0^*$, and $1^{s_0}0^*$ respectively. By
    Definition \ref{def:mws}.(2), for each pair of sequences: $(\bar{w}_{1..k_2}, \bar{x})$,
    $(\bar{x}_{1..k_3}, \bar{y})$ and $(\bar{y}_{1..k_4}, \bar{z})$, if there is a 1 on the
    $j$ position in the right sequence, it must be a corresponding 1 on the same position in
    the left one. Therefore, we have: $p_0 \geq q_0 \geq r_0 \geq s_0 \geq 0$. Moreover,
    $p_0 - q_0 \leq p_0 \leq k_1 \leq h_0$, $q_0 - r_0 \leq q_0 \leq k_2 \leq h_0$ and $r_0
    - s_0 \leq r_0 \leq k_3 \leq h_0$. Finally, $p_0+q_0+r_0+s_0 \geq \min(k, p_0 + q_0 +
    r_0 + s_0)$, thus the lemma holds for $i=0$.

    In the inductive step $i > 0$ observe that the elements of $\check{w}^i$, $\check{x}^i$,
    $\check{y}^i$ and $\check{z}^i$ are defined by the sort operations over the elements
    with the same indices from vectors  $\check{w}^{i-1}$, $\check{x}^{i-1}$,
    $\check{y}^{i-1}$ and $\check{z}^{i-1}$. This means that the values of
    $w^{i-1}_{a_i+1,\dots,a_{i+1}}$, $x^{i-1}_{b_i+1,\dots,b_{i+1}}$ and
    $y^{i-1}_{c_i+1,\dots,c_{i+1}}$ are not used in the $i$-th iteration. Therefore, the
    numbers of 1's in columns that are sorted in the $i$-th iteration are defined by values:
    $p'_{i-1} = \min(a_i, p_{i-1})$, $q'_{i-1} = \min(b_i, q_{i-1})$, $q'_{i-1} = \min(c_i,
    r_{i-1})$ and $s'_{i-1} = s_{i-1}$. In the following we prove that the numbers with
    primes have the same properties as those without them.
    
    \begin{eqnarray}
      p'_{i-1} \ge q'_{i-1} \ge r'_{i-1} \ge s'_{i-1}, \label{eq:4}\\
      p'_{i-1} - q'_{i-1} \le h_{i-1} \textrm{~and~} q'_{i-1} - r'_{i-1} \le h_{i-1} 
      \textrm{~and~} r'_{i-1} - s'_{i-1} \le h_{i-1}, \label{eq:5}\\
      p'_{i-1} + q'_{i-1} + r'_{i-1} + s'_{i-1} \ge \min(k, p_{i-1} + q_{i-1} + 
      r_{i-1} + s_{i-1}). \label{eq:6}
    \end{eqnarray}
    
    The proofs of inequalities in \ref{eq:4} are quite direct and follow from
    the monotonicity of $\min$. For example, we can observe that $p'_{i-1} =
    \min(a_i,p_{i-1}) \ge \min(b_i, q_{i-1}) = q'_{i-1}$, since we have $a_i \ge
    b_i$ and $p_{i-1} \ge q_{i-1}$, by Observation \ref{obs:abc} and the induction
    hypothesis. The others can be shown in the same way.
    
    Let us now prove one of the inequalities of \eqref{eq:5}, say, the second one.
    We have $q'_{i-1} - r'_{i-1} = \min(b_i,q_{i-1}) - \min(c_i,r_{i-1}) \le
    \min(c_i + h_i,r_{i-1} + h_{i-1}) - \min(c_i,r_{i-1}) \le h_{i-1}$, by
    Observation \ref{obs:abc}, the fact that $h_{i-1} = 2h_i$ and the induction 
    hypothesis. The proofs of the others are similar.
    
    The proof of \eqref{eq:6} is only needed if at least one of the following 
    inequalities are true: $p'_{i-1} < p_{i-1}$, $q'_{i-1} < q_{i-1}$ or $r'_{i-1} < 
    r_{i-1}$. Obviously, the inequalities are equivalent to $a_i < p_{i-1}$, 
    $b_i < q_{i-1}$ and $c_i < r_{i-1}$, respectively. Therefore, to prove
    \ref{eq:6} we consider now three separate cases: (1) $a_i \ge p_{i-1}$ and $b_i 
    \ge q_{i-1}$ and $c_i < r_{i-1}$, (2) $a_i \ge p_{i-1}$ and $b_i < q_{i-1}$ and 
    (3) $a_i < p_{i-1}$.
    
    In the case (1) we have $p'_{i-1} = p_{i-1}$, $q'_{i-1} = q_{i-1}$ and $r'_{i-1}
    = c_i = \min(k_3,h_i + \floor{k/4}) < r_{i-1} \le q_{i-1} \le p_{i-1}$. It
    follows that $p'_{i-1} + q'_{i-1} + r'_{i-1} = p_{i-1} + q_{i-1} + c_i \ge c_i +
    1 + c_i + 1 + c_i = 3c_i + 2$. Since $k_3 \ge r_{i-1} \ge c_i + 1 = \min(k_3,
    \floor{k/4} + h_i) + 1$, we can observe that $c_i$ must be equal to $\floor{k/4}
    + h_i$. In addition, by the induction hypothesis we have $s_{i-1} \ge r_{i-1} -
    h_{i-1} \ge c_i + 1 - 2h_i$. Merging those facts we can conclude that $p'_{i-1}
    + q'_{i-1} + r'_{i-1} + s'_{i-1} \ge 4c_i +3 - 2h_i = 4\floor{k/4} + 4h_i + 3 -
    2h_i \ge k$, so we are done in this case.
    
    In the case (2) we have $p'_{i-1} = p_{i-1}$ and $q'_{i-1} = b_i = \min(k_2, c_i 
    + h_i) < q_{i-1} \le p_{i-1}$. It follows that $p'_{i-1} + q'_{i-1} \ge 2b_i + 
    1$. Since $k_2 \ge q_{i-1} \ge b_i + 1 = \min(k_2, c_i + h_i) + 1$, we can 
    observe that $b_i$ must be equal to $c_i + h_i$. In addition, by the induction 
    hypothesis and Observation \ref{obs:abc}, we can bound $r'_{i-1}$ as $r'_{i-1} = 
    \min(c_i, r_{i-1}) \ge \min(b_i-h_i, q_{i-1} - h_{i-1}) \ge b_i + 1 - 2h_i$. 
    Therefore, $p'_{i-1} + q'_{i-1} + r'_{i-1} \ge 3b_i + 2 - 2h_i = 3c_i + 2 + h_i$. 
    Since $c_i$ is defined as $\min(k_3,\floor{k/4} + h_i)$, we have to consider two 
    sub-cases of the possible value of $c_i$. If $c_i = k_3 \le \floor{k/4} + h_i$, 
    then we have $k_2 \ge q_{i-1} \ge b_i + 1 = c_i + h_i + 1 \ge k_3 + 2$. By 
    Observation \ref{obs:mws}.(3), $k_3$ must be equal to $\floor{k/3}$, thus $3c_i + 
    2 = 3k_3 + 2 \ge k$, and we are done. Otherwise, we have $c_i =\floor{k/4} + h_i$ 
    and since $s'_{i-1} = s_{i-1} \ge r_{i-1} - h_{i-1} \ge b_i + 1 - 4h_i = c_i + 1 
    - 3h_i$, we can conclude that $p'_{i-1} + q'_{i-1} + r'_{i-1} + s'_{i-1} \ge 4c_i 
    + 3 - 2h_i = 4\floor{k/4} + 3 + 2h_i \ge k$.
    
    The last case $a_i < p_{i-1}$ can be proved be the similar arguments. Having
    (\ref{eq:4}.\ref{eq:5},\ref{eq:6}), we can start proving the inequalities from the
    lemma. Observe that in Algorithm \ref{net:4mw_merge} the values of vectors $\bar{w}^i$, $\bar{x}^i$,
    $\bar{y}^i$ and $\bar{z}^i$ are defined with the help of three types of sorters:
    $\sort^4$, $\sort^3$ and $\sort^2$. The smaller sorters are used, when the corresponding
    index is out of the range and an input item is not available. In the following analysis
    we would like to deal only with $\sort^4$ and in the case of smaller sorters we extend
    artificially their inputs and outputs with 1's at the left end and 0's at the right end.
    For example, in line 18 we have $\tuple{y^i_{j}, x^i_{j + h_i}} \gets
    \sort^2(y^{i-1}_{j}, x^{i-1}_{j + h_i})$ so we can analyze this operation as $\tuple{1,
      y^i_{j}, x^i_{j + h_i}, 0} \! \gets \! \sort^4(1, y^{i-1}_{j}, x^{i-1}_{j + h_i}, 0)$. The 0
    input corresponds to the element of $\bar{w}^{i-1}$ with index $j+2h_i$, where $j+2h_i >
    k_1 \ge a_i$, and the 1 input corresponds to the element of $\bar{z}^{i-1}$ with index
    $j-h_i$, where $j-h_i < 0$. A similar situation is in lines 10, 12, 16 and 22, where
    $\sort^2$ and $\sort^3$ are used. Therefore, in the following we assume that elements of
    input sequences $\bar{w}^{i-1}$, $\bar{x}^{i-1}$, $\bar{y}^{i-1}$ and $\bar{z}^{i-1}$
    with negative indices are equal to 1 and elements of the inputs with indices above
    $a_i$, $b_i$, $c_i$ and $\floor{k/4}$, respectively, are equal to 0. This assumption
    does not break the monotonicity of the sequences and we also have the property that
    $w^{i-1}_j = 1$ if and only if $j \le p'_{j-1}$ (and similar ones for $\bar{x}^{i-1}$
    and $q'_{i-1}$, and so on).
    
    It should be clear now that, under the assumption above, we have:
    
    \begin{subequations}
      \begin{align}
        w^i_j = & \min(w^{i-1}_j, x^{i-1}_{j-h_i}, y^{i-1}_{j-2h_i}, z^{i-1}_{j-3h_i}) 
        & \textrm{~for~} 1 \le j \le a_i,\\
        x^i_j = & \snd(w^{i-1}_{j+h_i}, x^{i-1}_{j}, y^{i-1}_{j-h_i}, z^{i-1}_{j-2h_i}) 
        & \textrm{~for~} 1 \le j \le b_i,\\
        y^i_j = & \trd(w^{i-1}_{j+2h_i}, x^{i-1}_{j+h_i}, y^{i-1}_{j}, z^{i-1}_{j-h_i}) 
        & \textrm{~for~} 1 \le j \le c_i, \displaybreak[0]\\
        z^i_j = & \max(w^{i-1}_{j+3h_i}, x^{i-1}_{j+2h_i}, y^{i-1}_{j+h_i}, z^{i-1}_{j})
        & \textrm{~for~} 1 \le j \le \floor{k/4},
      \end{align}
    \end{subequations}
    
    where $\snd$ and $\trd$ denote the second and the third smallest element of
    its input, respectively. Since the functions $\min$, $\snd$, $\trd$ and
    $\max$ are monotone and the input sequences are monotone. we can conclude
    that $w^i_j \ge w^i_{j+1}$, $x^i_j \ge x^i_{j+1}$, $y^i_j \ge y^i_{j+1}$ and
    $z^i_j \ge z^i_{j+1}$. Let $p_i$, $q_i$, $r_i$ and $s_i$ denote the numbers
    of 1's in them. Clearly, we have $p_i + q_i + r_i + s_i = p'_{i-1} +
    q'_{i-1} + r'_{i-1} + s'_{i-1}$, thus $p_i + q_i + r_i + s_i \ge \min(k,
    p_{i-1} + q_{i-1} + r_{i-1} + s_{i-1}) \ge \min(k, p_0 + q_0 + r_0 + s_0)$,
    by the induction hypothesis. Thus. we have proved monotonicity of
    $\bar{w}^i$, $\bar{x}^i$, $\bar{y}^i$ and $\bar{z}^i$ and that
    \eqref{equ:3} holds.
    
    To prove \eqref{equ:2} we show that $p_i - h_i \le q_i$,
    $q_i - h_i \le r_i$, $r_i - h_i \le s_i$, that is, that the following
    equalities are true: $x^i_{p_i-h_i} = 1$, $y^i_{q_i-h_i} = 1$ and
    $z^i_{r_i-h_i} = 1$. The first equality follows from the fact that
    $w^i_{p_i} = 1$ and $w^i_{p_i} \le x^i_{p_i-h_i}$, because they are output
    in this order by a single sort operation. The other two equalities can be shown
    by the similar arguments.
    
    The last equation we have to prove is \eqref{equ:1}. By the induction hypothesis and our 
    assumption we have $w^{i-1}_j \ge x^{i-1}_j \ge y^{i-1}_j \ge z^{i-1}_j$, for any $j$. 
    We know also that the vectors are non-increasing. We use these facts to show that 
    $w^i_{q_i} = 1$, which is equivalent to $p_i \ge q_i$. Since $w^i_{q_i} = 
    \min(w^{i-1}_{q_i}, x^{i-1}_{q_i-h_i}, y^{i-1}_{q_i-2h_i}, z^{i-1}_{q_i-3h_i})$, we 
    have to prove that all the arguments of the $\min$ function are 1's. From the 
    definition of $q_i$ we have $1 = x^i_{q_i} = \snd(w^{i-1}_{q_i+h_i}, x^{i-1}_{q_i}, 
    y^{i-1}_{q_i-h_i}, z^{i-1}_{q_i-2h_i})$, thus the maximum of any pair of arguments of 
    $\snd$ must be 1. Now we can see that $w^{i-1}_{q_i} \ge \max(w^{i-1}_{q_i+h_i}, 
    x^{i-1}_{q_i}) \ge 1$, $x^{i-1}_{q_i-h_i} \ge \max(x^{i-1}_{q_i}, y^{i-1}_{q_i-h_i}) 
    \ge 1$, $y^{i-1}_{q_i-2h_i} \ge \max(y^{i-1}_{q_i-h_i}, z^{i-1}_{q_i-2h_i}) \ge 1$ and 
    $z^{i-1}_{q_i-3h_i} \ge \max(y^{i-1}_{q_i-h_i}, z^{i-1}_{q_i-2h_i}) \ge 1$. In the last 
    inequality we use the fact that for any $j$ it is true that $z^{i-1}_{j-2h_i} \ge 
    y^{i-1}_{j}$ (because $r'_{i-1} - 2h_i \le s'_{i-1}$, by \eqref{eq:5}). Thus, 
    $w^i_{q_i}$ must be 1 and we are done. The other two inequalities can be proved in the 
    similar way with the help of two additional relations: $x^{i-1}_{j-2h_i} \ge 
    w^{i-1}_{j}$ and $y^{i-1}_{j-2h_i} \ge x^{i-1}_{j}$, which follows from the induction 
    hypothesis. 
  \end{proof}

  After $i = \ceil{\log k_1}$ iterations of the main loop in Algorithm \ref{net:4mw_merge} we have $h_i = 1$
  and, by Lemma \ref{lma:4wsel}, elements in vectors (columns) $\check{w}^i$, $\check{x}^i$,
  $\check{y}^i$ and $\check{z}^i$ are in non-increasing order. In the following part of this
  subsection the value of $i$ is fixed and we do not write it as the upper index. By
  \eqref{equ:2} and \eqref{equ:1} of the lemma, we have also the same order in diagonal lines:
  $w_j \le x_{j-1} \le y_{j-2} \le z_{j-3}$ and in rows: $w_j \ge x_j \ge y_j \ge z_j$. From
  \eqref{equ:3} it follows that the vectors contains the $k$ largest elements of the
  input sequences: $\check{w}^0$, $\check{x}^0$, $\check{y}^0$ and $\check{z}^0$. The goal
  of the lines 26--30 in Algorithm \ref{net:4mw_merge} is to correct the order in the vectors
  in such a way that the $k$ largest elements appear at the beginning in the row-major
  order. Figure \ref{fig:4col} shows the mentioned-above order relations. Observe that the
  order should be possible corrected between $z_{j-1}$ and $w_{j+1}$ and then the 4-tuples
  $\tuple{x_j, w_j, z_{j-1}, y_{j-1}}$ and $\tuple{x_{j+1}, w_{j+1}, z_j, y_j}$ should be
  sorted to get the row-major order.
  
  \begin{theorem}\label{thm:4mw_merge}
    The output of Algorithm \ref{net:4mw_merge} is top $k$ sorted.
  \end{theorem}

  \begin{proof}
    The $\zip$ operation outputs its input vectors in the row-major order. By
    \eqref{equ:3} of Lemma \ref{lma:4wsel}, we know that elements in the $out$
    sequence are dominated by the $k$ largest elements in the output vectors of the
    main loop. From the order diagram given in Fig. \ref{fig:4col} it follows that
    $\tuple{y_{i-1}, \max(z_{i-1}, w_{j+1}), w_j, x_j}$ dominates $\tuple{y_i, z_i,
      \min(z_{i-1}, w_{j+1}), x_{i+1}}$, thus, after the sorting operations in lines
    15--16, the values appear in the row-major order. The two special cases,
    where the whole sequence of four elements is not available for the $\sort^4$
    operation, are covered by lines 17 and 18. If the vector $\check{x}^i$ does not
    have the element with index $k_4 + 1$, then just 3 elements are sorted in line
    17. If $k \bmod 4 = 3$ then the element $y_{k_4 + 1}$ is the last one in desired
    order, but the vector $\check{z}^i$ does not have the element with index $k_4 +
    1$, so the network must sort just 2 elements in line 18. Note that in the first
    $\floor{k/4}$ rows we have $4\floor{k/4}$ elements of the top $k$ ones, so
    the values of $k - 4\floor{k/4}$ leftmost elements in the row $k_4+1$ should be
    corrected.
  \end{proof}

\section{Comparison of Pairwise Selection Networks}

The number of comparators in the Pairwise Selection Network (Algorithm \ref{net:pw}) can be defined using this recursive formula:

\begin{equation}
|pw\_sel^n_k| = \left\{ 
    \begin{array}{l l}
      |pw\_sel^{n/2}_k| + |pw\_sel^{n/2}_{k/2}|+ &  \\ 
      + |pw\_split^n| + |pw\_merge^{n}_k| & \quad \text{if $k<n$}\\
      |oe\_sort^k| & \quad \text{if $k=n$} \\
      |max^n| & \quad \text{if $k=1$} \\
    \end{array} \right.
    \label{eq:pw}
\end{equation}

We denote the splitting step as a network $pw\_split^n$. One can check that
it requires $|pw\_split^n|$ $= n/2$ comparators and the merging step requires $|pw\_merge^{n}_k|=k\log k - k + 1$
comparators \cite{zazonpairwise}. In the formula above we assume $n$ and $k$ to be powers of 2.
This way it is always true that $\min(n/2,k)=n/2$, if $k<n$, thus simplifying the calculations.

The schema of this network is presented in Figure \ref{fig:pw_sel}. We want to count
the number of variables and clauses used when merging 4 outputs of the recursive steps,
therefore we expand the recursion by one level (see Figure \ref{fig:pw_sel2}).

 \begin{figure}
  \centering
      \subfloat[one step\label{fig:pw_sel1}]{%
      \includegraphics[width=0.48\textwidth]{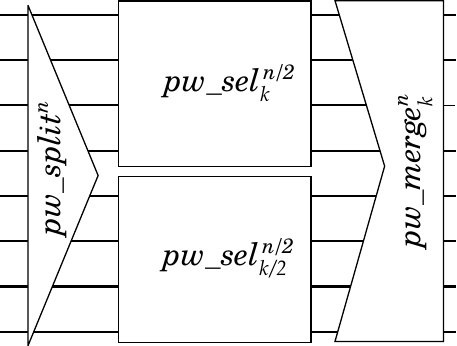}
      }
      ~
      \subfloat[two steps\label{fig:pw_sel2}]{%
      \includegraphics[width=0.48\textwidth]{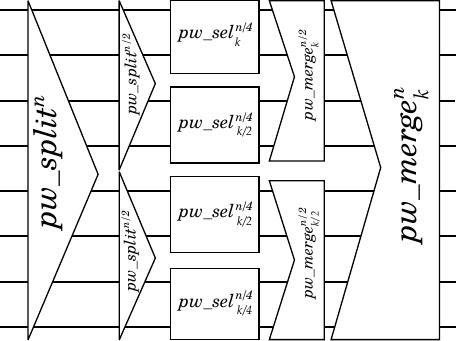}
      }
      \caption{The Pairwise Selection Network}
    \label{fig:pw_sel}
  \end{figure}

 \begin{lemma}
   Let $k,n \in \nat$ be powers of 2, and $k<n$. Then $V(pw\_merge^{n/2}_{k}) + V(pw\_merge^{n/2}_{k/2}) + V(pw\_merge^{n}_k) = 5k\log k -6k +6$
   and $C(pw\_merge^{n/2}_{k}) + C(pw\_merge^{n/2}_{k/2}) + C(pw\_merge^{n}_k) = \frac{15}{2}k\log k - 9k + 9$.
 \end{lemma}

 \begin{proof}
   The number of 2-comparator used is:
   
   \[
     |pw\_merge^{n/2}_{k}| + |pw\_merge^{n/2}_{k/2}| + |pw\_merge^{n}_{k}| = \frac{5}{2}k \log k - 3k + 3
   \]

   Elementary calculation gives the desired result. 
 \end{proof}

 We now count the number of variables and clauses for the 4-Wise Selection Network, again,
 disregarding the recursive steps.

 \begin{lemma}
   Let $k \in \nat$. Then:

   \begin{align*}
     V(4w\_merge^{4k}_k) &= k\log k + \frac{7}{6}k - 5, \\
     C(4w\_merge^{4k}_k) &= \frac{15}{4}k \log k - \frac{33}{24}k - 10.
   \end{align*}
 \end{lemma}

 \begin{proof}
   We separately count the number of 2, 3 and 4-comparators used in the merger (Algorithm \ref{net:4mw_merge}).
   By the assumption that $k \leq n/4$ we get $|\bar{w}|=k$, $|\bar{x}|=k/2$, $|\bar{y}|=k/3$,
   $|\bar{z}|=k/4$ and $h_1=k/2$. We consider iterations 1, 2 and 3 separately and then provide
   the formulas for the number of comparators for iterations 4 and beyond. Results are summarized in
   Table \ref{tbl:4w_comp}.

   In the first iteration $h_1 = k/2$, which means that sets of $j$-values
   in the first two inner loops (lines 7--10 and 11--13) are empty. On the other hand, $1 \leq j \leq \min(k_1 - h_1, k_2, h_1) = k/2$
   (line 14), therefore $k/2$ 2-comparators are used. In fact it is
   true for every iteration $i$ that $\min(k_1 - h_i, k_2, h_i) = h_i = k/2^i$. We note this fact in column 3 of Table \ref{tbl:4w_comp}
   (the first term of each expression). For the next iterations we only need to consider the
   first and second inner loops of the algorithm.

   In the second iteration $h_2 = k/4$, therefore in the first inner loop only the condition in line 10 holds, and only
   when $j \leq k/3 - k/4 = k/12$, hence $k/12$ 2-comparators are used. In the second inner loop the $j$-values satisfy condition
   $1 \leq j \leq \min(k_2 - h_2,k_3,h_2)=k/4$. Therefore the condition in line 12
   holds for each $j \leq k/4$, hence $k/4$ 3-comparators
   are used. In fact it is true for every iteration $i \geq 2$ that $\min(k_2 - h_i,k_3,h_i)=k/2^i$,
   so the condition in line 12 is true for $1 \leq j \leq k/2^i$.
   We note this fact in column 4 of Table \ref{tbl:4w_comp}. For the next iterations we only need to consider the
   first inner loop of the algorithm.

   In the third iteration $h_3 = k/8$, therefore $j$-values in the first inner loop satisfy the condition
   $1 \leq j \leq \min(k_3 - h_3,k_4) = 5k/24$ and
   the condition in line 8 holds for each $j \leq 5k/24$, hence
   $5k/24$ 4-comparators are used.
   
   From the forth iteration $h_i \leq k/16$, therefore for every $1 \leq j \leq \min(k_3 - h_i,k_4) = k/4$ the condition in line 8
   holds. Therefore $k/4$ 4-comparators are used.

   What's left is to sum 2,3 and 4-comparators throughout $\log k$ iterations of the algorithm.
   The results are are presented in Table \ref{tbl:4w_comp}.
   Elementary calculation gives the desired result.
   
   \begin{table}[t] \setlength{\tabcolsep}{4pt}
     \centering
     \bgroup
     \def\arraystretch{1.5}
     \begin{tabular}{ c | c || c | c | c | } \hline
       \multicolumn{1}{|c|}{$i$} & $h_i$ & \#2-comparators & \#3-comparators & \#4-comparators \\ \hline
       \multicolumn{1}{|c|}{$1$} & $\frac{k}{2}$ & $\frac{k}{2}$ & $0$ & $0$ \\ \hline
       \multicolumn{1}{|c|}{$2$} & $\frac{k}{4}$ & $\frac{k}{4} + \frac{k}{12}$ & $\frac{k}{4}$ & $0$ \\ \hline
       \multicolumn{1}{|c|}{$3$} & $\frac{k}{8}$ & $\frac{k}{8}$ & $\frac{k}{8}$ & $\frac{5k}{24}$ \\ \hline
       \multicolumn{1}{|c|}{$\geq4$} & $\frac{k}{2^i}$ & $\frac{k}{2^i}$ & $\frac{k}{2^i}$ & $\frac{k}{4}$ \\ \hline
       & Sum & $\frac{13}{12}k - 1$ & $\frac{k}{2} - 1$ & $\frac{1}{4} k \log k - \frac{13k}{24}$ \\ \cline{2-5} 
     \end{tabular}
     \def\abovecaptionskip{10pt}
     \caption{Number of comparators used in different iterations of Algorithm \ref{net:4mw_merge}}\label{tbl:4w_comp}
     \egroup
   \end{table}
\end{proof}

 Let:

 \begin{align*}
   V_{PSN}&=V(pw\_merge^{n/2}_k) + V(pw\_merge^{n/2}_{k/2}) + V(pw\_merge^{n}_{k}), \\
   C_{PSN}&=C(pw\_merge^{n/2}_k) + C(pw\_merge^{n/2}_{k/2}) + C(pw\_merge^{n}_{k}), \\
   V_{4W}&=V(4w\_merge^{4k}_k), \\
   C_{4W}&=C(4w\_merge^{4k}_k).
 \end{align*}
 
 \begin{figure}[t!]
  \centering
      \subfloat[$N=2^7$\label{fig:mw_prd1}]{%
      \includegraphics[width=0.31\textwidth]{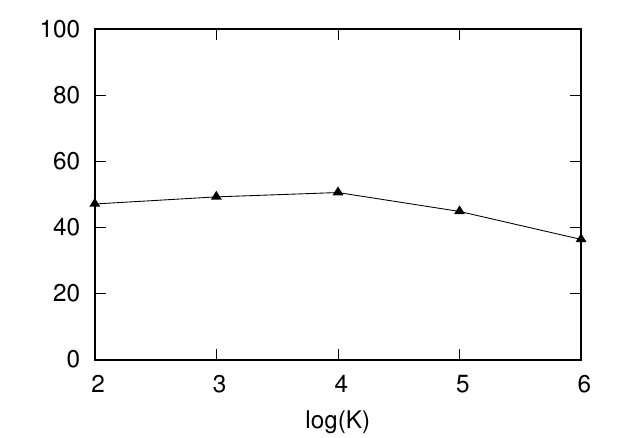}
      }
      ~
      \subfloat[$N=2^{15}$\label{fig:mw_prd2}]{%
      \includegraphics[width=0.31\textwidth]{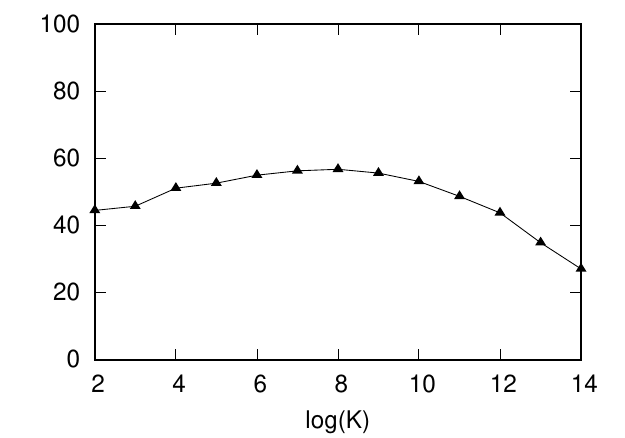}
      }
      ~
      \subfloat[$N=2^{31}$\label{fig:mw_prd3}]{%
      \includegraphics[width=0.31\textwidth]{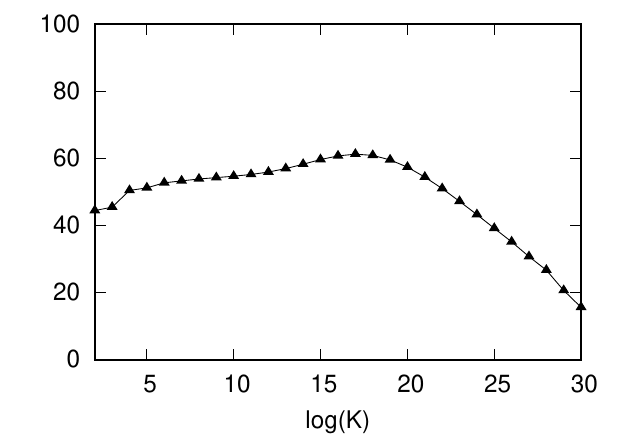}
      }
      \caption{Percentage of variables saved in 4-Wise Selection Networks compared to Pairwise Selection Networks, for selected values of $N$ and $K$.
      Graphs are plotted from the formula $100 \cdot (V(pw\_sel^N_K) - V(4w\_sel^N_K))/V(pw\_sel^N_K)$.}
    \label{fig:mw_prd}
  \end{figure}

 The following corollary shows that 4-column pairwise merging networks
 produces smaller encodings than their 2-column counterpart.
 
 \begin{corollary}\label{crly:pw}
   Let $k \in \nat$ such that $k \geq 4$. Then $V_{4W} < V_{PCN}$ and $C_{4W} < C_{PCN}$.
 \end{corollary}

 For the following theorem, note that $V(4w\_split^n)=V(pw\_split^n)=n$.
 
 \begin{theorem}
    Let $n,k \in \nat$ such that $1\leq k \leq n/4$ and $n$ and $k$ are both powers of 4. Then
    $V(4w\_sel^n_k) \leq V(pw\_sel^n_k)$.
  \end{theorem}

  \begin{proof}
    By induction. For the base case, consider $1=k<n$. It follows that
    $V(4w\_sel^n_k)=V(pw\_sel^n_k)=V(max^n)$. For the induction step assume that for each
    $(n',k') \prec (n,k)$ (in lexicographical order), where $k \geq 4$, the inequality holds, we get:

    \[
    \begin{tabular}{l r}
      \multicolumn{2}{l}{$V(4w\_sel^n_k)= V(4w\_split^n) + \sum_{1 \leq i \leq 4} V(4w\_sel^{n/4}_{k/i}) + V(4w\_merge^{4k}_k)$} \\
      & \small{\bf (by the construction of $4w\_sel$)} \\
      \multicolumn{2}{l}{$\leq V(4w\_split^n) + \sum_{1 \leq i \leq 4} V(pw\_sel^{n/4}_{k/i}) + V(4w\_merge^{4k}_k)$} \\
      & \small{\bf (by the induction hypothesis)} \\
      \multicolumn{2}{l}{$\leq V(pw\_split^n) + 2V(pw\_split^{n/2}) + \sum_{1 \leq i \leq 4} V(pw\_sel^{n/4}_{k/i})$} \\
      \multicolumn{2}{l}{$\quad + V(pw\_merge^{n/2}_k) + V(pw\_merge^{n/2}_{k/2}) + V(pw\_merge^{n}_{k})$} \\
      & \small{\bf (by Corollary \ref{crly:pw} and because $V(4w\_split^n) < V(pw\_split^n) + 2V(pw\_split^{n/2})$)} \\
      \multicolumn{2}{l}{$\leq V(pw\_split^n) + 2V(pw\_split^{n/2}) + V(pw\_sel^{n/4}_{k})+ 2V(pw\_sel^{n/4}_{k/2})$} \\
      \multicolumn{2}{l}{$\quad + V(pw\_sel^{n/4}_{k/4}) + V(pw\_merge^{n/2}_k) + V(pw\_merge^{n/2}_{k/2}) + V(pw\_merge^{n}_{k})$} \\
      & \small{\bf (because $V(pw\_sel^{n/4}_{k/3}) \leq V(pw\_sel^{n/4}_{k/2})$)} \\
      \multicolumn{2}{l}{$= V(pw\_sel^n_k)$} \\
      & \small{\bf (by the construction of $pw\_sel$)} \\
    \end{tabular}
    \] 
  \end{proof}

  In Figure \ref{fig:mw_prd} we show what percentage of variables is saved while using our 4-Wise Selection Networks instead
  of Pairwise Selection Networks. We see that the number of variables saved can be up to 60\%.

  \section{Summary}

  In this chapter we presented a family of multi-column selection networks based
  on the pairwise approach, that can be used to encode cardinality
  constraints. We showed a detailed construction where the number of columns is
  equal to 4 and we showed that the encoding is smaller than its 2-column
  counterpart.

\part{Odd-Even Selection Networks}

\chapter[Generalized Odd-Even Selection Networks]{Generalized Odd-Even \\ Selection Networks}\label{ch:cp18}

    \def\nqueenssolution{Qd4, Qe2, Qf8, Qa5, Qc1, Qg6}
    \setchessboard{smallboard,labelleft=false,labelbottom=false,showmover=false,setpieces=\nqueenssolution}

    \begin{tikzpicture}[remember picture,overlay]
      \node[anchor=east,inner sep=0pt] at (current page text area.east|-0,3cm) {\chessboard};
    \end{tikzpicture}

Even though it has been shown that the pairwise networks use less comparators than the odd-even networks \cite{zazonpairwise}
(for selected values of $n$ and $k$),
it is the latter that achieve better practical results in the context of encoding cardinality constraints \cite{abio2013parametric}.
In this chapter we show a construction of a generalized selection
network based on the odd-even approach called the {\em $4$-Odd-Even Selection Network}. We show that our network is not only more efficient,
but it is also easier to implement (and to prove its correctness) than the GSN based on the pairwise approach from the previous chapter. 

The construction is the generalization of the multi-way merge sorting network by Batcher and Lee \cite{lee1995multiway}.
The main idea is to split the problem into $4$ sub-problems, recursively select $k$ elements in them and
then merge the selected subsequences using an idea of multi-way merging. In such a
construction, we can encode more efficiently comparators in the combine phase of the
merger: instead of encoding each comparator separately by 3 clauses and 2 additional
variables, we propose an encoding scheme that requires 5 clauses and 2 variables on
average for each pair of comparators. 

We give a detailed construction for the $4$-Odd-Even Merging Network. We compare the numbers of variables and
clauses of the encoding and its counterpart: the 2-Odd-Even Merging Network \cite{codish2010pairwise}. 
The calculations show that encodings based on our network use fewer
variables and clauses, when $k<n$.

The construction is parametrized by any values of $k$ and $n$ (just like $m$-Wise Selection Network from the previous chapter),
so it can be further optimized by mixing them with other constructions. For example, in our experiments
we mixed them with the direct encoding for small values of parameters.
We show experimentally that multi-column selection networks are superior to standard selection
networks previously proposed in the literature, in context
of translating cardinality constraints into propositional formulas.

We also empirically compare our encodings with other state-of-the-art encodings, not only based on
comparator networks, but also on binary adders and binary decision diagrams. Those are mainly
used in encodings of Pseudo-Boolean constraints, but it is informative to see how well they
perform when encoding cardinality constraints.

At the end of this chapter we show how we can generalize the $4$-Odd-Even Selection Network
to the {\em $m$-Odd-Even Selection Network}, for any $m \geq 2$, just like we showed 
the {\em $m$-Wise Selection Network} in Chapter \ref{ch:mw}.

\section{4-Odd-Even Selection Network}

  We begin with the top-level algorithm for constructing the $4$-Odd-Even Selection
  Network (Algorithm \ref{net:oe_sel}) where we use $oe\_4merge^{s}_{k}$ as a black box. It
  is a $4$-merger of order $k$, that is, it outputs top $k$ sorted sequence from the inputs consisting of $4$ sorted sequences.
  We give detailed construction of a $4$-merger called {\em 4-Odd-Even Merger} in the next sub-section.

  The idea we use is the generalization of the one used in 2-Odd-Even Selection Network
  from \cite{codish2010pairwise}, which is based on the Odd-Even Sorting Network by Batcher
  \cite{batcher1968sorting}, but we replace the last network with Multiway Merge Sorting Network by
  Batcher and Lee \cite{lee1995multiway}. We arrange the input sequence into $4$ columns of
  non-increasing sizes (lines 3--6) and then recursively run the selection algorithm on each column
  (lines 9--11), where at most top $k$ items are selected from each column. Notice that
  each column is represented by ranges derived from the increasing value of variable
  $\mathit{offset}$. Notice further, that sizes of the columns are selected in such a way that in
  most cases all but first columns are of equal length and the length is a power of two
  (lines 3--5) that is close to the value of $k/4$ (observe that $[k/6, k/3)$ is the smallest
  symmetric interval around $k/4$ that contains a power of 2). Such a choice produces much
  longer propagation paths for small values of $k$ with respect to $n$.
  In the recursive calls selected items are sorted and form prefixes of the
  columns, which are then the input to the merging procedure (line 13). The base case,
  when $k=1$ (line 2), is handled by the selector $sel^n_1$.
  
  \begin{figure}[t!]
    \centering
    \includegraphics[width=1.0\textwidth]{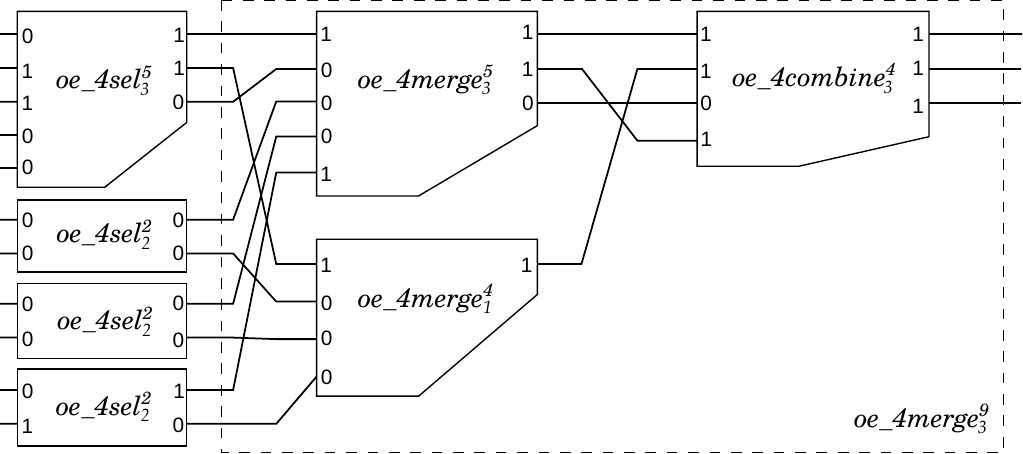}
    \caption{An example of 4-Odd-Even Selection Network, with $n=11$, $k=3$, $n_1=5$, $n_2 =   
    n_3 = n_4 = 2$}
    \label{fig:4oe_sel_cp18}
  \end{figure}

  \begin{examplebox}
  \begin{example}
    In Figure \ref{fig:4oe_sel_cp18} we present a schema of 4-Odd-Even Selection Network, which
    selects 3 largest elements from the input $01100000001$. In this example, $n=11$,
    $k=3$, $n_1=5$, $n_2 = n_3 = n_4 = 2$. First, the input is passed to the recursive
    calls, then the procedure $oe\_4merge^{9}_3$ is applied (Algorithm \ref{net:4oe_merge_cp18}).
  \end{example}
  \end{examplebox}

  \begin{algorithm}[t!]
    \caption{$oe\_4sel^n_k$}\label{net:oe_sel}
    \begin{algorithmic}[1]
      \Require {$\bar{x} \in \bool^n$; $0 \le k \le n$}
      \Ensure{The output is top $k$ sorted and is a permutation of the inputs}
      \If {$k = 0 \textbf{~or~} n \le 1$} \Return $\bar{x}$
        \ElsIf {$k = 1$} \Return $sel^n_1(\bar{x})$
      \EndIf
      \If {$n < 8 \textbf{~or~} k = n$}
          $n_2 = \lfloor (n+2)/4 \rfloor; ~n_3 = \lfloor (n+1)/4 \rfloor; ~n_4 = \lfloor n/4 
          \rfloor; $ \Comment{divide evenly}
        \ElsIf {$2^{\lceil\log(k/6)\rceil} \le \lfloor n/4 \rfloor$} 
          $n_2 = n_3 = n_4 = 2^{\lceil\log(k/6)\rceil}$ \Comment{ divide into powers of 2}
        \Else ~$n_2 = n_3 = n_4 = \lfloor k/4 \rfloor$ 
          \Comment{otherwise, if the power of 2 is too far from k/4}
      \EndIf
      \State $n_1 = n - n_2 - n_3 - n_4$ \Comment{ $n = n_1+\dots+n_4$ and $n_1 \ge n_2 \ge 
      n_3 \ge n_4$ }
      \State $\mathit{offset} = 1$
      \ForAll {$i \in \{1,\dots,4\}$}
        \State $k_i=\min(k, n_i)$
        \State $\bar{y}^i \gets oe\_4sel^{n_i}_{k_i}(\tuple{x_{\mathit{offset}},\dots,x_{\mathit{offset} + n_i - 1}})$
          \Comment{recursive calls}
        \State $\mathit{offset} += n_i$
      \EndFor
      \State $s = \sum_{i=1}^{4} k_i$;
         ~~$\overline{out} = \suff(k_1+1,\bar{y}^1) :: \dots :: \suff(k_4+1,\bar{y}^4)$
      \State \Return $oe\_4merge^{s}_{k}(
        \tuple{\pref(k_1,\bar{y}^1),\dots,\pref(k_4,\bar{y}^4)}) :: \overline{out}$
    \end{algorithmic}
  \end{algorithm}

  \begin{theorem}\label{thm:oe_sel}
    Let $n,k \in \nat$, such that $k \leq n$. Then $oe\_4sel^n_k$ is a $k$-selection network.
  \end{theorem}

  \begin{proof}
    Observe that $\bar{y} = \bar{y}^1 :: \dots :: \bar{y}^4$ is a
    permutation of the input sequence $\bar{x}$. We prove by induction that for each $n,k
    \in \nat$ such that $1 \leq k \leq n$ and each $\bar{x} \in \bool^n$:
    $oe\_4sel^n_k(\bar{x})$ is top $k$ sorted. If $1=k \leq n$ then $oe\_4sel^n_k = sel^n_1$,
    so the theorem is true. For the induction step assume that $n \ge k \ge 2$
    and for each $(n^{*},k^{*}) \prec (n,k)$ (in lexicographical order) the theorem holds.
    We have to prove that the sequence $\bar{w} = \pref(k_1,\bar{y}^1)::\dots
    ::\pref(k_4,\bar{y}^4)$ contains $k$ largest elements from $\bar{x}$. If all 1's from
    $\bar{y}$ are in $\bar{w}$, we are done. So assume that there exists $y^i_j=1$ for
    some $1 \leq i \leq 4$, $k_i < j \leq n_i$. We show that $|\bar{w}|_1 \geq k$.
    Notice that $k_i=k$, otherwise $j > k_i = n_i$ -- a contradiction. Since $|\bar{y}^i|
    = n_i \leq n_1 < n$, from the induction hypothesis we get that $\bar{y}^i$ is top $k_i$
    sorted. In consequence, each element of $\pref(k_i,\bar{y}^i)$ is greater or equal to $y^i_j$,
    which implies that $|\pref(k_i,y^i)|_1 = k_i = k$. We conclude that $|\bar{w}|_1 \geq
    |\pref(k_i,y^i)|_1 = k$. Note also that in the case $n = k$ we have all $k_i = \min(n_i,k) 
    < k$, so the case is correctly reduced.
    
    Finally, using $oe\_4merge^{s}_{k}$ the algorithm returns 
    $k$ largest elements from $\bar{x}$, which completes the proof.
  \end{proof}

  \section{4-Odd-Even Merging Network} \label{sec:merge}

  In this section we give the detailed construction of the network $oe\_4merge$ -- the
  $4$-Odd-Even Merger -- that merges four sequences (columns) obtained from the recursive
  calls in Algorithm \ref{net:oe_sel}. We can assume that input columns are sorted and of
  length at most $k$.
  
    \begin{algorithm}[t!]
      \caption{$oe\_4merge^s_k$}\label{net:4oe_merge_cp18}
      \begin{algorithmic}[1]
        \Require {A tuple of sorted sequences $\tuple{\bar{w}, \bar{x}, \bar{y}, \bar{z}}$, 
          where $1 \le k \le s = |\bar{w}| + |\bar{x}| + |\bar{y}| + |\bar{z}|$
          and $k \ge |\bar{w}| \ge |\bar{x}| \ge |\bar{y}| \ge |\bar{z}|$.}
        \Ensure{The output is top $k$ sorted and is a permutation of the inputs}
        \If {$|\bar{x}| = 0$} \Return $\bar{w}$ \EndIf
        \If {$|\bar{w}| = 1$}
          \Return $sel^{s}_{k}(\bar{w}::\bar{x}::\bar{y}::\bar{z})$ 
            \Comment{Note that $s \le 4$ in this case}
        \EndIf
        \State $s_a = \ceil{|\bar{w}|/2}+\ceil{|\bar{x}|/2}+\ceil{|\bar{y}|/2}+\ceil{|\bar{z}|/2}$;
          ~~$k_a = \min(s_a, \floor{k/2}+2)$;
        \State $s_b = \floor{|\bar{w}|/2} + \floor{|\bar{x}|/2} + \floor{|\bar{y}|/2} + 
          \floor{|\bar{z}|/2}$; ~~$k_b = \min(s_b,\floor{k/2})$
        \State $\bar{a} \gets oe\_4merge^{s_a}_{k_a}(
          \bar{w}_{\odd},\bar{x}_{\odd},\bar{y}_{\odd},\bar{z}_{\odd})$
        \Comment{Recursive calls.}
        \State $\bar{b} \gets oe\_4merge^{s_b}_{k_b}(
          \bar{w}_{\even},\bar{x}_{\even},\bar{y}_{\even},\bar{z}_{\even})$
        \State \Return $oe\_4combine^{k_a+k_b}_k(\pref(k_a,\bar{a}),pref(k_b,\bar{b})) ::
         \suff(k_a+1,\bar{a}) :: \suff(k_b+1,\bar{b})$
      \end{algorithmic}
    \end{algorithm}
  
    \begin{algorithm}[t!]
      \caption{$oe\_4combine^s_k$}\label{net:4oe_combine_cp18}
      \begin{algorithmic}[1]
        \Require {A pair of sorted sequences $\tuple{\bar{x}, \bar{y}}$, 
          where $k \leq s = |\bar{x}| + |\bar{y}|$, $|\bar{y}| \le \floor{k/2}$, 
          $|\bar{x}| \le \floor{k/2} + 2$ and $|\bar{y}|_1 \le |\bar{x}|_1 \le |\bar{y}|_1+4$.}
        \Ensure{The output is sorted and is a permutation of the inputs}
        \State Let $x(i)$ denote $0$ if $i > |\bar{x}|$ or else $x_i$. 
          Let $y(i)$ denote $1$ if $i < 1$ or $0$ if $i > |\bar{y}|$ or $y_i$, otherwise.
        \ForAll {$j \in \{1, \dots, |\bar{x}| + |\bar{y}|\}$}
            \State $i = \ceil{j/2}$
            \If {$j$ is even} $a_j \gets \max(\max(x(i+2), y(i)), \min(x(i+1), y(i-1)))$
                \Else       ~~$a_j \gets \min(\max(x(i+1), y(i-1)), \min(x(i), y(i-2)))$
            \EndIf
        \EndFor
        \State \Return $\bar{a}$
      \end{algorithmic}
    \end{algorithm}

  The network is presented in Algorithm \ref{net:4oe_merge_cp18}. The input to the procedure is
  $\langle\pref(k_1,\bar{y}^1)$, $\ldots$, $\pref(k_4,\bar{y}^4)\rangle$, where each
  $\bar{y}^i$ is the output of the recursive call in Algorithm \ref{net:oe_sel}. The goal is
  to return the $k$ largest (and sorted) elements. It is done by splitting each input sequence
  into two parts, one containing elements of odd index, the other containing elements of
  even index. Odd sequences and even sequences are then recursively merged (lines 5--6)
  into two sequences $\bar{a}$ and $\bar{b}$ that are top $k$ sorted. The sorted prefixes
  are then combined by $oe\_4combine$ into a sorted sequence to which the suffixes
  of $\bar{a}$ and $\bar{b}$ are appended. The result is top $k$ sorted.
  For base cases, since we assume that $|\bar{w}| \ge |\bar{x}| \ge |\bar{y}| \ge |\bar{z}|$,
  we need only to check -- in line 1 -- if $\bar{x}$ is empty (then only $\bar{w}$ is non-empty) or -- in line 2 -- if $\bar{w}$
  contains only a single element -- then the rest of the sequences contains at most one element and we can simply order them with a selector.
  In other cases we have $|\bar{w}|\geq 2$ and $|\bar{x}|\geq 1$, thus $|\bar{w}_{\odd}| < |\bar{w}|$ and $|\bar{w}_{\even}| < |\bar{w}|$,
  so the sizes of sub-problems solved by recursive calls decrease.

  Our network is the generalization of the classic Multiway Merge Sorting Network by
  Batcher and Lee \cite{lee1995multiway}, where we use 4-way mergers and each merger consists of two
  sub-mergers and a combine sub-network. The goal of our network is to select and sort the
  $k$ largest items of four sorted input sequences. The combine networks are described
  and analyzed in \cite{lee1995multiway}. 

  In the case of 4-way merger, the combine operation by Batcher and Lee \cite{lee1995multiway} uses two layers of comparators to
  fix the order of elements of two sorted sequences $\bar{x}$ and $\bar{y}$,
  as presented in Figure \ref{fig:combine}. The combine operation takes sequences $\tuple{x_0,x_1,\dots}$ and $\tuple{y_0,y_1,\dots}$ and
  performs a zip operation: $\tuple{x_0,y_0,x_1,y_1,x_2,y_2,\dots}$. Then, two layers of comparators are applied:
  $[y_i : x_{i+2}]$, for $i=0,1,\dots$, resulting in $\tuple{x_0',y_0',x_1',y_1',\dots}$, and then $[y'_i : x'_{i+1}]$,
  for $i=0,1,\dots$, to get $\tuple{x_0'',y_0'',x_1'',y_1'',\dots}$.

  \begin{figure}[t!]
    \begin{center}
      \includegraphics[scale=1]{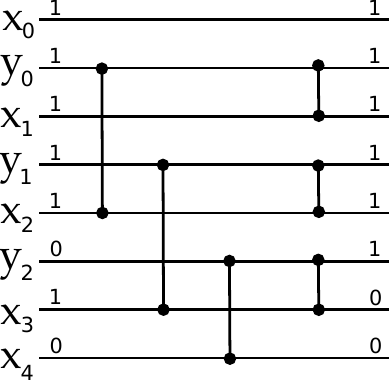}
    \end{center}
    \caption{Comparators of $oe\_4combine^8_6$ and the result of ordering the sequence $11111010$}
    \label{fig:combine}
  \end{figure}

  If we were to directly encode each comparator separately in a combine operation we would need to use
  3 clauses and 2 additional variables on each comparator. The novelty of our construction is that
  the encoding of a combine phase requires 5 clauses and 2 variables on average for each pair of comparators, using
  the following observations:

  \begin{align*}
    y'_i &= \max(y_i,x_{i+2}) \equiv y_i \vee x_{i+2} \quad i=0,1,\dots \\
    x'_i &= \min(y_{i-2},x_i) \equiv y_{i-2} \wedge x_i \quad i=2,3,\dots
  \end{align*}

  \noindent and

  \begin{align*}
    y''_i &= \max(y'_i,x'_{i+1}) = y'_i \vee x'_{i+1} = y_i \vee x_{i+2} \vee (y_{i-1} \wedge x_{i+1}), \\
    x''_i &= \min(y'_{i-1},x'_i) = y'_{i-1} \wedge x'_i = (y_{i-1} \vee x_{i+1}) \wedge y_{i-2} \wedge x_i \\
          &= (y_{i-1} \wedge y_{i-2} \wedge x_i) \vee (y_{i-2} \wedge x_i \wedge x_{i+1}) = (y_{i-1} \wedge x_i) \vee (y_{i-2} \wedge x_{i+1})
  \end{align*}

  \noindent In the above calculations we use the fact that the input sequences are sorted, therefore
  $y_{i-1} \wedge y_{i-2} = y_{i-1}$ and $x_i \wedge x_{i+1} = x_{i+1}$. By the above observations,
  the two calculated values can be encoded using the following set of 5 clauses:

  \[
    y_i \Rightarrow y''_i, x_{i+2} \Rightarrow y''_i, y_{i-1} \wedge x_{i+1} \Rightarrow y''_i, y_{i-1} \wedge x_i \Rightarrow x''_i, y_{i-2} \wedge x_{i+1} \Rightarrow x''_i
  \]

  \noindent if 1's should be propagated from inputs to outputs, otherwise:

  \[
    y''_{i} \Rightarrow y_{i-1} \vee x_{i+2}, y''_{i} \Rightarrow y_{i} \vee x_{i+1}, x''_{i} \Rightarrow x_{i}, x''_{i} \Rightarrow y_{i-2}, x''_{i} \Rightarrow y_{i-1} \vee x_{i+1}.
  \]

  This saves one clause and two variables for each pair of comparators in the original combine operation, which scales to
  $\frac{1}{2}k$ clauses and $k$ variables saved for each two layers of comparators associated with the use
  of a 4-way merger. The pseudo code for our combine procedure is presented in Algorithm \ref{net:4oe_combine_cp18}.

  \begin{examplebox}
  \begin{example}
    In Figure \ref{fig:4oe_sel_cp18}, in dashed lines, a schema of 4-Odd-Even merger is
    presented with $s=9$, $k=3$, $k_1 = 3$ and $k_2 = k_3 = k_4 = 2$. First, the input
    columns are split into two by odd and even indexes, and the recursive calls are made.
    After that, a combine operation fixes the order of elements, to output the 3 largest
    ones. For more detailed example of Algorithm \ref{net:4oe_merge_cp18}, assume that $k=6$ and
    $\bar{w}=100000$, $\bar{x}=111000$, $\bar{y}=100000$, $\bar{z}=100000$. Then $\bar{a}
    = oe\_4merge^{12}_{5}(100,110,100,100)=111110000000$ and $\bar{b} =
    oe\_4merge^{12}_{3}(000,100,000,000)=100000000000$. The combine operation gets
    $\bar{x} = \pref(5, \bar{a}) = 11111$ and $\bar{y} = \pref(3,\bar{b}) = 100$. Notice
    that $|\bar{x}|_1 - |\bar{y}|_1=4$ and after zipping we get $11101011$. Thus, two
    comparators from the first layer are needed to fix the order.
  \end{example}
  \end{examplebox}

  \begin{theorem}\label{thm:4oe_merge}
    The output of Algorithm \ref{net:4oe_merge_cp18} is top $k$ sorted.
  \end{theorem}

  We start with proving a lemma stating that the result of applying network $oe\_4combine$ 
  to any two sequences that satisfy the requirements of the network is sorted and is a 
  permutation of inputs. Then we prove the theorem.
  
  \begin{lemma}\label{lma:4oe_merge_lma}
  Let $k \geq 1$ and $\bar{x}, \bar{y} \in \{0,1\}^*$ be a pair of sorted sequences such
  that $k \leq s = |\bar{x}| + |\bar{y}|$, $|\bar{y}| \le \floor{k/2}$, $|\bar{x}| \le
  \floor{k/2} + 2$ and $|\bar{y}|_1 \le |\bar{x}|_1 \le |\bar{y}|_1+4$. Let $\bar{a}$ be
  the output sequence of $oe\_4combine^s_k(\bar{x}, \bar{y})$. Then for any $j$, $1 \leq j
  < s$ we have $a_j \geq a_{j+1}$. Moreover, $\bar{a}$ is a permutation of $\bar{x} ::
  \bar{y}$.
  \end{lemma}
   
  \begin{proof}
  Note first that the notations $x(i)$ and $y(i)$ (introduced in Algorithm \ref{net:4oe_combine_cp18}) defines
  monotone sequences that extend the given input sequences $\bar{x}$ and $\bar{y}$ (which
  are sorted). Observe that the inequality is obvious for an even $j = 2i$, because
  $a_{2i} = \max($ $\max(x(i+2), y(i)), \min(x(i+1), y(i-1))) \geq \min(\max(x(i+2), y(i)),
  \min(x(i+1), y(i-1))) = a_{2j+1}$. Consider now an odd $j = 2i-1$ for which $a_{2i-1} =
  \min(\max(x(i+1), y(i-1)), \min(x(i), y(i-2)))$. We show that all three
  values: (1) $\max(x(i+1), y(i-1))$, (2) $x(i)$ and (3) $y(i-2)$ are upper bounds on
  $a_{2j}$. Then the minimum of them is also an upper bound on $a_{2j}$.
  
  We have the following inequalities as the consequence of the assumptions: $x(l) \geq
  y(l) \ge x(l+4)$ for any integer $l$. Using them and the monotonicity of $x(i)$, $y(i)$
  and the $\min/\max$ functions, we have:
  \begin{enumerate}
  \item[(1)] $\max(x(i+1), y(i-1)) \geq \max(\max(x(i+2), y(i)), \min(x(i+1), y(i-1))) = 
  a_{2j}$,
  \item[(2)] $x(i) \geq \max(\max(x(i+2), y(i)), \min(x(i+1), y(i-1))) = a_{2j}$ and 
  \item[(3)] $y(i-2) \geq \max(\max(x(i+2), y(i)), \min(x(i+1), y(i-1))) = a_{2j}$. 
  \end{enumerate}
  In (2) we use $x(i) \geq x(i+1) \geq \min(x(i+1), y(i-1))$. In (3) - the similar ones.
  
  To prove the second part of the lemma let us introduce an intermediate sequence $b_j$, $1\leq 
  j \leq s+1$ such that $b_{2i} = \max(x(i+2), y(i))$ and $b_{2i-1} = \min(x(i), y(i-2)$ and 
  observe that it is a permutation of $\bar{x} :: \bar{y} :: 0$, since a pair $b_{2i}$ and 
  $b_{2i+3} = \min(x(i+2), y(i))$ is a permutation of the pair $x(i+2)$ and $y(i)$. Now we can 
  write $a_{2i}$ as $\max(b_{2j}, b_{2j+1})$ and $a_{2i+1}$ as $\min(b_{2j}, b_{2j+1})$, thus 
  the sequence $\bar{a} :: 0$ is a permutation of $\bar{b}$ and we are done.
  \end{proof}
  
  \begin{proof}[Proof of Theorem \ref{thm:4oe_merge}]
  Let $k \geq 1$ and $\bar{w}$, $\bar{x}$, $\bar{y}$ and $\bar{z}$ be sorted binary
  sequences such that $k \le s = |\bar{w}| + |\bar{x}| + |\bar{y}| + |\bar{z}|$ and $k \ge
  |\bar{w}| \ge |\bar{x}| \ge |\bar{y}| \ge |\bar{z}|$. Assume that they are the inputs to the
  network $oe\_4merge^s_k$, so we can use in the following the variables and sequences
  defined in it. The two base cases are: (1) all but first sequences are empty, and (2) all
  sequences contain at most one item. In both of them the network trivially select the top
  $k$ items. In the other cases the construction of $oe\_4merge^s_k$ is recursive, so we
  proceed by induction on $s$. Observe then that $s_a, s_b < s$, since $|\bar{w}| \geq 2$
  and $|\bar{x}| \geq 1$. By induction hypothesis, $\bar{a}$ is top $k_a$ sorted and
  $\bar{b}$ is top $k_b$ and $\bar{a} :: \bar{b}$ is an permutation of the inputs. Let
  $\bar{c} = oe\_4combine_k^{k_a + k_b}(\pref(k_a, \bar{a}), \pref(k_b, \bar{b})$. By
  previous lemma $\bar{c}$ is sorted and is a permutation of $\pref(k_a, \bar{a}) ::
  \pref(k_b, \bar{b})$. Thus the output sequence $\bar{c} :: \suff(k_a+1, \bar{a}) ::
  \suff(k_b+1, \bar{b})$ is a permutation of $\bar{w} ::\bar{x} :: \bar{y} ::\bar{z}$ and
  it remains only to prove that the output is top $k$ sorted.
  
  If $\suff(k_a+1, \bar{a}) :: \suff(k_b+1, \bar{b})$ contains just zeroes, there is nothing to 
  prove. Assume then that it contains at least one 1. In this case we prove 
  that $\pref(k_a, \bar{a}) :: \pref(k_b, \bar{b})$ contains at least $k$ 1's, thus $c_k$ is 
  $1$, and the output is top $k$ sorted. Observe that $k_a + k_b \geq k$, because $s_b 
  \leq s_a \leq s_b+4$ and $s_a + s_b = s \geq k$ so $s_a \geq \ceil{s/2}$ and $s_b \geq 
  \ceil{s/2} - 2$.
  
  It is clear that $|\bar{b}|_1 \leq |\bar{a}|_1 \leq |\bar{b}|_1 + 4$, because in each
  sorted input at odd positions there is the same number of 1's or one more as at even
  positions. Assume first that $|\suff(k_a+1, \bar{a})|_1 > 0$. Then the suffix in
  non-empty, so $k_a = \floor{k/2} + 2$ and $\pref(k_a,\bar{a})$ must contain only 1's,
  thus $|\bar{a}|_1 \geq k_a + 1 \geq \floor{k/2} + 3$. It follows that $|\bar{b}|_1 \geq 
  \floor{k/2} - 1$. If $\pref(k_b, \bar{b})$ contains only 1's then $\bar{c}$ also consists 
  only of 1's. Otherwise the prefix must contain $\floor{k/2} - 1$ 1's and thus the total 
  number of 1's in $\bar{c}$ is at least $\floor{k/2} + 2 + \floor{k/2} - 1 \geq k$. 
  
  Assume next that $|\suff(k_b+1, \bar{b})|_1 > 0$. Then $\pref(k_b,\bar{b})$ must contain only 
  1's and $|\bar{b}|_1 \geq k_b +1 \geq \floor{k/2} + 1$. Since $k_a \geq \ceil{k/2}$ and 
  $|\bar{a}|_1 \geq |\bar{b}|_1 \ge \floor{k/2} + 1$, we get $|\pref(k_a,\bar{a})|_1 \geq 
  \ceil{k/2}$ and finally $|\bar{c}|_1 \geq \floor{k/2} + \ceil{k/2} \geq k$.
  \end{proof}

  \section{Comparison of Odd-Even Selection Networks} \label{sec:comp}

In this section we estimate and compare the number of variables and clauses in encodings
based on our algorithm to some other encoding based
on the odd-even selection. Such encoding -- which we call the 2-Odd-Even Selection Network --
was already analyzed by Codish and Zazon-Ivry \cite{codish2010pairwise}.
We start by counting how many variables and clauses are needed in order to merge 4 sorted sequences returned
by recursive calls of the 4-Odd-Even Selection Network and the 2-Odd-Even Selection Network. Then, based on those values we
prove that the total number of variables and clauses is almost
always smaller when using the 4-column encoding rather than the 2-column encoding.
In the next section we show that the new encoding is not just smaller, but
also have better solving times for many benchmark instances.

To simplify the presentation we assume that $k \leq n/4$ and both $k$ and $n$
are the powers of $4$. We also omit the ceiling and floor
function in the calculations, when it is convenient for us.

\begin{definition}
  Let $n,k \in \nat$. For given (selection) network $f^n_k$ let $V(f^n_k)$ and $C(f^n_k)$ denote the number
  of variables and clauses used in the standard CNF encoding of $f^n_k$.
\end{definition}

We remind the reader that a single 2-sorter uses 2 auxiliary variables and 3 clauses.
In case of a 4-sorter the numbers are 4 and 15 (by Definition \ref{def:msel}).

  We count how many variables and clauses are needed in order to merge 4 sorted sequences
  returned by recursive calls of the 2-Odd-Even Selection Network and the 4-Odd-Even Selection Network, respectively.
  Two-column selection network using the odd-even approach is presented in \cite{codish2010pairwise}.
  We briefly introduce this network with the following three-step recursive procedure (omitting the base case):

  \begin{enumerate}
    \item Split the input $\bar{x} \in \bool^n$ into two sequences $\bar{x}^1 = \bar{x}_{\odd}$
      and $\bar{x}^2 = \bar{x}_{\even}$.
    \item Recursively select top $k$ sorted elements from $\bar{x}^1$ and top $k$ sorted
      elements from $\bar{x}^2$ (into $\bar{y}^1$ and $\bar{y}^2$, respectively).
    \item Merge the outputs of the previous step using an 2-Odd-Even Merging Network and
      output the top $k$ from $2k$ elements (top $k$ elements from $\bar{y}^1$ and top $k$ elements from $\bar{y}^2$).
  \end{enumerate}

  If we treat the merging step as a network $oe\_2merge^{2k}_k$, then the number of 2-sorters used in the
  2-Odd-Even Selection Network can be written as:

  \begin{equation}
    |oe\_2sel^n_k| = \left\{ 
    \begin{array}{l l}
      2|oe\_2sel^{n/2}_k| + |oe\_2merge^{2k}_k| & \quad \text{if $k<n$}\\
      |oe\_sort^k| & \quad \text{if $k=n$} \\
      |max^n| & \quad \text{if $k=1$} \\
    \end{array} \right.
    \label{eq:oe}
  \end{equation}
  
  One can check that Step 3 requires $|oe\_2merge^{2k}_k|=k\log k + 1$ 2-sorters (see \cite{codish2010pairwise}),
  which leads to the simple lemma.

  \begin{lemma}\label{lma:2oe_vars_cls}
   $V(oe\_2merge^{2k}_k) = 2 k \log k + 2$, $C(oe\_2merge^{2k}_k) = 3 k \log k + 3$.
  \end{lemma}

  The schema of this network is presented in Figure \ref{fig:oe_sel}.
  In order to count the number of comparators used when merging 4 sorted sequences we need
  to expand the recursive step by one level (see Figure \ref{fig:oe_sel2}).

  \begin{figure}[t!]
    \centering
    \subfloat[one step\label{fig:oe_sel1}]{%
      \includegraphics[width=0.48\textwidth]{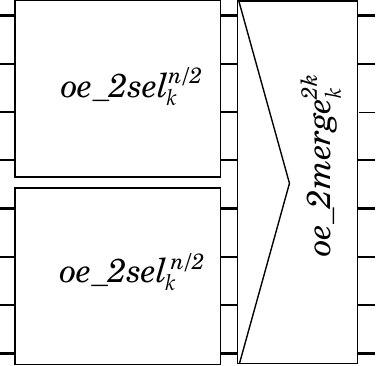}
    }
    ~
    \subfloat[two steps\label{fig:oe_sel2}]{%
      \includegraphics[width=0.48\textwidth]{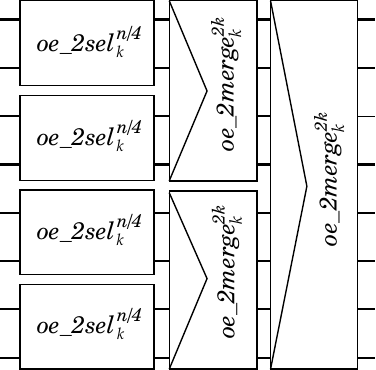}
    }
    \caption{The 2-Odd-Even Selection Network}
    \label{fig:oe_sel}
  \end{figure}

  Now we do the counting for our 4-way merging network based on
  Algorithm \ref{net:4oe_merge_cp18}.
  
  \begin{lemma}\label{lma:4oe_vars_cls}
    Let $k \in \nat$, then: $V(oe\_4merge^{4k}_k) \leq (k - 2)\log k + 5k - 1$;
    $C(oe\_4merge^{4k}_k)$ $\leq$ $(\frac{5}{2}k - 5)\log k + 21k - 6$.
  \end{lemma}

  \begin{proof}
    We separately count the number of variables and clauses used. 

    In the base case (line 2) we can assume -- for the sake of the upper bound -- 
    that we always use 4-sorters. Notice, that the number of 4-sorters is only dependent 
    on the variable $s$. The solution to the following recurrence
    gives the sought number: $\{A(4) = 1; A(s) = 2A(s/2), \text{for } s>4\}$, which is equal to $s/4$. Therefore
    we use $s$ auxiliary variables and $(15/4)s$ clauses. We treat the recursive case separately below.

    The number of variables used in the combine network is at most $k-1$, because a new variable is not needed for $a_i$, where $i>k$,
    because such $a_i$ can be replaced by a zero in clauses containing it, and not for $a_1=x_1$. Therefore, 
    the total number of variables is bounded by solution to the following recurrence:

    \[
      B(s,k) =
        \begin{cases}
          0                            & \quad \text{if } s \leq 4 \\
          B(s_a,k_a)+B(s_b,k_b)+k-1    & \quad \text{otherwise }   \\
        \end{cases}
    \]

    \noindent where $k \leq s = s_a + s_b \leq 4k$, $s_b \leq s_a \leq s_b + 4$ and $k_a = \min(s_a,\floor{k/2}+2)$ and
    $k_b$ $=$ $\min(s_b,\floor{k/2})$. Therefore $s/2 \leq s_a \leq s/2 + 2$, $s/2-2 \leq s_b \leq s/2$, $k_a \leq k/2+2$
    and $k_b \leq k/2$. We claim that $B(s,k) \leq (k-2)(\log s - 2) + \frac{1}{4}s - 1$.
    This can be easily verified by induction.

    The upper bound of the number of clauses can now be easily computed noticing that in the combine we require either
    2 or 3 clauses for each new variable (depending on the parity of the index), therefore the number of clauses in the combiner is
    bounded by $2.5 \times \#vars + 3.5$. Constant factor 3 is added because additional clauses can be added for
    values $a_{k+1}$ and $a_{k+2}$ (see equations in Section \ref{sec:merge}).
    The overall number of clauses in the merger (omitting base cases) is then at most $2.5 \cdot B(s,k) + 3.5(k-1)$, where
    factor $(k-1)$ is the upper bound on the number of combines used in the recursive tree of the merger.
    Elementary calculations give the desired result.
  \end{proof}

  Combining Lemmas \ref{lma:2oe_vars_cls} and \ref{lma:4oe_vars_cls} gives the following corollary.

  \begin{corollary}\label{crly:oe}
    Let $k \in \nat$. Then $3V(oe\_2merge^{2k}_k) - V(oe\_4merge^{4k}_k) \geq (5k+2)\log(\frac{k}{2}) + 9 \geq 0$,
    and for $k \geq 8$, $3C(oe\_2merge^{2k}_k) - C(oe\_4merge^{4k}_k) \geq (\frac{13}{2}k + 5)\log(\frac{k}{8})-\frac{3}{2}k+30 \geq 0$.
  \end{corollary}

  This shows that using our merging procedure gives a smaller encoding than its 2-column counterpart and the differences
  in the number variables and clauses used is significant.

  The main result of this section is as follows.

  \begin{theorem}
    Let $n,k \in \nat$ such that $1\leq k \leq n/4$ and $n$ and $k$ are both powers of 4. Then:

    \[
      dsV_k(n) \stackrel{df}{=} V(oe\_2sel^n_k) - V(oe\_4sel^n_k) \geq \frac{(n-k)(5k+2)}{3k}\log\left(\frac{k}{2}\right) + 3\left(\frac{n}{k}-1\right).
    \]
  \end{theorem}

  \begin{proof}
    Let $dV_k = 3V(oe\_2merge^{2k}_k) - V(oe\_4merge^{4k}_k)$ (from Corollary \ref{crly:oe}), then:

    \begin{align*}
      dsV_k(n) &= V(oe\_2sel^n_k) - V(oe\_4sel^n_k) \\
               &= 2V(oe\_2sel^{n/2}_k) + V(oe\_2merge^{2k}_k) - 4V(oe\_4sel^{n/4}_k) - V(oe\_4merge^{4k}_k) \\
               &= 4V(oe\_2sel^{n/4}_k) + 3V(oe\_2merge^{2k}_k) - 4V(oe\_4sel^{n/4}_k) - V(oe\_4merge^{4k}_k) \\
               &= 4dsV_k(n/4) + dV_k
    \end{align*}

    \noindent The solution to the above recurrence is $dsV_k(n) \geq \frac{1}{3}(\frac{n}{k}-1)dV_k$. Therefore:

    \begin{align*}
      dsV_k(n) &\geq \frac{1}{3}\left(\frac{n}{k}-1\right)\left((5k+2)\log\left(\frac{k}{2}\right) + 9\right) \\
               &=\frac{(n-k)(5k+2)}{3k}\log\left(\frac{k}{2}\right) + 3\left(\frac{n}{k}-1\right).
    \end{align*}
  \end{proof}
  
  \noindent Similar theorem can be proved for the number of clauses (when $k\geq8$).

\section{Experimental Evaluation} \label{sec:exp}

As it was observed in \cite{abio2013parametric}, having a smaller encoding in terms of number of
variables or clauses is not always beneficial in practice, as it should also be
accompanied with a reduction of SAT-solver runtime. In this section we assess how our
encoding based on the new family of selection networks affect the performance of a
SAT-solver.

\subsection{Methodology}

Our algorithms that encode CNF instances with cardinality constraints into CNFs were
implemented as an extension of \textsc{MiniCard} ver. 1.1, created by Mark Liffiton and
Jordyn Maglalang\footnote{See https://github.com/liffiton/minicard}. \textsc{MiniCard}
uses three types of solvers:

\begin{itemize}
  \item {\em minicard} - the core \textsc{MiniCard} solver with native AtMost constraints,
  \item {\em minicard\_encodings} - a cardinality solver using CNF encodings for AtMost constraints,
  \item {\em minicard\_simp\_encodings} - the above solver with simplification / pre-processing.
\end{itemize}

The main program in {\em minicard\_encodings} has an option to generate a CNF formula, given a CNFP
instance (CNF with the set of cardinality constraints)
and to select a type of encoding applied to cardinality constraints.
Program run with this option outputs a CNF
instance that consists of collection of the original clauses with
the conjunction of CNFs generated by given method
for each cardinality constraint. No additional pre-processing and/or simplifications are made.
Authors of {\em minicard\_encodings} have implemented six methods to encode cardinality constraints
and arranged them in one library called {\em Encodings.h}. 
Our modification of \textsc{MiniCard} is that we added implementation 
of the encoding presented in this chapter and put it in
the library {\em Encodings\_MW.h}. Then, for each CNFP instance and each encoding method,
we used \textsc{MiniCard} to generate CNF instances.
After preparing, the benchmarks were run on a different SAT-solver.
Our extension of \textsc{MiniCard}, which we call \textsc{KP-MiniCard},
is available online\footnote{See https://github.com/karpiu/kp-minicard}.

In our evaluation we use the state-of-the-art SAT-solver \textsc{COMiniSatPS}
by Chanseok Oh\footnote{See http://cs.nyu.edu/\%7echanseok/cominisatps/} \cite{oh2016improving},
which have collectively won six medals in SAT Competition 2014 and Configurable SAT Solver Challenge 2014.
Moreover, the modification of this solver called \textsc{MapleCOMSPS} won the Main Track category of
SAT Competition 2016\footnote{See http://baldur.iti.kit.edu/sat-competition-2016/}.
All experiments were carried out on the machines with Intel(R) Core(TM) i7-2600 CPU @ 3.40GHz.

Detailed results are available online\footnote{See http://www.ii.uni.wroc.pl/\%7ekarp/sat/2018.html}.
We publish spreadsheets showing running time for each instance,
speed-up/slow-down tables for our encodings, number of time-outs met and total running time.

\subsection{Encodings}

We use our multi-column selection network for evaluation -- the 4-Odd-Even Selection Network ({\bf 4OE})
based on Algorithms \ref{net:oe_sel}, \ref{net:4oe_combine_cp18} and \ref{net:4oe_merge_cp18}.
We compare our encoding to some others found in the literature.
We consider the Pairwise Cardinality Networks \cite{codish2010pairwise}. We also consider a solver called
\textsc{MiniSat+}\footnote{See https://github.com/niklasso/minisatp}
which implements techniques to encode Pseudo-Boolean constraints to propositional formulas \cite{minisatp}.
Since cardinality constraints are a subclass of Pseudo-Boolean constraints, we can
measure how well the encodings used in \textsc{MiniSat+} perform, compared with our methods.
The solver chooses between three techniques to generate SAT encodings for Pseudo-Boolean constraints.
These convert the constraint to: a BDD structure, a network of binary adders, a network of sorters.
The network of adders is the most concise encoding, but it can have poor propagation properties and often
leads to longer computations than the BDD based encoding. The network of sorters is the implementation
of classic odd-even (2-column) sorting network by Batcher \cite{batcher1968sorting}. Calling the solver we can
choose the encoding with one of the parameters: {\em -ca}, {\em -cb}, {\em -cs}.
By default, \textsc{MiniSat+} uses the so called {\bf Mixed} strategy, where program
chooses which method (adders, BDDs or sorters) to use in the encodings.
We do not include the {\bf Mixed} strategy in the results,
as the evaluation showed that it performs almost the same as {\em -cb} option.
The generated CNFs were written to files with the option {\em -cnf=$<$file$>$}.
Solver \textsc{MiniSat+} have been slightly modified, namely, we fixed a pair of bugs such as
the one reported in the experiments section of \cite{aavani2013new}.

To sum up, here are the competitors' encodings used in this evaluation:

\begin{itemize}
  \item {\bf PCN} - the Pairwise Cardinality Networks (our implementation),
  \item {\bf CA} - encodings based on Binary Adders (from \textsc{MiniSat+}),
  \item {\bf CB} - encodings based on Binary Decision Diagrams (from \textsc{MiniSat+}),
  \item {\bf CS} - the 2-Odd-Even Sorting Networks (from \textsc{MiniSat+}).
\end{itemize}

Encodings {\bf 4OE} and {\bf PCN} were extended,
following the idea presented in \cite{abio2013parametric}, where authors use
Direct Cardinality Networks in their encodings for sufficiently
small values of $n$ and $k$. Values of $n$ and $k$ for which we substitute the recursive calls
with Direct Cardinality Network were selected based on the optimization idea in \cite{abio2013parametric}.
We minimize the function $\lambda \cdot V + C$, where $V$ is the number of variables and $C$
the number of clauses to determine when to switch to direct networks,
and following authors' experimental findings, we set $\lambda=5$. 

Additionally, we compare our encodings with two state-of-the-art general purpose constraint solvers. First is
the \textsc{PBLib} ver. 1.2.1, by Tobias Philipp and Peter Steinke \cite{philipp2015pblib}. This solver implements a plethora of encodings
for three types of constraints: at-most-one, at-most-k (cardinality constraints) and Pseudo-Boolean constraints.
The \textsc{PBLib} automatically normalizes the input constraints and decides which encoder provides the most effective translation.
One of the implemented encodings for at-most-k constraints is based on the sorting network from the paper by Ab\'io et al. \cite{abio2013parametric}.
One part of the \textsc{PBLib} library is the program called {\em PBEncoder} which takes an input
file and translate it into CNF using the \textsc{PBLib}. We have generated CNF formulas from all benchmark instances using this program,
then we have run \textsc{COMiniSatPS} on those CNFs. Results for this method are labeled {\bf PBE} in our evaluation.

The second solver is the \textsc{npSolver} by Norbert Manthey
and Peter Steinke\footnote{See http://tools.computational-logic.org/content/npSolver.php},
which is a Pseudo-Boolean solver that translates Pseudo-Boolean constraints to SAT similar to \textsc{MiniSat+},
but which incorporates novel techniques. We have exchanged the SAT-solver used by default in \textsc{npSolver}
to \textsc{COMiniSatPS} because the results were better with this one. Results for this method are labeled {\bf NPS} in our evaluation.

\begin{figure}[t!]
  \centering
  \includegraphics[width=1.0\textwidth]{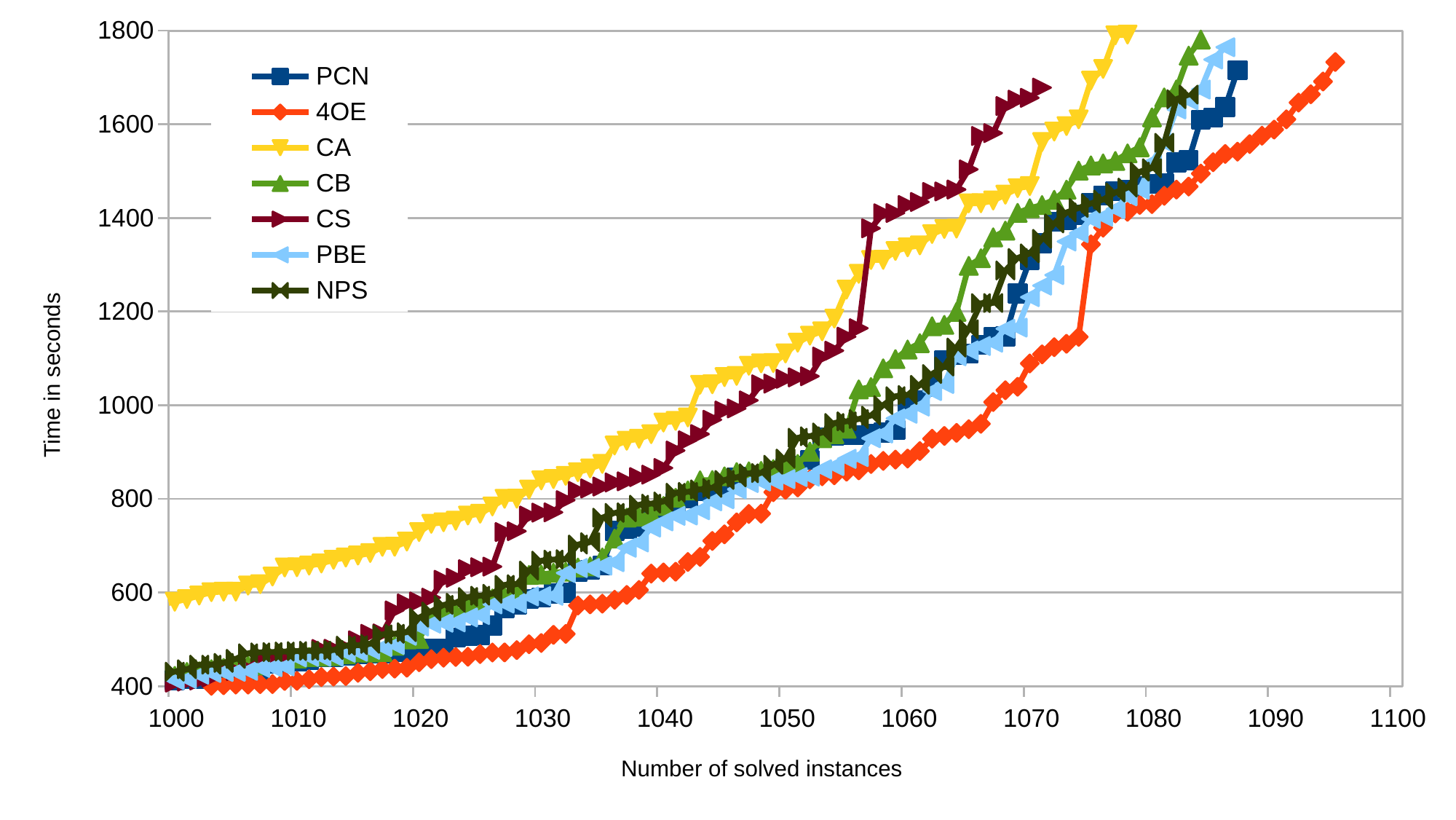}
  \caption{The number of solved instances of PB15 suite in a given time}
  \label{fig:plots}
\end{figure}

\subsection{Benchmarks}

The set of benchmarks we used is {\bf PB15 suite}, which is a set of instances from the
Pseudo-Boolean Evaluation 2015\footnote{See {\texttt http://pbeva.computational-logic.org/}}.
One of the categories of the competition was {\em DEC-LIN-32-CARD}, which contains 2289 instances
-- we use these in our evaluation. Every instance is a collection of cardinality constraints.

\renewcommand{\arraystretch}{1.0}

\subsection{Results}

The time-out limit in the SAT-solver was set to 1800 seconds. When comparing two encodings we only considered instances
for which at least one achieved the SAT-solver runtime of at least 10\% of the time-out limit. All other instances
were considered trivial, and therefore were not included in the speed-up/slow-down results. We also filtered out instances for which
relative percentage deviation of the running time of encoding {\bf A} w.r.t. the running time of encoding {\bf B}
was less than 10\% (and vice-versa).

In Figure \ref{fig:plots} we present a cactus plot, where x-axis gives the number of solved instances of PB15 suite
and the y-axis the time needed to solve them (in seconds) using given encoding.
From the plot we can see that the {\bf 4OE} encoding outperforms all other encodings.

Table \ref{tbl:res-pb} presents speed-up and slow-down factors for
encoding {\bf 4OE} w.r.t. all other encodings.
From the evaluation we can conclude that the best performing encoding is {\bf 4OE}.
From the data presented in Table \ref{tbl:res-pb} our encoding
achieve better speed-up factor w.r.t. all other encodings. Total running time for {\bf 4OE} is 629.78 hours on all 2289 instances. 
All other encodings required more time to finish the computation. Also, {\bf 4OE} solved the most number of instances -- 1095.
The second to last column of Table \ref{tbl:res-pb} shows the difference in total running time of all encodings w.r.t. {\bf 4OE} 
(in HH:MM format -- hours and minutes). The last column indicates the difference in the number of solved instances of all encodings w.r.t. {\bf 4OE} 
(here all instances are counted, even the trivial ones).
We can see, for example, that for {\bf 4OE} computations finished about 7 hours sooner for {\bf 4OE} than {\bf CS}.
This shows that using 4-column selection networks
is more desirable than using 2-column selection/sorting networks for encoding cardinality constraints.
Encodings {\bf CA} and {\bf CS} had the worst performance on PB15 suite. We can also see that even the
state-of-the-art constraint solvers have larger running times and solved less instances on this set of benchmarks,
as {\bf PBE} and {\bf NPS} finished computations more than about 3--4 hours
later than {\bf 4OE}.

\begin{table}[!t]\setlength{\tabcolsep}{4pt}
    \centering
    \begin{tabular}{ c | c | c | c | c | c | c || c | c | c | c | c | c || c | c | } \cline{2-13}
         & \multicolumn{6}{|c||}{4OE speed-up} & \multicolumn{6}{c||}{4OE slow-down} & \multicolumn{2}{c}{} \\ \cline{2-15}
                                  & TO & 4.0 & 2.0 & 1.5 & 1.1 & Total & TO & 4.0 & 2.0 & 1.5 & 1.1 & Total & Time dif. & \#s dif. \\ \hline
      \multicolumn{1}{|l|}{PCN}   & 9  & 11  & 5   & 5   & 5   & 35    & 1  & 2   & 1   & 3   & 3   & 10    & +02:55 & -8  \\ \hline %
      \multicolumn{1}{|l|}{CA}    & 22 & 15  & 28  & 21  & 21  & 107   & 5  & 3   & 5   & 4   & 11  & 28    & +10:54 & -17 \\ \hline %
      \multicolumn{1}{|l|}{CB}    & 18 & 11  & 15  & 8   & 24  & 76    & 7  & 2   & 6   & 3   & 28  & 46    & +04:54 & -11 \\ \hline %
      \multicolumn{1}{|l|}{CS}    & 27 & 13  & 14  & 13  & 18  & 85    & 3  & 0   & 18  & 14  & 13  & 48    & +06:55 & -24 \\ \hline %
      \multicolumn{1}{|l|}{PBE}   & 15 & 13  & 10  & 10  & 20  & 68    & 6  & 16  & 5   & 6   & 27  & 60    & +02:48 & -9  \\ \hline %
      \multicolumn{1}{|l|}{NPS}   & 17 & 15  & 7   & 11  & 29  & 67    & 5  & 16  & 6   & 5   & 26  & 58    & +03:51 & -12 \\ \hline %
      \multicolumn{15}{c}{ } \\
    \end{tabular}
    \caption{Comparison of encodings in terms of SAT-solver runtime on PB15 suite. We count number of 
    benchmarks for which {\bf 4OE} showed speed-up or slow-down factor with respect to 
    different encodings, the difference in total running time of each encoding w.r.t. {\bf 4OE} and
    the difference in the number of solved instances of each encoding w.r.t. {\bf 4OE}.}
    \label{tbl:res-pb}
\end{table}

\subsection{4-Wise vs 4-Odd-Even}

Notice that in the evaluation we have omitted the 4-Wise Selection Network from Chapter \ref{ch:mw}.
It is because preliminary experiments showed that the 4-Odd-Even Selection Network is superior, and since running the {\bf PB15 suite}
is very time-consuming we decided to showcase only our best encoding. To remedy this situation we have performed similar experiment
on a different set of instances.

{\bf MSU4 suite} is a set consisting of about 14000 benchmarks, each of which contains a mix of CNF formula and multiple cardinality constraints.
This suite was created from a set of MaxSAT instances reduced from real-life problems, and then it was converted
by the implementation of {\em msu4} algorithm \cite{marques2011algorithms}.
This algorithm reduces a MaxSAT problem to a series of SAT problems with cardinality constraints.
The MaxSAT instances were taken from the Partial Max-SAT
division of the Third Max-SAT evaluation\footnote{See http://www.maxsat.udl.cat/08/index.php?disp=submitted-benchmarks}.
The time-out limit was set to 600 seconds.

The results are summarized in Figure \ref{fig:plots_msu}. We show a cactus plot for {\bf 4OE}, {\bf PCN} and {\bf 4WISE} --
the implementation of the 4-Wise Selection Network based on Algorithms \ref{net:mw_sel} and \ref{net:4mw_merge} (available in \textsc{KP-MiniCard}).
Similar as before, we extended {\bf 4WISE} by using Direct Cardinality Networks for sufficiently small values of $n$ and $k$. The graph shows a clear
difference in performance between all three encodings.

\begin{figure}[t!]
  \centering
  \includegraphics[width=1.0\textwidth]{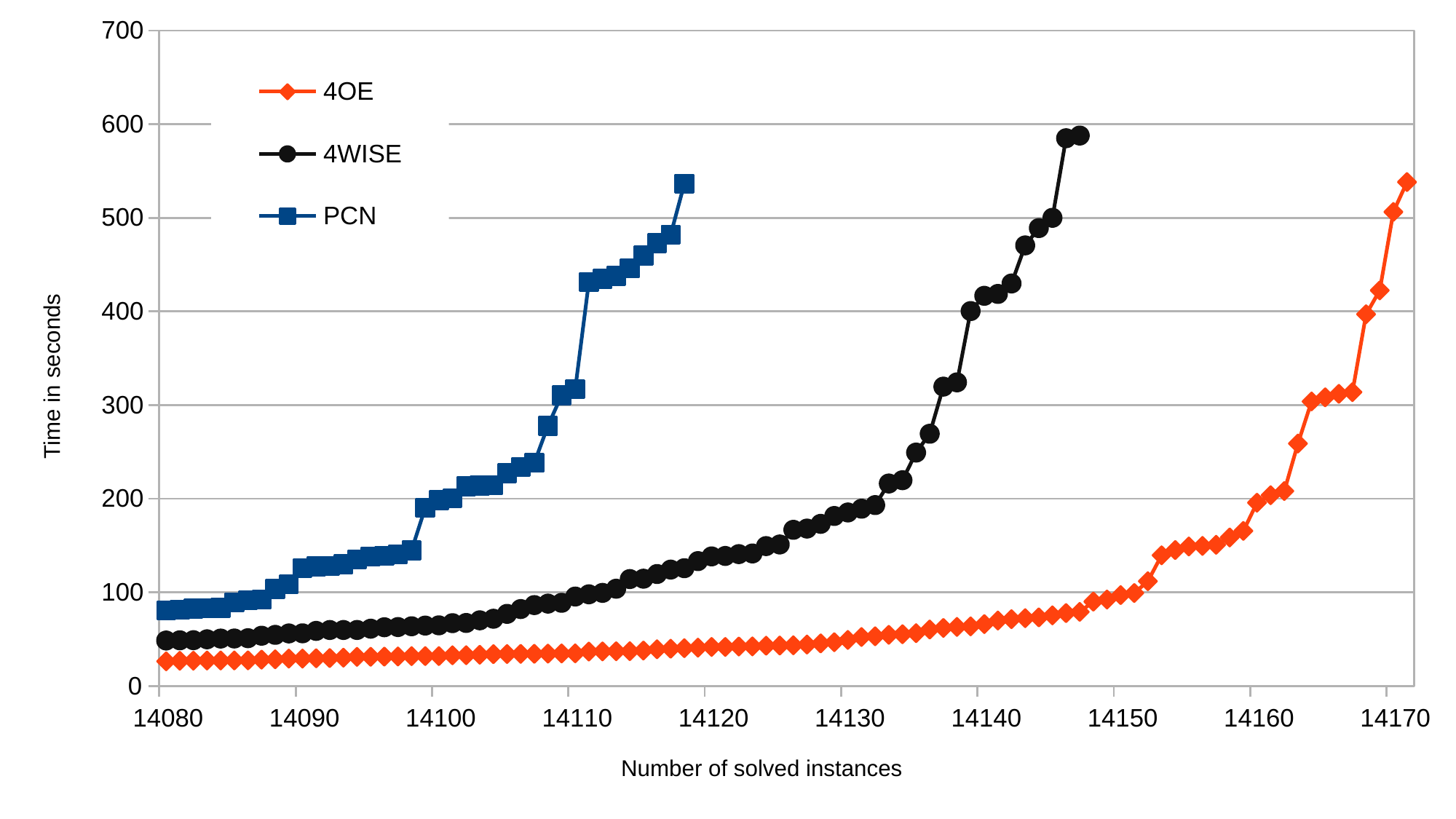}
  \caption{The number of solved instances of MSU4 suite in a given time}
  \label{fig:plots_msu}
\end{figure}

\section{m-Odd-Even Selection Network}

We show that we can generalize our algorithm further, so that it can be parametrized by any value of $m\geq2$.
The construction of the $m$-Odd-Even Selection Network is presented in Algorithm \ref{net:oe_sel_m}.

  We arrange the input sequence into $m$ columns of non-increasing sizes and we recursively
  run the selection algorithm on each column (lines 3--6), where at most $k$ items are selected
  from each column. Selected items are sorted and form prefixes of the columns and they
  are the input to the merging procedure (line 7--8). The base case, when $k=1$,
  is handled by the selector $sel^n_1$.
  
  \begin{algorithm}[t]
    \caption{$oe\_sel^n_k$}\label{net:oe_sel_m}
    \begin{algorithmic}[1]
      \Require {$\bar{x} \in \bool^n$; $n_1, \dots , n_m \in \nat$ where $n > n_1 \geq \dots \geq
        n_m$ and $\sum n_i = n$; $1 \le k \le n$}
      \Ensure{The output is top $k$ sorted and is a permutation of the inputs}
      \If {$k = 1$}
        \Return $sel^n_1(\bar{x})$
      \EndIf
      \State $\mathit{offset} = 1$
      \ForAll {$i \in \{1,\dots,m\}$}
        \State $k_i=\min(k, n_i)$
        \State $\bar{y}^i \gets oe\_sel^{n_i}_{k_i}(\bar{x}_{\mathit{offset},\dots,\mathit{offset} + n_i - 1})$
        \State $\mathit{offset} += n_i$
      \EndFor
      \State $s = \sum_{i=1}^{m} k_i$; 
         ~~$\overline{out} = \suff(k_1+1,\bar{y}^1) :: \dots :: \suff(k_m+1,\bar{y}^m)$
      \State \Return $oe\_merge^{s}_{k}(
        \tuple{\pref(k_1,\bar{y}^1),\dots,\pref(k_m,\bar{y}^m)}) :: \overline{out}$
    \end{algorithmic}
  \end{algorithm}

  \begin{theorem}\label{thm:oe_sel_m}
    Let $n,k \in \nat$, such that $k \leq n$. Then $oe\_sel^n_k$ is a $k$-selection network.
  \end{theorem}

  \begin{proof}
    Observe that $\bar{y} = \bar{y}^1 :: \dots :: \bar{y}^m$ is a
    permutation of the input sequence $\bar{x}$. We prove by induction that for each $n,k
    \in \nat$ such that $1 \leq k \leq n$ and each $\bar{x} \in \bool^n$:
    $oe\_sel^n_k(\bar{x})$ is top $k$ sorted. If $1=k \leq n$ then $oe\_sel^n_k = sel^n_1$,
    so the theorem is true. For the induction step assume that $n \ge k \ge 2$, $m \ge 2$
    and for each $(n^{*},k^{*}) \prec (n,k)$ (in lexicographical order) the theorem holds.
    We have to prove that the sequence $\bar{w} = \pref(k_1,\bar{y}^1)::\dots
    ::\pref(k_m,\bar{y}^m)$ contains $k$ largest elements from $\bar{x}$. If all 1's in
    $\bar{y}$ are in $\bar{w}$, we are done. So assume that there exists $y^i_j=1$ for
    some $1 \leq i \leq m$, $k_i < j \leq n_i$. We show that $|\bar{w}|_1 \geq k$.
    Notice that $k_i=k$, otherwise $j > k_i = n_i$ -- a contradiction. Since $|\bar{y}^i|
    = n_i \leq n_1 < n$, from the induction hypothesis we get that $\bar{y}^i$ is top $k_i$
    sorted. In consequence, $\pref(k_i,\bar{y}^i) \succeq \tuple{y^i_j}$, which implies
    that $|\pref(k_i,y^i)|_1 = k_i = k$. We conclude that $|\bar{w}|_1 \geq
    |\pref(k_i,y^i)|_1 = k$. Note also that in the case $n = k$ we have all $k_i = \min(n_i,k) 
    < k$, so the case is correctly reduced.
    
    Finally, using $oe\_merge^{s}_{k}$ the algorithm returns 
    $k$ largest elements from $\bar{x}$, which completes the proof.
  \end{proof}

\section{Summary} \label{sec:concl}

In this chapter we presented a multi-column selection networks based on the odd-even approach,
that can be used to encode cardinality constraints. We showed that the CNF
encoding of the $4$-Odd-Even Selection Network is smaller than the 2-column version. We extended the encoding by
applying Direct Cardinality Networks \cite{abio2013parametric} for sufficiently small input. The
new encoding was compared with the selected state-of-the-art encodings based
on comparator networks, adders and binary decision diagrams as well as
with two popular general constraints solvers.
The experimental evaluation shows that the new encoding yields
better speed-up and overall runtime in the SAT-solver performance.

We have also showed (here, and in the previous chapter) how to generalize the multi-column networks for any number of columns.
The conclusion is that the odd-even algorithm is much easier to implement than the pairwise
algorithm.

\chapter[Encoding Pseudo-Boolean Constraints]{Encoding Pseudo-\\ Boolean Constraints}\label{ch:pos18}

    \def\nqueenssolution{Qd4, Qe2, Qf8, Qa5, Qc1, Qg6, Qh3}
    \setchessboard{smallboard,labelleft=false,labelbottom=false,showmover=false,setpieces=\nqueenssolution}

    \begin{tikzpicture}[remember picture,overlay]
      \node[anchor=east,inner sep=0pt] at (current page text area.east|-0,3cm) {\chessboard};
    \end{tikzpicture}

To replicate the success of our algorithm from the previous chapter in the field of PB-solving, we implemented
the 4-Odd-Even Selection Network in \textsc{MiniSat+} and removed the 2-Odd-Even Sorting Network from
the original implementation \cite{minisatp}. In Chapter \ref{ch:cp18}
we have showed a top-down, divide-and-conquer algorithm for constructing the 4-Odd-Even Selection Network.
The difference in our new implementation is that we build our network in a bottom-up manner,
which results in the easier and cleaner implementation.

We apply a number of optimization techniques in our solver, some based on the work of other researchers. In particular, we use optimal base searching
algorithm based on the work of Codish et al. \cite{codish2011} and ROBDD structure \cite{abio2012} instead of BDDs for one of the encodings in
\textsc{MiniSat+}. We also substitute sequential search of minimal value of the goal function in optimization problems with binary search similarly to
Sakai and Nabeshima \cite{sakai2015}. We use \textsc{COMiniSatPS} \cite{oh2016improving} by Chanseok Oh as the underlying SAT-solver,
as it has been observed to perform better than the original \textsc{MiniSat} \cite{een2003extensible} for many instances.

We experimentally compare our solver with other state-of-the-art general constraints solvers like
\textsc{PBLib} \cite{philipp2015pblib} and \textsc{NaPS} \cite{sakai2015} to prove that our techniques 
are good in practice. There have been organized a series of Pseudo-Boolean Evaluations \cite{manquinho2006first}
which aim to assess the state-of-the-art in the field of PB-solvers. We use the competition problems
from the PB 2016 Competition as a benchmark for the solver proposed in this chapter.

\section{System Description}

\subsection{4-Way Merge Selection Network}



It has already been observed that using selection networks instead of sorting networks is more efficient
for the encoding of constraints \cite{codish2010pairwise}. This fact has been successfully used in encoding cardinality constraints,
as evidenced, for example, by the results of this thesis. We
now apply this technique to PB-constraints. Here we describe the algorithm for constructing a bottom-up version of the 4-Odd-Even Selection Network.

The procedure can be described as follows. Assume $k\leq n$ and that we have the sequence of Boolean literals $\bar{x}$ of length $n$
and we want to select $k$ largest, sorted elements, then:

\begin{itemize}
  \item If $k=0$, there is nothing to do.
  \item If $k=1$, simply select the largest element from $n$ inputs using a $1$-selector of order $n$.
  \item If $k>1$, then split the input into subsequences of either the same literals (of length at least 2) or sorted 5 singletons 
  (using a $\min(5,k)$-selector of order $5$). Next, sort subsequences by length, in a non-increasing order.
  In loop: merge each 4 (or less) consecutive subsequences into one (using 4-Odd-Even Merging Network as a sub-procedure)
  and select at most $k$ largest items until one subsequence remains.
\end{itemize}

\begin{examplebox}
\begin{example}
  See Figure \ref{fig:4oe_sel} for a schema of our selection network
  (where $n=18$ and $k=6$) which selects 6 largest elements from the input $011000010000001011$.
  In the same figure, in dashed lines, a schema of 4-Odd-Even Merger is presented
  with $s=18$, $k=6$, $|\bar{w}|=|\bar{x}|=|\bar{y}|=5$ and $|\bar{z}|=3$. First, the inputs are split into two by odd and even indices,
  and the recursive calls are made. After that, a combine operation fixes the order of elements, to output the 6 largest ones.
\end{example}
\end{examplebox}

\begin{figure}[t!]
  \centering
  \includegraphics[width=1.00\textwidth]{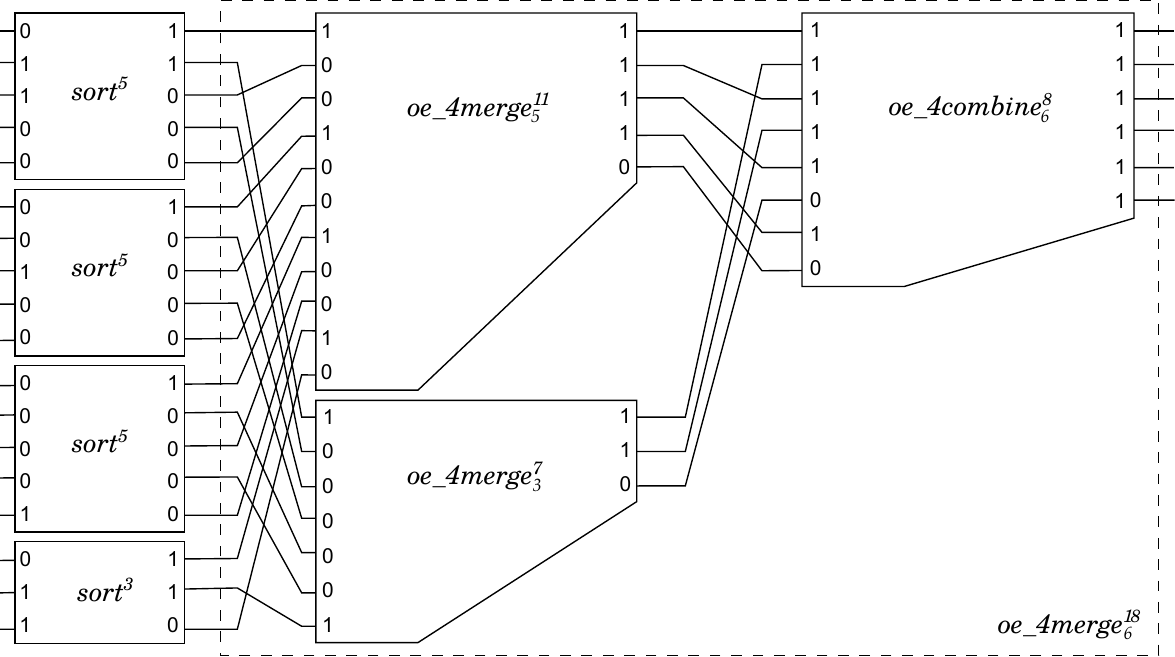}
  \caption{An example of 4-Odd-Even Selection Network, with $n=18$ and $k=6$}
  \label{fig:4oe_sel}
\end{figure}

In the following subsections we explain how we have extended the \textsc{MiniSat+} solver to make the above process (and the entire PB-solving computation) more efficient.

\subsection{Simplifying Inequality Assertions}

We use the following optimization found in \textsc{NaPS} \cite{sakai2015} for simplifying inequality assertions in a constraint.
We introduce this concept with an example. In order to assert the constraint $a_1l_1 + \dots + a_nl_n \geq k$
the encoding compares the digits of the sum of the terms on the left side of the constraint with those from $k$ (in some base $\mathbf{B}$)
from the right side. Consider the following example:

\[
  5x_1 + 7x_2 \geq 9
\]

\noindent Assume that the base is $\mathbf{B}=\tuple{2,2}$. Then $9=\tuple{1,0,2}_{\mathbf{B}}$, but if we add $7$ to both sides
of the inequality:

\[
  7 + 5x_1 + 7x_2 \geq 16
\]

\noindent then those constraints are obviously equivalent and $16=\tuple{0,0,4}_{\mathbf{B}}$. Now in order to assert the inequality we only
need to assert a single output variable of the encoding of the sum of LHS coefficients (using a singleton clause). 
The constant 7 on the LHS has a very small impact on the size of LHS encoding. This simplification allows for the reduction 
of the number of clauses in the resulting CNF encoding, as well as allows better propagation.

\subsection{Optimal Base Problem}

We have mentioned that \textsc{MiniSat+} searches for a mixed radix base such that the sum of all the digits of
the coefficients written in that base, is as small as possible. In their paper \cite{minisatp} authors mention in the footnote that:

\begin{quote}
  {\em The best candidate is found by a brute-force search trying all prime numbers $< 20$.
  This is an ad-hoc solution that should be improved in the future. Finding
  the optimal base is a challenging optimization problem in its own right.}
\end{quote}

Codish et al. \cite{codish2011} present an algorithm which scales to find an optimal
base consisting of elements with values up to $1,000,000$ and they consider several measures of optimality for finding the base.
They show experimentally that in many cases finding a better base leads also to better SAT solving time.
We use their algorithm in our solver, but we restrict the domain of the base to prime numbers less than $50$,
as preliminary experiments show that numbers in the base are usually small.

\subsection{Minimization Strategy}

The key to efficiently solve Pseudo Boolean optimization problems is the repeated use of a SAT-solver. Assume we have a minimization problem
with an objective function $f(x)$. First, without considering $f$, we run the solver on a set of constraints to get an initial solution $f(x_0)=k$.
Then we add the constraint $f(x) < k$ and run the solver again. If the problem is unsatisfiable, $k$ is the optimal solution. If not, the process is repeated with
the new (smaller) candidate solution $k'$. The minimization strategy is about the choice of $k'$. If we choose $k'$ as reported by the SAT-solver,
then we are using the so called {\em sequential} strategy -- this is implemented in \textsc{MiniSat+}.

Sakai and Nabeshima \cite{sakai2015} propose the {\em binary} strategy for the choice of new $k'$. Let $k$ be the best known goal value
and $l$ be the greatest known lower bound, which is initially the sum of negative coefficients of $f$. After each iteration, new constraint
$p \Rightarrow f(x) < \floor{(k(q-1)+l)/q}$ is added, where $p$ is a fresh variable (assumption) and $q$ is a constant (we set $q=3$ as
default value). Depending on the new SAT-solver answer, $\floor{(k(q-1)+l)/q}$ becomes the new upper or lower bound (in this case $p$ is set to $0$), and the process begin anew.

In our implementation we use binary strategy until the difference between the upper and lower bounds of the goal value is less than $96$,
then we switch to the sequential strategy. We do this in order to avoid a situation when a lot of computation is needed when searching
for UNSAT answers, which could arise
if only binary strategy was used. This was also observed in \cite{sakai2015} and the authors have used it as a default strategy.
Moreover, they propose to alternate between binary and sequential strategy depending
on the SAT-solver answer in a given iteration.

\subsection{ROBDDs Instead of BDDs}

One of the encodings of \textsc{MiniSat+} is based on {\em Binary Decision Diagrams} (BDDs). We have improved
the implementation of this encoding by using the more recent {\em Reduced Ordered BDD} (ROBDD) construction \cite{abio2012}. Now ROBDD is used
to create a DAG representation of a constraint. One of the advantages of ROBBDs is that we can reuse nodes in the ROBDD structure,
which results in a smaller encoding. This concept is illustrated in Figure \ref{fig:bdd},
which shows an example from \cite{abio2012} of BDD and ROBDD for the PB-constraint $2x_1 + 3x_2 + 5x_3 \leq 6$.
Two reductions are applied (until fix-point) for obtaining ROBDD: removing nodes with identical children and merging isomorphic subtrees.
This was already explained in Section \ref{sec:pb}. See \cite{abio2012} for a more detailed example.

\begin{figure}[t!]
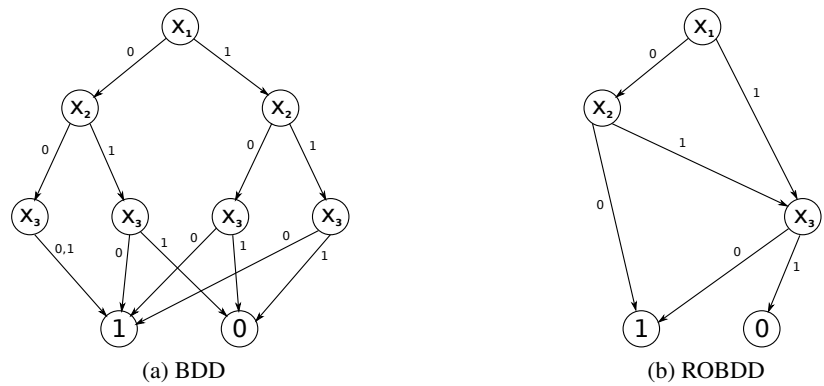

  \centering
  \subfloat[BDD\label{fig:bdd_a}]{%
    \includegraphics[height=0.3\textwidth]{bdd.pdf}
  }
  ~~~~~~~~~~~~~~~~~~~~~~~~~~~~~~
  \subfloat[ROBDD\label{fig:bdd_b}]{%
    \includegraphics[height=0.3\textwidth]{robdd.pdf}
  }
  \caption{Construction of a BDD (left) and a ROBDD (right) for $2x_1 + 3x_2 + 5x_3 \leq 6$}
  \label{fig:bdd}
\end{figure}

\subsection{Merging Carry Bits}

In the construction of interconnected sorters in \textsc{MiniSat+} carry bits from one sorter are being fed to the next sorter. In the footnote
in \cite{minisatp} it is mentioned that:

\begin{quote}
  {\em Implementation note: The sorter can be improved by using the fact that the carries are already sorted.}
\end{quote}

\noindent We go with this suggestion and use our merging network to merge the carry bits with a sorted digits representation instead
of simply forwarding the carry bits to the inputs of the next selection network.

\subsection{SAT-solver}

The underlying SAT-solver of \textsc{MiniSat+} is \textsc{MiniSat} \cite{een2003extensible} created by Niklas E\'en and Niklas S\"orensson.
It has served as an extension to many new solvers, but it is now quite outdated. We have integrated a solver called \textsc{COMiniSatPS}
by Chanseok Oh \cite{oh2016improving}, which have collectively won six medals in SAT Competition 2014 and Configurable SAT Solver Challenge 2014.
Moreover, the modification of this solver called \textsc{MapleCOMSPS} won the Main Track category of
SAT Competition 2016\footnote{See http://baldur.iti.kit.edu/sat-competition-2016/}.

\section{Experimental Evaluation}

Our extension of \textsc{MiniSat+} based on the features explained in this chapter,
which we call \textsc{KP-MiniSat+}, is available online\footnote{See https://github.com/karpiu/kp-minisatp}.
It should be linked with a slightly modified \textsc{COMiniSatPS}, also available online\footnote{See https://github.com/marekpiotrow/cominisatps}.
Detailed results of the experimental evaluation are also available online\footnote{See http://www.ii.uni.wroc.pl/\%7ekarp/pos/2018.html}.

The set of instances we use is from the
Pseudo-Boolean Competition 2016\footnote{See {\texttt http://www.cril.univ-artois.fr/PB16/}}. We use
instances with linear, Pseudo-Boolean constraints that encode either decision or optimization problems.
To this end, three categories from the competition have been selected:

\begin{itemize}
  \item {\bf DEC-SMALLINT-LIN} - 1783 instances of decision problems with small coefficients in the constraints
         (no constraint with sum of coefficients greater than $2^{20}$).
         No objective function to optimize. The solver must simply find a solution.
  \item {\bf OPT-BIGINT-LIN} - 1109 instances of optimization problems with big coefficients in the constraints
         (at least one constraint with a sum of coefficients greater than $2^{20}$).
         An objective function is present. The solver must find a solution with the best possible value of the
         objective function.
  \item {\bf OPT-SMALLINT-LIN} - 1600 instances of optimization problems. Like
         {OPT-BIGINT-LIN} but with small coefficients (as in {DEC-SMALLINT-LIN}) in the constraints.
\end{itemize}

We compare our solver (abbreviated to {\bf KP-MS+}) with two state-of-the-art general purpose constraint solvers. First is
the \textsc{pbSolver} from \textsc{PBLib} ver. 1.2.1, by Tobias Philipp and Peter Steinke \cite{philipp2015pblib} (abbreviated to {\bf PBLib} in the results).
This solver implements a plethora of encodings for three types of constraints: at-most-one, at-most-k (cardinality constraints) and Pseudo-Boolean constraints.
The \textsc{PBLib} automatically normalizes the input constraints and decides which encoder provides the most effective translation.
We have launched the program \texttt{./BasicPBSolver/pbsolver} of \textsc{PBLib} on each instance with the default parameters.

The second solver is \textsc{NaPS} ver. 1.02b by Masahiko Sakai and Hidetomo Nabeshima \cite{sakai2015} which implements
improved ROBDD structure for encoding constraints in band form, as well as other optimizations.
This solver is also built on the top of \textsc{MiniSat+}. \textsc{NaPS} won two of the optimization categories in the Pseudo-Boolean Competition 2016:
OPT-BIGINT-LIN and OPT-SMALLINT-LIN. We have launched the main program of \textsc{NaPS} on each instance, with parameters \texttt{-a -s -nm}.

We also compare our solver with the original \textsc{MiniSat+} in two different versions, one using the original \textsc{MiniSat} SAT-solver
and the other using the \textsc{COMiniSatPS} (the same as used by us). We label these {\bf MS+} and {\bf MS+COM} in the results. We prepared results for
{\bf MS+COM} in order to show that the advantage of using our solver does not come simply from changing the underlying SAT-solver.

We have launched our solver on each instance, with parameters \texttt{-a -s -cs -nm}, where \texttt{-cs}  means that in experiments
the solver used just one encoding technique, the 4-Odd-Even Selection Networks combined with a direct encoding of small sub-networks.

All experiments were carried out on the machines with Intel(R) Core(TM) i7-2600 CPU @ 3.40GHz. Timeout limit is set to 1800 seconds and
memory limit is 15 GB, which are enforced with the following commands:  \texttt{ulimit -Sv 15000000} and \texttt{timeout -k 20 1809
  <solver> <parameters> <instance>}.

We would like to note that we also wanted to include in this evaluation the winner of DEC-SMALLINT-LIN category, which is the solver based on the
{\em cutting planes} technique, but we refrained from that for the following reason. We have not found the source code of this solver
and the only working version found in the author's website\footnote{See {\texttt http://www.csc.kth.se/\%7eelffers/}} is a binary file 
without any documentation. Because of this, we were unable to get any meaningful results of running aforementioned program on optimization instances.

\renewcommand{\arraystretch}{1.0}
\begin{table}[t!]\setlength{\tabcolsep}{6pt}
    \centering
    \begin{tabular}{ c || c | c | c | c | c | c | c | c | c | }
      \multirow{2}{*}{solver} & \multicolumn{3}{c|}{DEC-SMALLINT-LIN} & \multicolumn{3}{c|}{OPT-BIGINT-LIN} & \multicolumn{3}{c|}{OPT-SMALLINT-LIN} \\\cline{2-10}
                              & Sat & UnSat & cpu            & Opt & UnSat & cpu            & Opt & UnSat & cpu        \\ \hline\hline
      KP-MS+                  & 432 & 1049  & 647041         & 359 & 72    & 1135925        & 808 & 86    & 1289042    \\ \hline
      NaPS                    & 348 & 1035  & 816725         & 314 & 69    & 1314536        & 799 & 81    & 1330536    \\ \hline
      PBLib                   & 349 & 922   & 1104508        & --  & --    & --             & 691 & 56    & 1611247    \\ \hline
      MS+                     & 395 & 951   & 895774         & 149 & 71    & 1647958        & 715 & 73    & 1515166    \\ \hline
      MS+COM                  & 428 & 1027  & 703269         & 174 & 71    & 1609433        & 734 & 71    & 1491269    \\ \hline
      \multicolumn{10}{c}{ } \\
    \end{tabular}
    \caption{Number of solved instances of PB-competition benchmarks.}
    \label{tbl:solved}
\end{table}

In Table \ref{tbl:solved} we present the number of solved instanced for each competition category.
{\bf Sat}, {\bf UnSat} and {\bf Opt} have the usual meaning, while {\bf cpu} is the total solving time of the solver over all instances of
a given category. Results clearly favor our solver. We observe significant improvement in the number
of solved instances in comparison to \textsc{NaPS} in categories DEC-SMALLINT-LIN and OPT-BIGINT-LIN.
The difference in the number of solved instances in the OPT-SMALLINT-LIN category is not so significant. Solver \textsc{PBLib} had the worst performance
in this evaluation. Notice that the results of \textsc{PBLib} for OPT-BIGINT-LIN division is not available. This is because \textsc{PBLib}
is using 64-bit integers in calculations, thus could not be launched with all OPT-BIGINT-LIN instances.

Figures \ref{fig:plot1}, \ref{fig:plot2} and \ref{fig:plot3} show cactus plots of the results, which indicate the
number of solved instances within the time. We see clear advantage of our solver over the competition in the 
DEC-SMALLINT-LIN and OPT-BIGINT-LIN categories.

\begin{figure}[t!]
  \centering
  \includegraphics[width=1.00\textwidth]{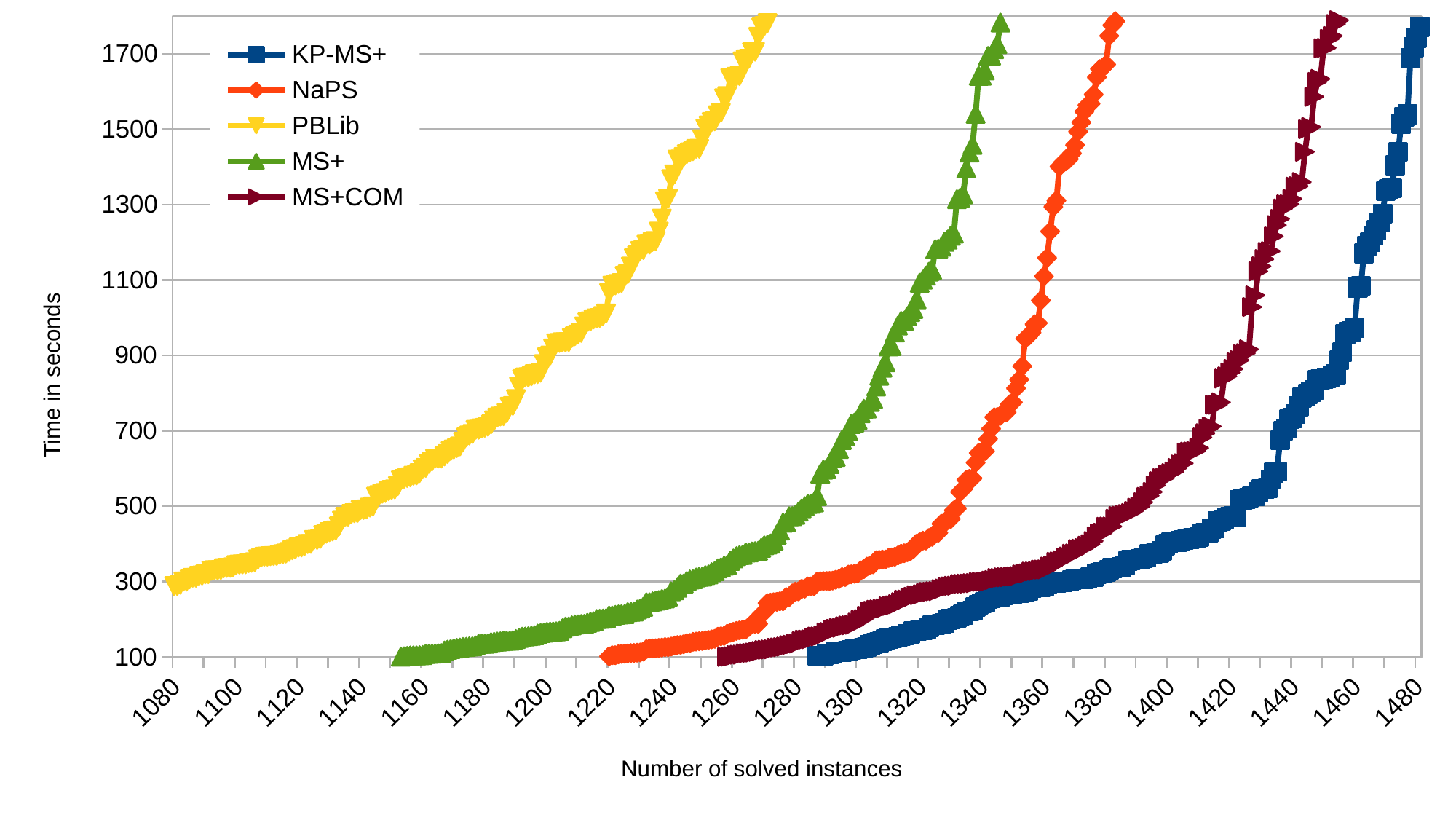}
  \caption{Cactus plot for DEC-SMALLINT-LIN division}
  \label{fig:plot1}
\end{figure}

\begin{figure}[t!]
  \centering
  \includegraphics[width=1.00\textwidth]{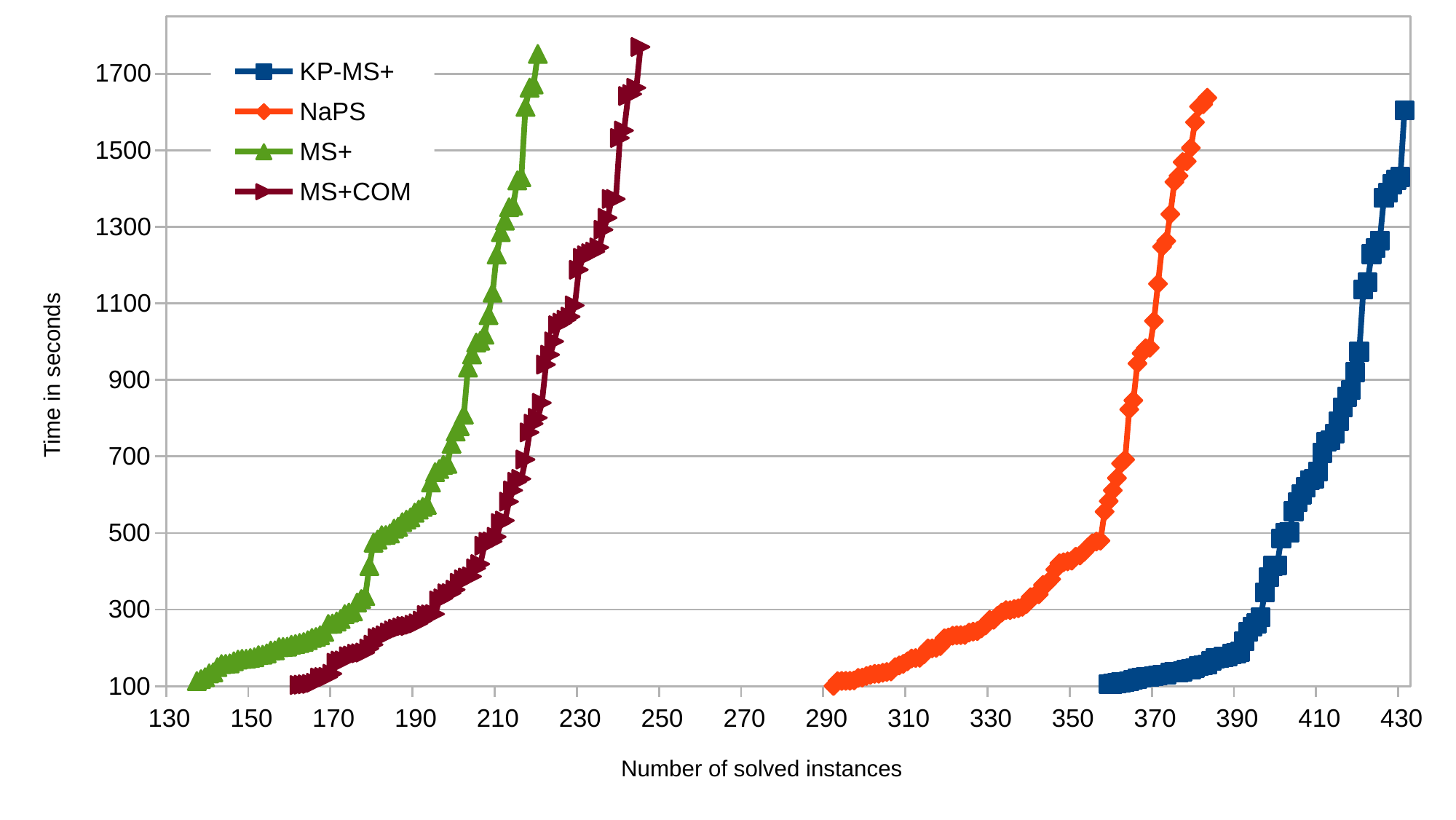}
  \caption{Cactus plot for OPT-BIGINT-LIN division}
  \label{fig:plot2}
\end{figure}

\begin{figure}[t!]
  \centering
  \includegraphics[width=1.00\textwidth]{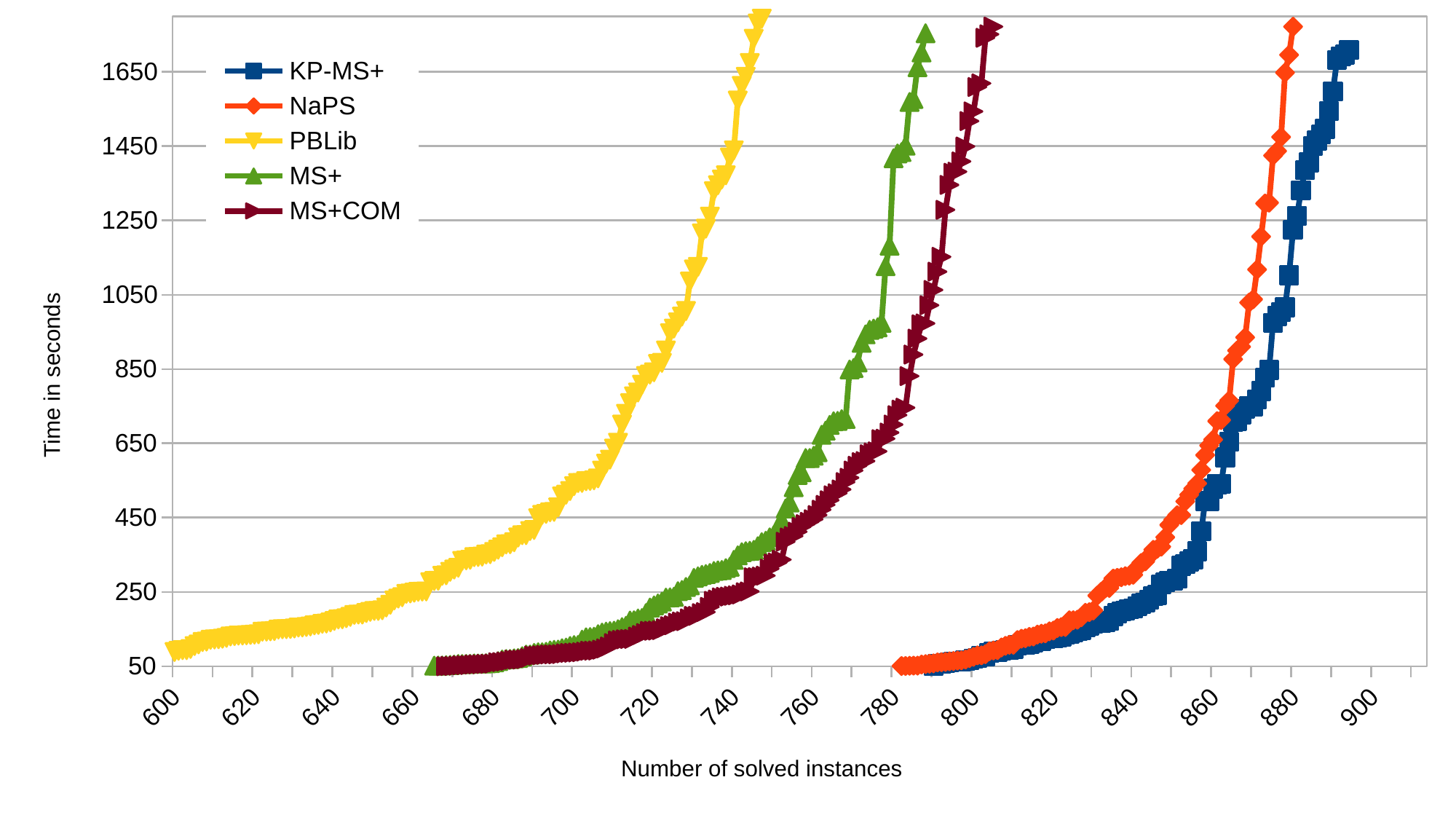}
  \caption{Cactus plot for OPT-SMALLINT-LIN division}
  \label{fig:plot3}
\end{figure}

\section{Summary}

In this chapter we proposed a new method of encoding PB-constraints into SAT based
on sorters. We have extended the \textsc{MiniSat+} with the 4-way merge selection
algorithm and showed that this method is competitive to other state-of-the-art
solutions. Our algorithm is short and easy to implement. Our implementation is modular,
therefore it can be easily extracted and applied in other solvers.

\chapter[Final Remarks]{Final Remarks}\label{ch:c}

    \def\nqueenssolution{Qd4, Qe2, Qf8, Qa5, Qc1, Qg6, Qh3, Qb7}
    \setchessboard{smallboard,labelleft=false,labelbottom=false,showmover=false,setpieces=\nqueenssolution}

    \begin{tikzpicture}[remember picture,overlay]
      \node[anchor=east,inner sep=0pt] at (current page text area.east|-0,3cm) {\chessboard};
    \end{tikzpicture}

This chapter concludes the thesis, as well as shows the solution to the 8-Queens Puzzle which
we have been building throughout the dissertation (it is the only solution -- modulo symmetry --
with the property that no three queens are in a straight line!). Here we would like to give some
final remarks and possibilities for future work.

We have shown new classes of networks that have efficient translations to CNFs.
Encodings -- of both cardinality and Pseudo-Boolean constraints -- based on our 4-way Merge Selection Network
are very competitive, as evidenced by the experimental evaluation presented here.
Our encodings easily compete with other state-of-the-art encodings, even with the winners
of recent competitions. We also prove the arc-consistency property for all encodings based
on the standard encoding of generalized selection networks. This captures all the new encodings
presented in this thesis, as well as all past (and future) encodings (based on sorting/selection networks).

Possibly the biggest mystery we can leave the reader with is: why encoding cardinality constraints using comparator
networks is so efficient? It is in fact a very fundamental question. We can only see the empirical evidence as shown in this thesis,
as well as other papers. We can compare two networks using various measures like the number of
comparators, depth, the number of variables and clauses the encoding generates, or variables to
clauses ratio. But all of those measures do not seem to give a conclusive answer to our question, especially
if we want to collate our encodings with the ones not based on comparator networks. Arc-consistency
property is important, but even with it we still cannot decisively distinguish between the top used encodings.
In fact, being arc-consistent is a must if an encoding is to be competitive. It looks like another, more complex propagation
property is needed in order to theoretically prove superiority of encodings based on comparator networks. Such property
may have not been discovered yet.

We see, by the new constructions and the old ones referenced in the first part of this thesis,
that comparator networks considered in the field
are based exclusively on the odd-even or pairwise approach.
Nevertheless, there are many other sorting networks proposed in the literature that could be used.
It would be informative to see an empirical study on much wider collection of encodings based on networks not yet considered in practice.
So a survey on this topic is a niche ready to be filled.

{\printbibliography}

\end{document}